\documentclass[11pt, a4paper]{article}
\usepackage[T1]{fontenc}
\usepackage[utf8]{inputenc}
\usepackage{newpxtext}
\usepackage[sc,osf]{mathpazo}
\usepackage{a4wide}
\usepackage{caption}
\usepackage{amsmath}
\allowdisplaybreaks
\usepackage{amsmath, amssymb, bbm, bm, amsfonts, amsthm, mathrsfs, mathtools}
\usepackage[dvipsnames]{xcolor}
\definecolor{oxf}{RGB}{0, 33, 71}
\definecolor{newblue}{HTML}{3A59D1}
\definecolor{newred}{HTML}{FF3737}
\usepackage{tikz-cd}
\usepackage{tikz}
 \usetikzlibrary{arrows.meta}
\usepackage{array}
\usepackage[colorlinks=true, 
            linkcolor=oxf, 
            citecolor=oxf, 
            urlcolor=oxf]{hyperref}
\usepackage{cleveref}

\usepackage[numbers,sort&compress]{natbib}

\newtheoremstyle{mystyle}
  {8pt}{8pt}{}{}{\bfseries}{.}{.5em}{}
\theoremstyle{mystyle}
\newtheorem{theorem}{Theorem}[section]
\newtheorem{result}{Result}
\newtheorem{question}{Question}
\newtheorem{assumption}{Assumption}
\newtheorem*{theorem*}{Theorem}

\newtheorem{definition}{Definition}[section]
\newtheorem*{definition*}{Definition}

\newtheorem{proposition}{Proposition}[section]
\newtheorem*{proposition*}{Proposition}
\newtheorem{corollary}{Corollary}[section]
\newtheorem{construction}{Construction}

\newtheorem{example}{Example}[section]
\newtheorem{remark}{Remark}[section]

\crefname{assumption}{assumption}{assumptions}
\Crefname{assumption}{Assumption}{Assumptions}
\usepackage{enumitem}
\setlist[itemize]{itemsep=0.4em, topsep=0.4em, parsep=2pt}
\setlist[enumerate]{itemsep=0.4em, topsep=0.4em, parsep=2pt}
\usetikzlibrary{decorations.pathreplacing,calligraphy}
\usetikzlibrary{arrows.meta,decorations.markings}
\usepackage{float}
\tikzset{->-/.style={decoration={
  markings,
  mark=at position .5 with {\arrow[xshift=0.5ex]{>}}},postaction={decorate}}}
    \definecolor{myblue}{RGB}{29, 66, 166}
    \definecolor{myorange}{RGB}{230,120,20}
\newcommand{\arrow}[3]{
  \begin{scope}[shift={#1}, rotate=#3]
    \fill
      (0,0)
      -- ({-0.5*#2},{-#2})
      -- ({0.5*#2},{-#2})
      -- cycle;
  \end{scope}
}
\definecolor{zima}{RGB}{91, 194, 231}

\renewcommand{\>}{\rangle}

\newcommand{\Rep}{\mathbf{Rep}}
\newcommand{\To}{\Rightarrow}

\newcommand{\Hom}{\text{Hom}}
\newcommand{\vc}{\mathbf{Vec}}

\newcommand{\id}{\text{id}}

\usepackage{euscript}

\newcommand{\bb}{\mathbb}
\newcommand{\bbm}{\mathbbm}
\renewcommand{\sf}{\mathsf}
\newcommand{\cc}[1]{\mathcal{#1}}

\newcommand{\cZ}{\cc{Z}}
\newcommand{\bC}{\mathbb{C}}

\newcommand{\Z}{\mathbb{Z}}

\newcommand{\C}{\cc{C}}
\newcommand{\D}{\cc{D}}
\newcommand{\M}{\cc{M}}

\newcommand{\End}{\mathrm{End}}

\newcommand{\ut}{\bbm{1}}

\newcommand{\TY}{\mathrm{TY}}
\newcommand{\hide}[1]{}
\renewcommand{\H}{\cc{H}}
\newcommand{\ldim}{\mathrm{ldim}}
\newcommand{\rdim}{\mathrm{rdim}}
\newcommand{\qdim}{\mathrm{qdim_{lat}}}
\newcommand{\ind}{\mathrm{ind}}
\newcommand{\irr}{\mathrm{Irr}}

\usepackage{amssymb}
\usepackage{amsfonts}
\usepackage{epstopdf}
\usepackage{dcolumn}
\usepackage{amsmath}
\usepackage{latexsym,bm}
\usepackage{amsthm}
\usepackage{slashed}
\usepackage{float}
\usepackage{url}
\usepackage{rotating}
\usepackage[normalem]{ulem}
\usepackage{faktor}
\usepackage{tikz-cd}
\usepackage{tensor}
\usepackage{array}
\usepackage{cleveref}

\usepackage{tabularx}
\usepackage{makecell}

\numberwithin{equation}{section}

\newcommand{\be}{\begin{equation}}
\newcommand{\ee}{\end{equation}}
\newcommand{\ba}{\begin{aligned}}
\newcommand{\ea}{\end{aligned}}

\newcommand{\Dlat}{\mathsf{D}}
 \definecolor{navy}{RGB}{29, 66, 166}

\newcommand{\cC}{\mathcal{C}}

\newcommand{\cH}{\mathcal{H}}

\renewcommand{\Vec}{\mathsf{Vec}}

\usepackage[status=draft,inline,nomargin]{fixme}
\fxusetheme{color}
\FXRegisterAuthor{ki}{aki}{KI}
\begin{document}
\title{
Non-Invertible Symmetries on Tensor-Product Hilbert Spaces  and  Quantum Cellular Automata
 }
\author{%
Rui Wen$^{1}$,
Kansei Inamura$^{1,2}$,
Sakura Sch\"afer-Nameki$^{1}$\\[0.4em]
{\small $^{1}$Mathematical Institute, University of Oxford}\\
{\small Andrew Wiles Building, Woodstock Road, Oxford, OX2 6GG, UK}\\
{\small $^{2}$Rudolf Peierls Centre for Theoretical Physics, University of Oxford}\\
{\small Parks Road, Oxford, OX1 3PU, UK}
}
\maketitle
\begin{abstract}
\noindent
We investigate realizations of $(1+1)$-dimensional fusion category symmetries on tensor-product Hilbert spaces, allowing for mixing with quantum cellular automata (QCAs). It was argued recently that any such realizable symmetry must be weakly integral. We develop a systematic analysis of QCA-refined realizations of fusion categories and prove two statements. First, we show that, under certain physical assumptions on defects, any QCA-refined realization has QCA and symmetry-operator indices determined by the categorical data, 
up to the freedom of redefining the symmetry operators. Second, we construct  a lattice model that provides a QCA-refined  realization for any weakly integral fusion category symmetry on a tensor product Hilbert space. We also compute indices of the QCAs in our lattice model and show agreement with the first result. As an application of the general construction, we give an explicit QCA-refined realization of general Tambara-Yamagami categorical symmetries.
\end{abstract}

\tableofcontents

\section{Introduction and Summary}

The modern understanding of symmetries in quantum
many-body systems and quantum field theory has expanded well beyond the traditional setting of groups. In the continuum, it is expected that a finite generalized internal symmetry in spacetime dimension $(1+1)$D is described by a unitary fusion category \cite{Frohlich:2006ch, Bhardwaj:2017xup}. The extension to higher-dimensions, to so-called non-invertible or higher-fusion categories (see \cite{Schafer-Nameki:2023jdn, Shao:2023gho, Bhardwaj:2023kri} for reviews), triggered a renewed interest in such symmetries in 1+1D. 
This has led to the construction of new gapped phases and phase-transitions with fusion categorical symmetries using various continuum field theory and categorical methods
\cite{Thorngren:2019iar, Chang:2018iay, Huang2021TFTNonInvertible,Inamura2022LatticeFusionCategory, 
KLWZZ20a,KLWZZ20b,JiWen21,ChatWen23,ChatWen22,CJW25,BBPS24Landau,BBPS25Gapped, Bhardwaj:2024wlr,Warman:2024lir, WYP24, HC25,WP25,BBPS25Club,BPSW25,Bottini:2025hri,PR26,Huang25QC,BIT25}. 

On the other hand lattice model realizations of categorical symmetries are more subtle. Although applicable to any fusion category, the anyon chain in $(1+1)$D realizes these symmetries at the cost of not necessarily having a tensor-product Hilbert space structure 
\cite{Feiguin:2006ydp,
Trebst:2008fibonacci,
Gils:2008collective,
Gils:2009topology,
Pfeifer:2010translation,
Ardonne:2010yanglee,
Gils:2013spin1, 
Aasen:2016dop,
Buican:2017anyonic, 
Sinha:2023hum, 
Bhardwaj:2024kvy}. A related approach using MPOs/tensor networks also requires relinquishing the tensor product structure on the Hilbert space \cite{Lootens:2021tet, Molnar:2022nmh}.  
A natural question is then which categorical symmetries can be realized on a tensor-product Hilbert space\footnote{We always assume the on-site Hilbert space is finite dimensional.}. There are two separate cases to consider. First is when the symmetry is realized strictly on the Hilbert space. This means the symmetry operators form a representation of the fusion rules. It has been proven recently in \cite{Evans:2025msy} that a fusion categorical symmetry can be realized strictly on a tensor-product Hilbert space if and only if it is \textbf{integral}. A fusion category is called integral if all objects have integral quantum dimension. 
The second case, which is the focus of this paper, is when the symmetry is realized by mixing with quantum cellular automata (QCAs). In this paper, QCAs mean unitary operators that map local operators to nearby local operators. Prototypical examples of QCAs include lattice translations and finite-depth quantum circuits. See \cite{Farrelly:2019zds, Arrighi:2019uor} for reviews on this subject.

It has been observed in many models that it is possible to realize non-integral fusion categorical symmetries on a tensor-product Hilbert space if we allow mixing with QCAs \cite{Seiberg:2024gek, Seiberg:2023cdc, Lu:2026rhb,Cao:2024qjj,Seifnashri:2024dsd,Zhang:2020pco}. A basic example is the Ising fusion category. Specifically, on a spin chain with on-site Hilbert space dimension 2, it is possible to construct symmetry operators $\mathsf{D}_f, \mathsf{D}_\sigma$ that satisfy the following product rule \cite{Seiberg:2023cdc, Seiberg:2024gek}
\begin{equation}
    \mathsf{D}_f\mathsf{D}_f=1,~ \mathsf{D}_f\mathsf{D}_\sigma=\mathsf{D}_\sigma \mathsf{D}_f=\mathsf{D}_\sigma,~ \mathsf{D}_\sigma \mathsf{D}_\sigma=T(1+\mathsf{D}_f)~\label{eq_Ising_lattice},
\end{equation}
where $T$ is lattice translation by a site. We recover the Ising category fusion rules by setting $T$ to identity. In the IR of any system with the above symmetry, it is expected that lattice translation will flow to the trivial operator and the symmetry reduces to a $\Z_2$-Tambara Yamagami category (i.e.  $\TY(\Z_2, +)$ or $\TY(\Z_2, -)$ \cite{Seiberg:2024gek}). 
The appearance of lattice translation in the product rule~\Cref{eq_Ising_lattice} is what we mean by mixing with QCAs. The lattice translation operator $T$ is a non-trivial QCA.  Other examples of this type were studied in~\cite{Lu:2026rhb}. In general, a fusion category $\C$ is realized on a tensor-product Hilbert space up to QCAs if the symmetry operators $\mathsf{D}_x, x\in \C$ satisfy a modified product rule
\begin{equation}\label{eq_fusion_lattice}
    \mathsf{D}_x\mathsf{D}_y=\sum_{z\in \C} U^{z}_{xy} \mathsf{D}_z \,,
\end{equation}
where $U^z_{xy}$ are QCAs, such that the fusion rules of $\C$ is recovered by setting all QCAs to identity. We will call a realization of a fusion categorical symmetry in the form of \Cref{eq_fusion_lattice} a \textbf{QCA-refined realization}, and \Cref{eq_fusion_lattice}  the \textbf{QCA-refined fusion rules}.

The weak integrality conjecture \cite{Lu:2026rhb, Inamura:2026hif,tantivasadakarn2025noninvertible} states that a fusion categorical symmetry admitting a QCA-refined realization on a tensor-product Hilbert space must be \textbf{weakly integral}. A fusion category is called weakly integral if its total quantum dimension is an integer, or equivalently, if the square of the quantum dimension of any simple object is an integer. Evidence for this conjecture is emerging, including a recent conditional proof in~\cite{Inamura:2026hif} and a proof for the  case of duality operators in~\cite{Jones:2026dcb}. See also~\cite{seifnashri2026noninvertibleqca_talk} for a related discussion.

In this paper we are concerned with the following two questions.
\begin{question}
     Assuming a fusion categorical symmetry $\C$ is realized on a tensor-product Hilbert space as in~\Cref{eq_fusion_lattice}, what are the constraints on the QCAs $U^z_{xy}$ imposed by the categorical data of $\C$?
\end{question}
\begin{question}
    The converse to the weak integrality conjecture: Is it true that every weakly integral fusion category admits a QCA-refined realization on a tensor-product Hilbert space?
\end{question}
Our main results are complete answers to these two questions.

\begin{result}\label{result_1}
    Under certain physical assumptions about topological defects of symmetries, we prove the following: If a fusion categorical symmetry $\C$ admits a QCA-refined realization as in~\Cref{eq_fusion_lattice}, then the Gross-Nesme-Vogts-Werner (GNVW) indices \cite{Gross:2011yvb} of the QCAs $U^z_{xy}$ are uniquely determined by $\C$ up to redefinition of symmetry operators.
\end{result} 
The assumptions are detailed in \Cref{sec_asummptions}.
Here, by redefinition of symmetry operators, we allow for adding ancilla and redefining $\sf{D}_x\mapsto \alpha_x\sf{D}_x $ where $\alpha_x$ is a QCA. Two QCA-refined realizations of the same fusion category related by such a redefinition should be considered as being equivalent, similar to equivalence between projective representations of groups. 
\begin{result}\label{result_2}
We construct a lattice model that realizes any weakly integral fusion categorical symmetry on a tensor-product Hilbert space, up to QCAs. Therefore we have the following. 
\begin{theorem}
    Any weakly integral fusion categorical symmetry admits a QCA-refined realization on a tensor-product Hilbert space.
\end{theorem}
     
\end{result}

We will compute the GNVW indices of the QCAs $U^z_{xy}$ in our lattice model and confirm that they agree with the prediction of our \Cref{result_1}. We will also compute the generalized GNVW indices for the symmetry operators $\sf{D}_x$ following the definition of the generalized GNVW index in \cite{Inamura:2026hif}. These calculations provide a consistency check on the physical assumptions needed for \Cref{result_1}. Applying the general construction, we obtain an explicit tensor-product realization of general Tambara-Yamagami fusion categories up to translation, which has not been achieved in existing literature.

We stress that the weak integrality conjecture, in its full general form, remains open. The proof given in \cite{Inamura:2026hif}, which we will review in \Cref{sec_Necessity}, relies on certain physically well-motivated assumptions that should ultimately be derived from first principles. One possible route towards a proof is to justify the physical assumptions, starting with a proper definition of the symmetry operator. It may be that the upcoming work \cite{seifnashri2026noninvertibleqca} will elucidate this point. This combined with the results of the present paper could then provide us an "if and only if" statement about realizability of non-invertible symmetries on tensor-product Hilbert spaces up to QCAs. 




\paragraph*{Plan of the Paper.}
 In \Cref{sec_result_1} we establish our \Cref{result_1}. We start in \Cref{sec_asummptions} by reviewing the physical assumptions about defect Hilbert spaces and indices of non-invertible symmetries, following \cite{Inamura:2026hif}. In \Cref{sec_Necessity} we review the conditional proof of the weak integrality conjecture given in \cite{Inamura:2026hif}. We will then discuss the structure of weakly integral fusion categories in \Cref{sec_weakly_integral} and prove our \Cref{result_1} in \Cref{sec_result_1_proof}. In \Cref{sec_result_2} we establish our \Cref{result_2}: We give a lattice model that realizes any weakly integral fusion category up to QCAs. In \Cref{sec_index_computation} we calculate the indices of symmetry operators and QCAs in our lattice model and show that they agree with the prediction of~\Cref{result_1}. In \Cref{sec_relation_to_gen_translation} we discuss relation to the QCA index defined in \cite{Jones:2023imy} via operator algebras. In \Cref{sec_eg_TY} we apply our general construction to yield an explicit QCA-refined tensor-product realization of general Tambara-Yamagami categories. An appendix contains some fusion diagram conventions Appendix \ref{app:Conventions}.

\section{Constraints on Indices in QCA-Refined Realization}\label{sec_result_1} 
In this section we establish our~\Cref{result_1}. We begin with a review of physical assumptions about defect Hilbert spaces in~\Cref{sec_asummptions}. In~\Cref{sec_Necessity}, we review the physical proof of the  weak integrality conjecture given in~\cite{Inamura:2026hif}. In~\Cref{sec_weakly_integral}, we will discuss the mathematical structure of weakly integral fusion categories and give some examples. We complete the proof of our~\Cref{result_1} in~\Cref{sec_result_1_proof}. We will give some examples of QCA-refined fusion rules in \Cref{sec_eg_refined} following our \Cref{result_1}.

\subsection{Physical Assumptions and Consequences}\label{sec_asummptions}
We now review the physical assumptions about topological defects of symmetry operators that are needed for our \Cref{result_1}. Although we do not prove these, the
\Crefrange{ass_1}{ass_5} are satisfied by common realization of well-studied symmetries, including group-theoretical symmetries, non-anomalous symmetries, and Kramers-Wannier duality symmetry \cite{Inamura:2026hif}.
We consider a spin chain with tensor-product Hilbert space
\begin{equation}
    \cc{H}=\bigotimes_{i=1}^L h^{(i)}\,,
\end{equation}
where $L$ is the length of the chain and $h^{(i)}$ is the local Hilbert space on site $i$. We assume that $h^{(i)}$ is finite dimensional and independent of $i$ up to isomorphism, that is, $h^{(i)} \simeq h$ for some finite dimensional Hilbert space $h\simeq \bC^d$. Then $d$ is called the on-site dimension. Let $\{\mathsf{D}_x\}$ be a collection of possibly non-invertible symmetry operators acting on $\cc{H}$.

\begin{assumption}\label{ass_1}
    For every symmetry operator $\mathsf{D}_x$, we assume that there are corresponding left and right defect Hilbert spaces $\cc{H}_x^l$ and $\cc{H}_x^r$, which are obtained by modifying the original Hilbert space $\cc{H}$ in a finite region: 
\begin{equation}
    \cc{H}_{\mathsf{D}_x}^l= \left(\bigotimes_{i\in \{1,\cdots ,L\}-I_x}h^{(i)} \right)\otimes L_x,\qquad
    \cc{H}_{\mathsf{D}_x}^r=\left(\bigotimes_{i\in \{1,\cdots ,L\}-J_x}h^{(i)} \right) \otimes R_x.
\end{equation}
Here, $I_x$ and $J_x$ are finite intervals and $L_x$ and $R_x$ are finite dimensional Hilbert spaces, all of which are independent of the system size $L$. In other words, the defect Hilbert spaces are obtained by replacing the Hilbert space on a finite region by a new Hilbert space. Physically, the left (right) defect Hilbert space is the Hilbert space in the presence of a defect $x$ oriented upwards (downwards). We assume that the defect $x$ can be moved by local unitary
operators.\footnote{A defect being oriented upwards (downwards) means that moving the defect to the right (left) around the chain implements the symmetry action of the corresponding operator $\mathsf{D}_x$. More precisely, to obtain the full symmetry action, we need to compose the unitary movement operators with the defect creation and annihilation operators, which make the symmetry operator non-invertible in general \cite{Tantivasadakarn:2025txn}. \label{fn: defect orientation}} In particular, the defect Hilbert spaces do not depend on the position of a defect up to
isomorphism.
\end{assumption}

\begin{assumption}\label{ass_2}
    If $\mathsf{D}_x$ and $\mathsf{D}_y$ are symmetry operators, then so is $\mathsf{D}_x\mathsf{D}_y$. The defect Hilbert spaces for $\mathsf{D}_x\mathsf{D}_y$ are isomorphic to the Hilbert spaces in the presence of two defects $x$ and $y$ apart from each other. More concretely, we assume that there are isomorphisms
    \begin{equation}
    \cc{H}_{\mathsf{D}_x \mathsf{D}_y}^l \simeq \left(\bigotimes_{i\in \{1,\cdots ,L\}-I_x-I_y}h^{(i)}\right)\otimes L_x \otimes L_y, \quad
    \cc{H}_{\mathsf{D}_x \mathsf{D}_y}^r \simeq \left(\bigotimes_{i\in \{1,\cdots ,L\}-J_x-J_y}h^{(i)}\right)\otimes R_x \otimes R_y,
    \end{equation}
    where the overlaps $I_x \cap I_y$ and $J_x \cap J_y$ are empty.

\end{assumption} 
\begin{assumption}\label{ass_3}
    If $\mathsf{D}_x$ and $\mathsf{D}_y$ are symmetry operators, then so is $\mathsf{D}_x+\mathsf{D}_y$. We assume that there are isomorphisms 
\begin{equation}
    \cc{H}^{l}_{\mathsf{D}_x+\mathsf{D}_y}\simeq \cc{H}^{l}_{\mathsf{D}_x}\oplus \cc{H}^{l}_{\mathsf{D}_y}, \qquad
    \cc{H}^{r}_{\mathsf{D}_x+\mathsf{D}_y}\simeq \cc{H}^{r}_{\mathsf{D}_x}\oplus \cc{H}^{r}_{\mathsf{D}_y}.
\end{equation}
\end{assumption} 

\begin{assumption}\label{ass_4}
    If $\mathsf{D}_x$ is a symmetry operator, then so is its Hermitian conjugate $\mathsf{D}_x^\dagger$. We assume that there are isomorphisms 
\begin{equation}
    \cc{H}^l_{\mathsf{D}_x^\dagger}\simeq \cc{H}^r_{\mathsf{D}_x},\qquad
    \cc{H}^r_{\mathsf{D}_x^\dagger}\simeq \cc{H}_{\mathsf{D}_x}^l.
\end{equation}
\end{assumption} 

For each symmetry operator $\mathsf{D}_x$, we define its left and right dimensions by
\begin{equation}
    \ldim(\mathsf{D}_x) \coloneq \frac{\dim(\H_{\mathsf{D}_x}^l)}{\dim(\H)}, \qquad
    \rdim(\mathsf{D}_x) \coloneq \frac{\dim(\H_{\mathsf{D}_x}^r)}{\dim(\H)}.
\end{equation}
Our~\Cref{ass_2} and~\Cref{ass_3} lead to the relations 
\begin{equation}
\ldim(\mathsf{D}_x\mathsf{D}_y)=\ldim(\mathsf{D}_x)\ldim(\mathsf{D}_y),\qquad
\ldim(\mathsf{D}_x+\mathsf{D}_y)=\ldim(\mathsf{D}_x)+\ldim(\mathsf{D}_y),
\end{equation}
and similarly for $\rdim$. This implies that $\ldim$ and $\rdim$ define $\bb{Q}^{>0}$-representations of the algebra of symmetry operators. If $U$ is an invertible symmetry operator, then we have $UU^\dagger=1$, which combined with~\Cref{ass_4} gives  
\begin{equation}
    \ldim(U)\rdim(U)=1.
\end{equation}
For example, if $U$ is a lattice translation by one site to the right, then $\ldim(U)=\dim(h)$ and $\rdim(U)=\dim(h)^{-1}$ \cite{Seifnashri:2023dpa, Seiberg:2024gek}.

    Next we define the lattice quantum dimension $\qdim$ and the index $\ind$ of symmetry operators by 
\begin{equation}
    \qdim(\mathsf{D}_x) \coloneq \sqrt{\ldim(\mathsf{D}_x)\rdim(\mathsf{D}_x)},\qquad
    \ind(\mathsf{D}_x) \coloneq \sqrt{\frac{\ldim(\mathsf{D}_x)}{\rdim(\mathsf{D}_x)}}.
\end{equation}
Due to the multiplicativity  of the left and right dimensions, the lattice quantum dimension and index are also multiplicative: 
\begin{align}
    \qdim(\mathsf{D}_x\mathsf{D}_y)&=\qdim(\mathsf{D}_x)\qdim(\mathsf{D}_y), \\
    \ind(\mathsf{D}_x\mathsf{D}_y)&=\ind(\mathsf{D}_x)\ind(\mathsf{D}_y).
\end{align}
For an invertible symmetry operator $U$, we have $\qdim(U)=1$ and $\ind(U)=\ldim(U)$. For example, for a lattice translation by one site $T$, we have $\ind(T)=\dim(h)$.  For a general QCA, the index defined above agrees with the Gross-Nesme-Vogts-Werner (GNVW) index defined in \cite{Gross:2011yvb}, so our index can be thought of as a generalization of the GNVW index to non-invertible symmetries \cite{Inamura:2026hif}.

Finally, we have the following assumption about the index. 

\begin{assumption}\label{ass_5}
    [Homogeneity of the index] Let $\mathsf{D}_x$ and $\mathsf{D}_y$ be indecomposable symmetry operators that obey the fusion rules $\mathsf{D}_x\mathsf{D}_y=\sum_z \mathsf{D}_z$, where $\mathsf{D}_z$s are also indecomposable. Then, $\ind(\mathsf{D}_z)=\ind(\mathsf{D}_w)$ for all $\mathsf{D}_z$ and $\mathsf{D}_w$ that appear in the sum. 
\end{assumption}
We note here that the above condition is equivalent to that $\ind(\mathsf{D}_z)=\ind(\mathsf{D}_x\mathsf{D}_y)$ for all $\mathsf{D}_z$ that appear in the sum. The proof can be found in~\cite{Inamura:2026hif}.

\subsection{Physical Proof of the Weak Integrality Conjecture}
\label{sec_Necessity}
The above physical assumptions imply the weak integrality conjecture:  any fusion category $\C$ that is realized on a tensor-product Hilbert space $\cc{H}$ up to QCAs must be weakly integral. Here we review the proof given in~\cite{Inamura:2026hif}. We start with the following definition. 
\begin{definition}\label{def_lattice_realization}
    Let $\C$ be a unitary fusion category. A \emph{QCA-refined realization of $\C$} is 
    \begin{itemize}
        \item A tensor-product Hilbert space $\H$;
        \item A collection of symmetry operators\footnote{For our purpose the precise definition of a symmetry operator is irrelevant. Although a proper definition will be important for justifying the physical assumptions made earlier. One possible definition is as a quantum channel acting on local operators on $\H$ with bounded spread, but it is not known if the physical assumptions will follow directly from this definition. We direct the reader to \cite{Okada:2024qmk, Evans:2025msy} for more details of this formulation of symmetry operator.}  $\{\mathsf{D}_x:  x\in \irr(\C)\}$  on $\H$;
        \item Such that 
\begin{equation}
    \mathsf{D}_x\mathsf{D}_y=\sum_{z} U^z_{xy} \mathsf{D}_z\label{eq_lattice_fusion_1} 
\end{equation}
where $U^{z}_{xy}$ is a QCA and the sum  is over simple objects as well as fusion channels;
\item The fusion rules of $\C$ are recovered by setting $U^z_{xy}$ to identity in~\Cref{eq_lattice_fusion_1}.
    \end{itemize}
\end{definition}

We refer to~\Cref{eq_lattice_fusion_1} as  the \textit{QCA-refined fusion rules} of $\C$, to distinguish them from the fusion rules of $\C$. 
In $(1+1)$D. 
Then by applying $\ldim(-)$ to both sides of~\Cref{eq_lattice_fusion_1}, we obtain 
\be
\ba
&\ldim(\mathsf{D}_x)\ldim(\mathsf{D}_y)=\sum_z \ldim(U^z_{xy})\ldim(\mathsf{D}_z)\cr
    \Rightarrow& ~\qdim(\mathsf{D}_x)\qdim(\mathsf{D}_y) \ind(\mathsf{D}_x) \ind(\mathsf{D}_y)=\sum_z \qdim(\mathsf{D}_z) \ind(U^z_{xy}\mathsf{D}_z)\cr
    \Rightarrow& ~\qdim(\mathsf{D}_x) \qdim(\mathsf{D}_y)=\sum_z \qdim(\mathsf{D}_z)
\ea
\ee
In the first step we used that facts that $\ldim=\qdim\cdot \ind$ and  QCAs have trivial $\qdim$. In the second step we used the homogeneity of index. Therefore $\qdim(-)$ defines a one dimensional representation of the fusion rules. It is known that one dimensional representation of the fusion rules is unique and is given by the quantum dimension $\dim_\C(-)$ \cite{egno2015tensor}. We conclude that $\qdim(\mathsf{D}_x)=\dim_\C(x)$. By definition $\qdim(\mathsf{D}_x)$ is a square root of a rational number, which means $\dim_\C(x)^2$ is now a rational number. In any fusion category, the quantum dimension of any object is an algebraic integer \cite{egno2015tensor}. Therefore $\dim_\C(x)^2=\dim_\C(x\otimes x)$ is a rational algebraic integer. Any rational algebraic integer is an integer. We conclude $\dim_\C(x)^2$ is an integer for any simple $x\in \C$. This is exactly the condition that $\C$ is weakly integral. 
To clarify, there are two definitions of weakly integrality. One is that the total quantum dimension $\dim(\C)=\sum_{x\in \C}\dim_\C(x)^2$ is an integer, another is that $\dim_\C(x)^2$ is an integer for every simple object $x\in \C$. It turns out these two definitions are equivalent \cite{egno2015tensor}. 

\subsection{Weakly Integral Fusion Categories}\label{sec_weakly_integral}
Before we move on to prove~\Cref{result_1}, we need to review some basic facts about weakly integral fusion categories. The key structural result we need about weakly integral fusion categories is summarized by the following Proposition, which is adapted from Proposition 3.5.7 of ~\cite{egno2015tensor}. 
\begin{proposition}[\cite{egno2015tensor} Proposition 3.5.7]\label{prop_1}
    Let $\C$ be a weakly integral fusion category. Then there is a group $E\simeq \Z_2^r$ and a set of distinct square-free positive integers $n_g, g\in E$, with $n_0=1$, such that there is faithful $E$-grading $\C=\oplus_g \C_g$ and $d_x=a_x \sqrt{n_g},~\forall x\in \C_g$, with $a_x \in \Z^{>0}$. Furthermore for any $g,h\in E$, $q(g,h):=\sqrt{\frac{n_gn_h}{n_{gh}}}\in \Z^{>0}$.
\end{proposition}

 In other words, in a weakly integral fusion category $\C$,  the quantum dimension of a simple object $d_x$ is the product of an integral part $a_x$ and a square-root part $\sqrt{n_x}$, and the square-root part of quantum dimension defines a grading: if $z$ is contained in the fusion outcome of $x\otimes y$, then $\sqrt{n_x}\sqrt{n_y}=\sqrt{n_z}$ up to rational numbers. The trivial component $\C_0$ is spanned by simple objects with integral quantum dimension and is called the integral subcategory.
Let us look at some examples. 

    \begin{example}
       Let $A$ be an abelian group and $\TY(A)$ be a Tambara-Yamagami category \cite{Tambara:1998vmj}. Then we have $\TY(A)=\vc_A\oplus \vc$. Denote the non-invertible simple as $D$, we have $d_D=\sqrt{|A|}$. If $|A|$ is a perfect square (e.g. $A=\Z_4, \, \Z_2^2$), then $d_D\in \Z$ and $\TY(A)$ is an integral fusion category. If $|A|$ is not a perfect square, then we have two component $\C_0=\vc_A$ and $\C_1=\vc$. E.g., for any prime $p$ and $A=\Z_p$ we have $n_1=p$ and $q(1,1)=p$. We will discuss these fusion categories, QCA-refinement and their lattice realization in \Cref{sec_eg_TY}.
    \end{example}
\begin{example}
   The Metaplectic fusion categories $SO(M)_2$ are  weakly integral for any positive integer $M$.  Such a fusion category is integral precisely when $M$ is an odd perfect square or $M$ is even and $M/2$ is a perfect square. Let us consider a simple case $SO(3)_2$. It has simple objects
    \be 
    \ut,Z,Y, X, X'
    \ee
    with quantum dimensions
    \be 
    d_{\ut}=d_{Z}=1,\qquad d_{Y}=2,\qquad d_{X}=d_{X'}=\sqrt{3}.
    \ee
    Hence $SO(3)_2$ is weakly integral but not integral. The corresponding grading group is $E=\Z_2$, with
    \be 
    \C_0=\langle \ut,Y,Z\rangle,\qquad \C_1=\langle X,X'\rangle.
    \ee
    Thus the integral fusion subcategory is $\C_0$, and for the nontrivial element $1\in \Z_2$ we have
    \be 
    n_1=3,\qquad q(1,1)=\sqrt{\frac{n_1^2}{n_0}}=3.
    \ee
    The fusion rules are
    \begin{align*}
        Z\otimes Z&\cong \ut,\\
        Z\otimes X&\cong X',\qquad Z\otimes Y\cong Y,\qquad Z\otimes X'\cong X,\\
        X\otimes X&\cong \ut\oplus Y,\qquad X\otimes Y\cong X\oplus X',\qquad X\otimes X'\cong Y\oplus Z,\\
        Y\otimes Y&\cong \ut\oplus Y\oplus Z,\qquad Y\otimes X'\cong X\oplus X',\qquad X'\otimes X'\cong \ut\oplus Y.
    \end{align*}
    One checks that these fusion rules agree with the $\Z_2$-grading. Moreover, $\C_0$ is equivalent as a fusion category to $\Rep(S_3)$. Indeed, the simple objects of $\C_0$ satisfy
    \be 
    Z\otimes Z\cong \ut,\qquad Z\otimes Y\cong Y,\qquad Y\otimes Y\cong \ut\oplus Z\oplus Y,
    \ee
    which are exactly the fusion rules of the simple representations of $S_3$. We will discuss general Metaplectic categories and their QCA-refined realizations in more details in  \Cref{sec_eg_refined}.
\end{example}

\subsection{Constraints on Indices}\label{sec_result_1_proof}
We are now ready to prove~\Cref{result_1}. We present it here as the following theorem.
\begin{theorem}\label{thm_main}
    Let $\C=\oplus_g \C_g$ be a weakly integral fusion category as in~\Cref{prop_1}. If it has a QCA-refined realization
    \begin{equation}
        \mathsf{D}_x\mathsf{D}_y=\sum_z U^z_{xy} \mathsf{D}_z
    \end{equation}
    satisfying~\Crefrange{ass_1}{ass_5}, 
    then for $x\in \C_g, y\in \C_h$ we have
    \begin{equation}
        \ind(\mathsf{D}_x)=m_x\sqrt{n_g},\quad \ind(U^z_{xy})=\frac{m_xm_y}{m_z}q(g,h)=\frac{m_xm_y}{m_z}\sqrt{\frac{n_gn_h}{n_{gh}}}, 
    \end{equation} 
for some $\{m_x\in \bb{Q}^{>0}: x\in \irr(\C)\}$ not determined by $\C$. 

Notice that if we redefine symmetry operators by attaching with QCAs: $\widetilde{\sf{D}}_x:= U_{r_x}\sf{D}_x$, for some QCA $U_{r_x}$ that commutes with symmetry operators and has index $r_x$, then the redefined symmetry operators and the associated QCAs $\widetilde{U}^z_{xy}$ satisfy
\begin{equation}
   \ind(\widetilde{\sf{D}}_x)= m_xr_x\sqrt{n_g},\quad\ind(\widetilde{U}^z_{xy})= \frac{m_xr_xm_yr_y}{m_zr_z}\sqrt{\frac{n_gn_h}{n_{gh}}}.
\end{equation}
Thus up to redefinition of symmetry operators in the above sense, the indices $\ind(\sf{D}_x), \ind(U^z_{xy})$ are completely determined by the weakly integral fusion category. For instance, we can stack with another spin chain with on-site dimension $\mathrm{lcm}(\{m_x,x\in \irr(\C)\})$, and take $U_{r_x}=U_{m_x^{-1}}$ to be QCAs that only act on the stacked spin chain, then we can set $\ind(\sf{D}_x)=\sqrt{n_g},~\ind(U^z_{xy})=q(g,h)$, both of which are now determined by $\C$.

\proof From the proof of the  weak integrality conjecture in~\Cref{sec_Necessity} we know that $\qdim(\mathsf{D}_x)=d_x$. Thus $\ind(\mathsf{D}_x)=\frac{\qdim(\mathsf{D}_x)}{\rdim(\mathsf{D}_x)}=\frac{d_x}{\rdim(\mathsf{D}_x)}$. Let $x\in \C_g, y\in \C_h, z\in \C_{gh}$ be simple objects, we can write \begin{equation}
d_x=a_x\sqrt{n_g},~d_y=a_y\sqrt{n_h},~d_z=a_z\sqrt{n_{gh}}.
\end{equation}
By homogeneity of index we have 
\begin{equation}
    \ind(U^z_{xy})=\frac{\ind(\mathsf{D}_x)\ind(\mathsf{D}_y)}{\ind(\mathsf{D}_z)}=\frac{\rdim(\mathsf{D}_z)}{\rdim(\mathsf{D}_x)\rdim(\mathsf{D}_y)}\frac{d_xd_y}{d_z}=\frac{m_xm_y}{m_z}\sqrt{\frac{n_gn_h}{n_{gh}}}
\end{equation}
where we defined $m_x:=\frac{a_x}{\rdim(\mathsf{D}_x)}\in \bb{Q}^{>0}$.\qed 
\end{theorem}

In fact, we can ask for more. The lattice model that we will construct later satisfies the property that the QCAs $U^z_{xy}$ with the same index are in fact the same QCA. This together with the redefinition of symmetry operator trick discussed above implies that there exists a QCA-refined realization of the form 
\begin{equation}
    \sf{D}_x\sf{D}_y=U_{q(g,h)} \sum_z \sf{D}_z,\quad x\in \C_g, y\in \C_h, \label{eq_canonical_form}
\end{equation}
where $U_{q(g,h)}$ is a QCA with index $q(g,h)$. Let us make this into a definition.

\begin{definition}
    A QCA-refined realization of a weakly integral categorical symmetry is called \emph{canonical} if it takes the form \Cref{eq_canonical_form}.
\end{definition}
\begin{corollary}
    For every weakly integral fusion category, a canonical QCA-refined realization exists. \qed
\end{corollary}

\subsection{Examples of QCA-Refined Fusion Rules}\label{sec_eg_refined}
Here we consider two families of concrete weakly integral fusion categories and their canonical QCA-refined fusion rules \Cref{eq_canonical_form}. 
To date, there is no classification of weakly integral fusion categories. The main infinite series of examples are the Tambara-Yamagami categories and Metaplectic categories.  A complete classification of modular weakly integral categories\footnote{A modular category is a non-degenerate unitary braided fusion category. }  for rank $\le 7$ is obtained in \cite{bruillard2016classification}, which states any  weakly integral modular category with $\text{rank}\le 7$ is a tensor product of pointed categories, Ising categories and Metaplectic
categories. 
We now discuss the canonical QCA-refined realization for the Tambara-Yamagami (TY) categories, which are direct generalization of the Ising category, and the Metaplectic categories $SO(M)_{2}$.

\begin{example}\label{eg_TY_canonical_refined}{\bf Tambara-Yamagami Fusion Categories.}    \\
Let $A$ be an abelian group with $|A|$ not a perfect square, and $\TY(A)$ be a Tambara-Yamagami category. There are additional data such as the bicharacter and Frobenius-Schur indicator, which will not play a role here since they do not affect the fusion rules. 
Denote the non-invertible simple by $m$, which has $\dim (m)=\sqrt{|A|}$ and the invertible simples of dimension 1 by $a\in A$. The standard TY fusion rules are 
\be
m\otimes m = \bigoplus_{a\in A} a \,,\qquad m \otimes a = a\otimes m = m
\ee
and the $a$ fuse according to the group law in $A$. 
Then its canonical QCA-refined fusion rules are
\be 
\ba
\sf{D}_m\sf{D}_m &= U_{|A|} \sum_{a\in A}\sf{D}_a \cr 
\mathsf{D}_a\sf{D}_m &=\sf{D}_m \sf{D}_a=\sf{D}_m \,,\qquad  
\mathsf{D}_a\mathsf{D}_b =\mathsf{D}_{ab} \,.
\ea
\ee
We will construct the tensor-product Hilbert space realization of this in \Cref{sec_eg_TY}.
\end{example}

\begin{example}{\bf Metaplectic Fusion Categories.}
\label{sec:Meta}
Another class of weakly integral fusion categories are the Metaplectic fusion categories $SO(M)_2$, whose complete categorical data have been calculated \cite{ardonne2021classification, gelaki2009centers}. 
They can be viewed as $\Z_2$-graded extensions of the representation category of the dihedral group $D_M$, $\Rep (D_{M})$, or alternatively as obtained from taking a $\Z_2$-equivariantization of the Drinfeld center $\cZ(\Rep(D_M))$. We consider a Metaplectic fusion category $SO(M)_2$ for odd $M=2r+1$. The simple objects are  
\be
\mathbf{1},\ Z,\ X,\ X',\ Y_1,\dots,Y_r,
\qquad M= 2r+1 \,,
\ee
with quantum dimensions
\be\ba 
\dim(\mathbf{1}) &=\dim(Z)=1  \cr 
\dim(X)&=\dim(X')=\sqrt{M} \cr 
\dim(Y_j)&=2  \,.
\ea\ee
Hence the total quantum dimension is
\be\ba 
\dim(\cC)
&=1^2+1^2+(\sqrt{M})^2+(\sqrt{M})^2+r\cdot 2^2\cr 
&=2+2M+4\cdot \frac{M-1}{2}=4M \,.
\ea\ee
So these categories are {\bf weakly integral}. For
$M$ not a perfect square they are not integral. This can also be seen from the fact that  $X$ and $X'$ have non-integral  quantum dimension
$\sqrt{M}$ whenever $M$ is not a perfect square. 
The fusion rules are:
\be \ba
Z\otimes Z &= \mathbf{1} \cr 
Y_i\otimes Y_j &= Y_{\min\{i+j,M-i-j\}}\oplus Y_{|i-j|}
\qquad (i\neq j),\\
Y_i\otimes Y_i &= \mathbf{1}\oplus Z\oplus Y_{\min\{2i,M-2i\}} \cr 
Z\otimes X &= X', \qquad Z\otimes X' = X, \qquad Z\otimes Y_i = Y_i \cr 
X\otimes X &= X'\otimes X'= \mathbf{1}\oplus \bigoplus_{i=1}^{r} Y_i,\qquad
X\otimes X'= Z\oplus \bigoplus_{i=1}^{r} Y_i,\\
X\otimes Y_i &= X'\otimes Y_i= X\oplus X' \,.
\ea
\label{eq_meta_fusion_rule}
\ee
The generators $1, Z, Y_i, i=1, \cdots, r$ form a fusion subcategory  $\Rep (D_M)$ \cite{ardonne2021classification}, as can be seen from the first three lines of the fusion rules above.  The full $SO(M)_2$ is a $\Z_2$-extension of this $\Rep (D_M)$ subcategory with the nontrivial component generated by $X, X'$:
\be
\cC_0 =  \langle  \mathbf{1}, Z, Y_1, \cdots, Y_r \rangle \,,\qquad 
\cC_1= \langle X, X'\rangle \,.
\ee
We now consider the canonical QCA-refined fusion rules, assuming $M$ is square-free. Following the convention in \Cref{prop_1} we have 
\begin{equation}
    n_1=M, \quad  q(1,1)=M.
\end{equation}
Then by \Cref{eq_canonical_form} the canonical QCA-refined fusion rules are
\be \ba
\Dlat_X \Dlat_X &=\Dlat_{X'} \Dlat_{X'} = U_M \left(\Dlat_{\mathbf{1}} +  \sum_{i=1}^{r} \Dlat_{Y_i} \right) \cr 
\Dlat_{X} \Dlat_{X'} &= U_M \left(\Dlat_Z+  \sum_{i=1}^{r} \Dlat_{Y_i} \right)  \,,
\ea
\ee
where $U_M$ is a QCA with index $M$, and the rest of the fusion rules stay the same as \Cref{eq_meta_fusion_rule}. In particular, whenever both factors lie in the nontrivial graded component $\C_1$, the product acquires the QCA-refinement $U_M$, while all other products have trivial QCA-refinement. Finally the indices of symmetry operators in the canonical QCA-refined realization are 
 \be
\ind(\mathsf{D}_{\mathbf{1}})=\ind(\mathsf{D}_{Y_i})=\ind(\mathsf{D}_{Z})=1,
    \qquad
    \ind(\mathsf{D}_{X})=\ind(\mathsf{D}_{X'})=\sqrt{M}.
\ee
\end{example}


\section{Lattice Model for Weakly Integral Fusion Categories}\label{sec_result_2}
We now move on to present our~\Cref{result_2}. We will construct a lattice model that realizes any weakly integral fusion categories on a tensor-product Hilbert space up to QCAs, in the sense of~\Cref{def_lattice_realization}.  Our construction is a generalization of the construction for integral fusion categories given in~\cite{Evans:2025msy}, which we now review in a  physicist-friendly way.
\subsection{Lattice Model for Integral Fusion Categories}\label{sec_integral_model}
It is proven in~\cite{Evans:2025msy} that a fusion categorical symmetry can be realized on a tensor-product Hilbert space strictly (without QCA refinement) if and only if it is integral. The ``if" direction is proven by explicitly constructing a tensor-product realization of integral fusion categories.  The construction is based on the anyon chain model, depicted in~\Cref{eq_anyon_chain}. 

The anyon chain \cite{Feiguin:2006ydp,
Trebst:2008fibonacci,
Gils:2008collective,
Gils:2009topology,
Pfeifer:2010translation,
Ardonne:2010yanglee,
Gils:2013spin1, Aasen:2016dop,
Buican:2017anyonic} is defined by a triple $(\C,\M, \rho)$, where $\C$ is a unitary fusion category, $\M$ is a unitary module category over $\C$, and $\rho$ is a not necessarily simple object in $\C$:
\be 
\begin{split}
\begin{tikzpicture}
\begin{scope}[scale=1.8]
\draw [very thick,->-](-4,0) -- (-3,0); 
\draw [very thick,->-](-3,0) -- (-2,0);
\draw [very thick,->-](-2,0) -- (-1,0);
\draw [very thick,->-](-1,0) -- (0,0);
\draw [very thick,->-](0,0) -- (1,0);
\foreach \x in {-3,-2,-1,0}{
 \draw[fill=black] (\x,0) circle (0.0556);
}
 \draw [very thick,->-](-3,-0.6) -- (-3,0); 
  \draw [very thick,->-](-2,-0.6) -- (-2,0); 
   \draw [very thick,->-](-1,-0.6) -- (-1,0); 
 \draw [very thick,->-](0,-0.6) -- (0,0); 
 \foreach \x in {-3,-2,-1,0} {
  \node[below] at (\x, -0.6) {$\rho$};
 }
\node[above] at (-3.5, 0) {$x_i$} ;
\node[above] at (-2.5, 0) {$x_{i+1}$} ;
\node[above] at (-1.5, 0) {$x_{i+2}$} ;
\node[above] at (-0.5, 0) {$x_{i+3}$} ;
\node[above] at (0.5, 0) {$x_{i+4}$} ;
\node[above] at (-3, 0) {${\nu_i}$} ;
\node[above] at (-2, 0) {$\nu_{i+1}$} ;
\node[above] at (-1, 0) {$\nu_{i+2}$} ;
\node[above] at (0, 0) {$\nu_{i+3}$} ;
\end{scope}
\end{tikzpicture}
\end{split}
\label{eq_anyon_chain}
\ee 
Our conventions for fusion diagrams are summarized in \Cref{app:Conventions}.
The vertical legs are labeled by the object $\rho \in \C$. The horizontal links have degrees of freedom labeled by simple objects $x_i\in \M$. For any fixed  configuration of simple objects $\{x_i\}$, we assign a Hilbert space $V_{i}:=\Hom_\M(x_i\lhd \rho, x_{i+1})$ at each vertex. We choose the projection basis at each vertex: 
\begin{equation}
    |\nu_i\>=p_{ x_{i+1}}^{\nu_i}:\quad  x_i\otimes \rho\to x_{i+1},\quad \nu_i=1,\cdots,\dim(V_{i}).
\end{equation}
The Hilbert space of an anyon chain with the above input is then
\begin{equation}
    \cc{H}=\bigoplus_{\{x_i\in \M\}}\bigotimes_i \Hom_\M(x_i\lhd\rho, x_{i+1}),
\end{equation}
which is generally not a tensor product of local Hilbert spaces. Give an object $x\in \C$, the corresponding symmetry operator $\cc{D}_x$ is defined by fusing with an $x$-line from above the anyon chain, and then deforming the diagram to a linear combination of the basis define before using $F$-moves: 
\be \label{eq_anyon_chain_action}
\begin{split}
\begin{tikzpicture}
\begin{scope}[shift={(0,0)}, scale=1.8]
\draw [very thick,newblue, ->-] (-4,0.6) -- (1,0.6);
 \node[above] at (-1, 0.6) {$x$};
\foreach \x in {-4,-3,-2,-1,0} {
\draw [very thick,->-](\x,0) -- ({\x+1},0);
}
 \foreach \x in {-3,-2,-1,0} {
  \draw[fill=black] (\x,0) circle (0.0556);
\draw [very thick,->-](\x,-0.6) -- (\x,0);
 \node[below] at (\x, -0.6) {$\rho$};
} 
\node[above] at (-3.5, 0) {$x_i$} ;
\node[above] at (-2.5, 0) {$x_{i+1}$} ;
\node[above] at (-1.5, 0) {$x_{i+2}$} ;
\node[above] at (-0.5, 0) {$x_{i+3}$} ;
\node[above] at (0.5, 0) {$x_{i+4}$} ;
\end{scope}
\end{tikzpicture}
\end{split}
\ee
Now let us assume $\C$ is integral and  take $\M=\C$ to be the regular module category. The key observation is that in an integral fusion category the following object is well-defined,
\begin{equation}
    R:=\bigoplus_{x\in \irr(\C)} d_x x.
\end{equation}
Here $d_xx$ means $d_x$-fold direct sum of $x$: $d_xx=\oplus^{d_x}x$. This is well-defined because $d_x$ is an integer. The object $R$ is called the regular object in mathematics literature~\cite{egno2015tensor}. It has the crucial property that for any simple $z\in \C$, $z\otimes R=d_z R$. Indeed,

\be
\ba
    z\otimes R&=\bigoplus_x d_x z\otimes x=\bigoplus_{x,y} d_x N^y_{zx}y\\
    &=\bigoplus_{x,y} d_x N^{x}_{z^*y} y=\bigoplus_y d_zd_yy=d_z R.
\ea\ee
This then implies
\begin{equation}
   \dim \Hom_\C(x\otimes R, y)=d_xd_y, \quad \forall x,y\in \irr(\C).
\end{equation}
Then the Hilbert space $V_i$ associated to a vertex can be decomposed into a tensor product $\bC^{d_{x_i}}\otimes \bC^{d_{x_{i+1}}}$. One pictures that the space $\bC^{d_{x_i}}$ now lives on the link $x_i$, and $\bC^{d_{x_{i+1}}}$ lives on the link $x_{i+1}$. The link $x_i$ then hosts a Hilbert space $\bC^{d_{x_{i}}^2}$.  The total Hilbert space of the anyon chain is then 
\begin{equation}
    \cc{H}=\bigoplus_{\{x_i\in \C\}} \bigotimes_{i}\bC^{d_{x_i}^2}=\bigotimes_i \bC^{\sum_{x\in \C}d_x^2}=\bigotimes_i \bC^{\dim(\C)} \,,
\end{equation}
which takes a tensor product form with on-site Hilbert space $\bC^{\dim(\C)}$. Diagrammatically, we have rearranged the vertex Hilbert spaces as follows:
     \be 
\begin{split}
\begin{tikzpicture}
\begin{scope}[shift={(0,0)}, scale=1.8]
\draw [very thick,->-](-4,0) -- (-3,0); 
\draw [very thick,->-](-3,0) -- (-2,0);
\draw [very thick,->-](-2,0) -- (-1,0);
 \draw[fill=black] (-3,0) circle (0.0556);
 \draw[fill=black] (-2,0) circle (0.0556);;
 \draw [very thick,->-](-3,-0.6) -- (-3,0); 
  \draw [very thick,->-](-2,-0.6) -- (-2,0); 
 \node[below] at (-3, -0.6) {$R$};
 \node[below] at (-2, -0.6) {$R$};
\node[below] at (-3.5, 0) {$x_i$} ;
\node[below] at (-2.5, 0) {$x_{i+1}$} ;
\node[below] at (-1.5, 0) {$x_{i+2}$} ;
\node[above] at (-3, 0.1) {${\bC^{d_{x_i}d_{x_{i+1}}}}$} ;
\node[above] at (-2, 0.1) {$\bC^{d_{x_{i+1}}d_{x_{i+2}}}$} ;
\end{scope}
\end{tikzpicture}
\end{split}
\quad \quad 
\To
\begin{split}
\begin{tikzpicture}
\begin{scope}[shift={(0,0)}, scale=1.8]
\draw [very thick](-4,0) -- (-3,0); 
\draw [very thick](-3,0) -- (-2,0);
\draw [very thick](-2,0) -- (-1,0);
 \draw [very thick,->-](-3,-0.6) -- (-3,0); 
  \draw [very thick, ->-](-2,-0.6) -- (-2,0); 
 \node[below] at (-3, -0.6) {$R$};
 \node[below] at (-2, -0.6) {$R$};
\node[below] at (-3.5, 0) {$x_i$} ;
\node[below] at (-2.5, 0) {$x_{i+1}$} ;
\node[below] at (-1.5, 0) {$x_{i+2}$} ;
\node[above] at (-3.5, 0.1) {$\bC^{d_{x_i}^2}$} ;
\node[above] at (-2.5, 0.1) {$\bC^{d_{x_{i+1}}^2}$} ;
\node[above] at (-1.5, 0.1) {$\bC^{d_{x_{i+2}}^2}$} ;
  \draw[fill=orange,draw=black,very thick] (-2.5,0) circle (0.0556);
  \draw[fill=orange,draw=black,very thick] (-3.5,0) circle (0.0556);
  \draw[fill=orange,draw=black,very thick] (-1.5,0) circle (0.0556);
\end{scope}
\end{tikzpicture}
\end{split}
\ee

This shows the anyon chain for an integral fusion category $\C$ with input $(\C, \C,R)$ has a tensor-product Hilbert space structure. 

\subsection{Lattice Model for Weakly Integral Fusion Categories}\label{sec_lattice_model_weakly}
We now turn to a general weakly integral fusion category. Our construction is a combination of an anyon chain model similar to the previous one for integral fusion categories and an "anyonic" sequential circuit action. After discussing the Hilbert space structure of the modified anyon chain model, we will revisit the Ising category case to motivate our general construction that involves a sequential circuit action.
\subsubsection{Setup}
We now consider a weakly integral fusion category $\C=\oplus_{g\in E} \C_g$ with its standard grading by $E$. Consider the anyon chain model as in \Cref{eq_anyon_chain}, but change $\rho$ to 
\begin{equation}
    \rho=R_0:=\bigoplus_{x\in \irr(\C_0)} d_x x,
\end{equation}
i.e., the regular object in the integral subcategory $\C_0$. Take $x\in \C_g$, we now have
\begin{align}
    x\otimes R_0&=\bigoplus_{z\in \irr(\C_0)} d_zx\otimes z=\bigoplus_{z\in \irr(\C_0), y\in \irr(\C_g)}d_zN^y_{xz} y\nonumber\\
    &=\bigoplus_{z\in\irr(\C_0), y\in \irr(\C_g)}d_zN^{z}_{x^*y}y=\bigoplus_{y\in \irr(\C_g)} d_x d_y y.
\end{align}
Notice $d_x$ and $d_y$ may not be integers but $d_xd_y$ is always an integer: write 
\be 
d_x=a_x\sqrt{n_g}\,,\qquad d_y=a_y\sqrt{n_g}, 
\ee 
then we have 
\begin{equation}
    x\otimes R_0=n_g \bigoplus_{y\in \irr(\C_g)} a_xa_y y \,,
\end{equation}
which implies for $x\in \C_g, y\in \C_h$
\begin{equation}\label{eq_dim_weakly_int}
    \dim \Hom_\C(x\otimes R_0,y)=\delta_{g,h} a_xa_yn_g.
\end{equation}

This has the following consequences. First the total Hilbert space of the anyon chain now decomposes into $|E|$-many sectors: 
\begin{equation}\label{HWI}
    \cc{H}=\bigoplus_{g\in E}\cc{H}_g.
\end{equation}
This is because for $x_i, x_{i+1}$ belonging to different graded components of $\C$, we have $V_i=\Hom_{\C}(x_i\otimes R_0, x_{i+1})=0$. For each sector $g\in E$, the links have anyon labels $x_i\in \C_g$. For fixed anyon configuration $\{x_i\in \C_g\}$ in the $g$-sector, the Hilbert space associated with a vertex has dimension
\begin{equation}
    \dim \Hom_\C(x_i\otimes R_0, x_{i+1})=a_{x_i}a_{x_{i+1}}n_g.
\end{equation}

We may now rearrange the vertex Hilbert spaces as in \Cref{fig_anyon_chain_weakly_integral}.
\begin{figure}[ht]
    \centering
     \be 
\begin{split}
\begin{tikzpicture}
\begin{scope}[shift={(0,0)}, scale=1.8]
\draw [very thick,->-](-4,0) -- (-3,0); 
\draw [very thick,->-](-3,0) -- (-2,0);
\draw [very thick,->-](-2,0) -- (-1,0);
 \draw[fill=black] (-3,0) circle (0.0556);
 \draw[fill=black] (-2,0) circle (0.0556);
 \draw [very thick,->-](-3,-0.6) -- (-3,0); 
  \draw [very thick,->-](-2,-0.6) -- (-2,0); 
 \node[below] at (-3, -0.6) {$R_0$};
 \node[below] at (-2, -0.6) {$R_0$};
\node[below] at (-3.5, 0) {$x_i$} ;
\node[below] at (-2.5, 0) {$x_{i+1}$} ;
\node[below] at (-1.5, 0) {$x_{i+2}$} ;
\node[above] at (-2.8, 0.1) {${\bC^{n_ga_{x_i}a_{x_{i+1}}}}$} ;
\node[above] at (-1.8, 0.1) {$\bC^{n_ga_{x_{i+1}}a_{x_{i+2}}}$} ;
\end{scope}
\end{tikzpicture}
\end{split}
\quad \quad 
\To
\begin{split}
\begin{tikzpicture}
\begin{scope}[shift={(0,0)}, scale=1.8]
\draw [very thick](-4,0) -- (-3,0); 
\draw [very thick](-3,0) -- (-2,0);
\draw [very thick](-2,0) -- (-1,0);
 \draw [very thick,->-](-3,-0.6) -- (-3,0); 
  \draw [very thick,->-](-2,-0.6) -- (-2,0); 
 \draw[fill=blue,draw=black, very thick] (-3,0) circle (0.0556);
 \draw[fill=blue,draw=black, very thick] (-2,0) circle (0.0556);
 \node[below] at (-3, -0.6) {$R_0$};
 \node[below] at (-2, -0.6) {$R_0$};
\node[below] at (-3.5, 0) {$x_i$} ;
\node[below] at (-2.5, 0) {$x_{i+1}$} ;
\node[below] at (-1.5, 0) {$x_{i+2}$} ;
\node[above] at (-3.5, 0.1) {$\bC^{a_{x_i}^2}$} ;
\node[above] at (-2.5, 0.1) {$\bC^{a_{x_{i+1}}^2}$} ;
\node[above] at (-1.5, 0.1) {$\bC^{a_{x_{i+2}}^2}$} ;
\node[above] at (-2, 0.1) {$\bC^{n_g}$} ;
\node[above] at (-3, 0.1) {$\bC^{n_g}$} ;
  \draw[fill=orange,draw=black, very thick] (-2.5,0) circle (0.0556);
  \draw[fill=orange,draw=black, very thick] (-3.5,0) circle (0.0556);
  \draw[fill=orange,draw=black, very thick] (-1.5,0) circle (0.0556);
\draw[newred, dashed, very thick] (-3.7,0.1) -- (-2.8,0.1) -- (-2.8, -0.1) -- (-3.7,-0.1) -- (-3.7, 0.1);
\end{scope}
\end{tikzpicture}
\end{split}
\ee
\vspace{-0.2cm}
    \caption{Rearranging the local Hilbert spaces for the anyon chain with a weakly-integral fusion category input. The red rectangle is a single site.}
    \label{fig_anyon_chain_weakly_integral}
\end{figure}
After this rearranging, each link $i$ hosts a  Hilbert space of dimension $a_{x_i}^2$. Let us define a site to be a vertex together with the adjacent link to its left. Then each site has dimension $a_{x_i}^2n_g$. Summing over the anyon configurations $\{x_i\in \C_g\}$, the $g$-sector Hilbert space is 
\begin{equation}
    \cc{H}_g=\bigoplus_{\{x_i \in \C_g\}}\bigotimes_i \bC^{a_{x_i}^2n_g}=\bigotimes_{i}\bC^{\sum_{x\in \C_g} a_x^2n_g}=\bigotimes_i \bC^{\dim(\C_g)}=\bigotimes_i \bC^{\dim(\C_0)}.
\end{equation}

We see that each sector $\cc{H}_g$ is has a tensor product structure with on-site dimension $\dim(\C_g)=\dim(\C_0)$. This means we have site-to-site isomorphism $\cc{H}_g\simeq \cc{H}_0$, which roughly speaking shifts a $\bC^{n_g}$ by half a site.  Given a simple object $x\in \C_g$, we define the operator $\D_x$ to be the usual action on $\cc{H}$ by fusing with a line of $x$ from above the anyon chain (\Cref{eq_anyon_chain_action}). It acts on sectors as 
\begin{equation}
    \D_x: \cc{H}_h\to \cc{H}_{gh}.
\end{equation}

Since $\H=\oplus_{g\in E} \H_g$ does not have a tensor product structure, the operators $\D_x$ do not yet form a representation of $\C$ on a single tensor-product Hilbert space. 
Our strategy is to introduce locality-preserving isomorphisms $U_g: \cc{H}_g\simeq \H_0$ in a coherent way so that the modified symmetry operators $\sf{D}_x:=U_g\circ \D_x|_{\H_0}: \H_0\to \H_0$ form a representation of $\C$ on a single tensor product space $\cc{H}_0$. It turns out this can be achieved at the cost of introducing  QCA-refinement in the fusion rules. 

Next we consider a simple example that illustrates the main idea, namely the Ising fusion category.

\subsubsection{Ising Category Revisited}\label{sec_ising_example}
 It is known that the Ising fusion category symmetry can be realized on a tensor-product Hilbert space, up to lattice translation;  explicit lattice models can be found in~\cite{Seiberg:2023cdc, Seiberg:2024gek}. Our goal here is to revisit this example in a way that makes the correct generalization to an arbitrary weakly integral fusion category transparent.

The Ising fusion category is weakly integral with grading $\mathrm{Ising}=\vc_{\Z_2}\oplus \vc$, and is a special case of the TY categories in \Cref{eg_TY_canonical_refined} with $A=\Z_2$.
Denote the invertible simples as $1, f$ and the non-invertible simple as $\sigma$ with $\sigma^2 = 1\oplus f$. We have $a_1=a_f=a_\sigma=1$ and $n_1=n_f=1, n_\sigma=2$. The regular object of the integral subcategory is $R_0=1 \oplus f$.  Then the anyon chain with input $(\mathrm{Ising},\mathrm{Ising}, 1 \oplus f)$ has the Hilbert space structure
    \begin{equation}
        \cc{H}=\cc{H}_{1,f}\oplus \cc{H}_\sigma.
    \end{equation}  
    For the trivial sector $\cc{H}_{1,f}$, the links are labeled by $x_i\in\{1,f\}$. The vertex Hilbert space $V_i=\bC$ is one dimensional. Therefore the local Hilbert spaces can be associated with links and are $\bC^2=\mathrm{Span}(|1\>, |f\>)$. See \Cref{fig_anyon_chain_Ising_triv}.
    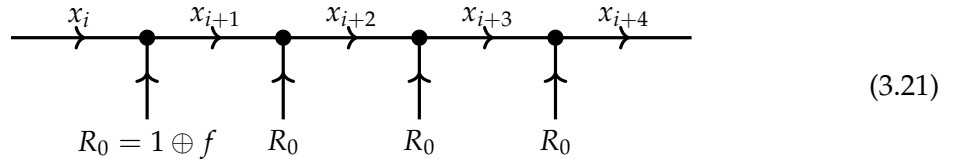
\begin{figure}[ht]
    \centering
     \be 
\begin{split}
\begin{tikzpicture}
\begin{scope}[shift={(0,0)}, scale=1.8]
\draw [very thick,->-](-4,0) -- (-3,0); 
\draw [very thick,->-](-3,0) -- (-2,0);
\draw [very thick,->-](-2,0) -- (-1,0);
\draw [very thick,->-](-1,0) -- (0,0);
\draw [very thick,->-](0,0) -- (1,0);
 \foreach \x in {-3,-2,-1,0} {
  \draw [very thick,->-](\x,-0.6) -- (\x,0); 
   \draw[fill=black] (\x,0) circle (0.0556);
 }
 \node[below] at (-3, -0.6) {$R_0=1 \oplus f$};
 \node[below] at (-2, -0.6) {$R_0$};
 \node[below] at (-1, -0.6) {$R_0$};
 \node[below] at (-0, -0.6) {$R_0$};
\node[above] at (-3.5, 0) {$x_i$} ;
\node[above] at (-2.5, 0) {$x_{i+1}$} ;
\node[above] at (-1.5, 0) {$x_{i+2}$} ;
\node[above] at (-0.5, 0) {$x_{i+3}$} ;
\node[above] at (0.5, 0) {$x_{i+4}$} ;
\end{scope}
\end{tikzpicture}
\end{split}
\ee
\vspace{-0.2cm}
    \caption{The trivial sector $\H_{1,f}$. At each link $x_i\in \{1,f\}$. }
\label{fig_anyon_chain_Ising_triv}
\end{figure}
For the nontrivial sector $\cc{H}_\sigma$, the links have fixed anyon label $\sigma$. The vertex Hilbert space is $V_i=\bC^2$, with a basis given by the two projections: 
    \begin{equation}
      \id_\sigma\otimes p_1,~\id_\sigma\otimes p_f: \quad \sigma\otimes (1 \oplus f) \to \sigma.
    \end{equation}
    We can place the projections on the vertical legs as shown in~\Cref{fig_anyon_chain_Ising_nontriv}. Then the on-site Hilbert space basis may be labeled as $|t_i\>\in \{1,f\}$:
\begin{figure}[H]
    \centering
     \be 
\begin{split}
\begin{tikzpicture}
\begin{scope}[shift={(0,0)}, scale=1.8]
\foreach \x in {-4,-3,-2,-1,0}{
\draw [very thick,newblue,->-](\x,0) -- ({\x+1},0); 
}
\foreach \x in {-3,-2,-1,0}{
 \draw[fill=black] (\x,-0.5) circle (0.0556);
  \draw [very thick,->-](\x,-1) -- (\x,-0.5);
 \draw [very thick,->-](\x,-0.5) -- (\x,0);
}
 \node[below] at (-3, -1) {$R_0=1 \oplus f$};
 \node[below] at (-2, -1) {$R_0$};
 \node[below] at (-1, -1) {$R_0$};
 \node[below] at (-0, -1) {$R_0$};
\node[above] at (-3.5, 0) {$\sigma$} ;
\node[above] at (-2.5, 0) {$\sigma$} ;
\node[above] at (-1.5, 0) {$\sigma$} ;
\node[above] at (-0.5, 0) {$\sigma$} ;
\node[above] at (0.5, 0) {$\sigma$} ;
\node[left] at (-3, -0.25) {$t_i$} ;
\node[left] at (-2, -0.25) {$t_{i+1}$} ;
\node[left] at (-1, -0.25) {$t_{i+2}$} ;
\node[left] at (0, -0.25) {$t_{i+3}$} ;
\end{scope}
\end{tikzpicture}
\end{split}
\ee
    \caption{The nontrivial sector $\H_{\sigma}$. The basis for each on-site Hilbert space is $t_i\in\{1,f\}$. The black dots are unique projections.}
    \label{fig_anyon_chain_Ising_nontriv}
\end{figure}

The tensor product space $\cc{H}_{1,f}$ forms a representation of the integral fusion subcategory  $\C_0=\vc_{\Z_2}$:
\begin{equation}
    \D_f|\{x_i\}\>=|\{fx_i\}\>.
\end{equation}
On the other hand the operator $\D_\sigma$ maps $\cc{H}_{1,f}\to \cc{H}_\sigma$. We now calculate the operator $\D_\sigma$ explicitly in the basis $|\{x_i\}\>$ for $\H_{1,f}$ and $|\{t_i\}\>$ for $\H_\sigma$:
\begin{align}
\label{fig:SymActionIsing}
\begin{split}
\begin{tikzpicture}
\begin{scope}[shift={(0,0)}, scale=1.8]
\draw [very thick,newblue, ->-] (-4,0.6) -- (1,0.6);
 \node[above] at (-1, 0.65) {$\sigma$};
\draw [very thick,->-](-4,0) -- (-3,0); 
\draw [very thick,->-](-3,0) -- (-2,0);
\draw [very thick,->-](-2,0) -- (-1,0);
\draw [very thick,->-](-1,0) -- (0,0);
\draw [very thick,->-](0,0) -- (1,0);
 \draw[fill=black] (-3,0) circle (0.0556);
 \draw[fill=black] (-2,0) circle (0.0556);
 \draw[fill=black] (-1,0) circle (0.0556);
 \draw[fill=black] (0,0) circle (0.0556);
 \foreach \x in {-3,-2,-1,0} {
 \draw [very thick,->-](\x,-0.6) -- (\x,0); 
 }
 \node[below] at (-3, -0.6) {$R_0=1 \oplus f$};
 \node[below] at (-2,-0.6) {$R_0$};
 \node[below] at (-1,-0.6) {$R_0$};
 \node[below] at (-0, -0.6) {$R_0$};
\node[above] at (-3.5, 0) {$x_i$} ;
\node[above] at (-2.5, 0) {$x_{i+1}$} ;
\node[above] at (-1.5, 0) {$x_{i+2}$} ;
\node[above] at (-0.5, 0) {$x_{i+3}$} ;
\node[above] at (0.5, 0) {$x_{i+4}$} ;
\end{scope}
\end{tikzpicture}
\end{split} \\
=\quad 
\begin{split}
\begin{tikzpicture}
\begin{scope}[shift={(0,0)}, scale=1.8]
\draw[very thick,newblue,->-] (-4,0) -- (-3.35,0);
\draw[very thick,->-]         (-3.35,0) -- (-3,0);
\draw[very thick,->-]         (-3,0) -- (-2.65,0);
\draw[very thick,newblue,->-] (-2.65,0) -- (-2.35,0);
\draw[very thick,->-]         (-2.35,0) -- (-2,0);
\draw[very thick,->-]         (-2,0) -- (-1.65,0);
\draw[very thick,newblue,->-] (-1.65,0) -- (-1.35,0);
\draw[very thick,->-]         (-1.35,0) -- (-1,0);
\draw[very thick,->-]         (-1,0) -- (-0.65,0);
\draw[very thick,newblue,->-] (-0.65,0) -- (-0.35,0);
\draw[very thick,->-]         (-0.35,0) -- (0,0);
\draw[very thick,->-]         (0,0) -- (0.35,0);
\draw[very thick,newblue,->-] (0.35,0) -- (1,0);
\foreach \x in {-3,-2,-1,0}{
\draw[very thick,newblue,->-]         (\x,0.35) -- ({\x+0.01},0.35);
}
\foreach \x in {-3,-2,-1,0}{
  \draw[fill=black] (\x,0) circle (0.0556);
}
\foreach \x in {-3,-2,-1,0}{
  \draw[very thick,newblue] (\x-0.35,0) arc[start angle=180,end angle=0,radius=0.35];
  \node[above,newblue] at (\x,0.35) {$\sigma$};
}
 \foreach \x in {-3,-2,-1,0} {
 \draw [very thick,->-](\x,-0.6) -- (\x,0); 
 }
 \node[below] at (-3, -0.6) {$R_0$};
 \node[below] at (-2,-0.6) {$R_0$};
 \node[below] at (-1,-0.6) {$R_0$};
 \node[below] at (-0, -0.6) {$R_0$};
\node[above,newblue] at (-3.675,0) {$\sigma$};
\node[above,newblue] at (-2.5,0) {$\sigma$};
\node[above,newblue] at (-1.5,0) {$\sigma$};
\node[above,newblue] at (-0.5,0) {$\sigma$};
\node[above,newblue] at (0.675,0) {$\sigma$};
\node[below] at (-3.2,0) {$x_i$};
\node[below] at (-2.8,0) {$x_{i+1}$};
\node[below] at (-2.2,0) {$x_{i+1}$};
\node[below] at (-1.8,0) {$x_{i+2}$};
\node[below] at (-1.2,0) {$x_{i+2}$};
\node[below] at (-0.8,0) {$x_{i+3}$};
\node[below] at (-0.2,0) {$x_{i+3}$};
\node[below] at (0.2,0) {$x_{i+4}$};
\end{scope}
\end{tikzpicture}
\end{split}
\\
=\quad \sum_{\{t_i\}}  \prod_i \delta_{t_i, x_ix_{i+1}}
\begin{split}
\begin{tikzpicture}
\begin{scope}[shift={(0,0)}, scale=1.8]
\draw [very thick,newblue,->-](-4,0) -- (-3,0); 
\draw [very thick,newblue,->-](-3,0) -- (-2,0);
\draw [very thick,newblue,->-](-2,0) -- (-1,0);
\draw [very thick,newblue,->-](-1,0) -- (0,0);
\draw [very thick,newblue,->-](0,0) -- (1,0);
 \draw[fill=black] (-3,-0.5) circle (0.0556);
 \draw[fill=black] (-2,-0.5) circle (0.0556);
 \draw[fill=black] (-1,-0.5) circle (0.0556);
 \draw[fill=black] (0,-0.5) circle (0.0556);
 \draw [very thick,->-](-3,-1) -- (-3,-0.5);
 \draw [very thick,->-](-3,-0.5) -- (-3,0);
  \draw [very thick,->-](-2,-1) -- (-2,-0.5);
 \draw [very thick,->-](-2,-0.5) -- (-2,0);
  \draw [very thick,->-](-1,-1) -- (-1,-0.5);
 \draw [very thick,->-](-1,-0.5) -- (-1,0);
  \draw [very thick,->-](0,-1) -- (0,-0.5);
 \draw [very thick,->-](0,-0.5) -- (0,0);
 \node[below] at (-3, -1) {$R_0$};
 \node[below] at (-2, -1) {$R_0$};
 \node[below] at (-1, -1) {$R_0$};
 \node[below] at (-0, -1) {$R_0$};
\node[above] at (-3.5, 0) {$\sigma$} ;
\node[above] at (-2.5, 0) {$\sigma$} ;
\node[above] at (-1.5, 0) {$\sigma$} ;
\node[above] at (-0.5, 0) {$\sigma$} ;
\node[above] at (0.5, 0) {$\sigma$} ;
\node[left] at (-3, -0.25) {$t_i$} ;
\node[left] at (-2, -0.25) {$t_{i+1}$} ;
\node[left] at (-1, -0.25) {$t_{i+2}$} ;
\node[left] at (0, -0.25) {$t_{i+3}$} ;
\end{scope}
\end{tikzpicture}
\end{split}
\end{align}

where in the last step we used relation
    \be
\begin{tikzpicture}[baseline=(base)]
\begin{scope}[shift={(0,0)}, scale=1.8]
\coordinate (base) at (-3,-0.5);

\draw[very thick,newblue] (-3.65,0) -- (-3.35,0);
\draw[very thick]         (-3.35,0) -- (-2.65,0);
\draw[very thick,newblue] (-2.65,0) -- (-2.35,0);

\draw[fill=black] (-3,0) circle (0.02);

\draw[very thick,newblue] (-3.35,0) arc[start angle=180,end angle=0,radius=0.35];
\node[above,newblue] at (-3,0.35) {$\sigma$};

\draw[very thick,->-] (-3,-1) -- (-3,0);

\node[below] at (-3,-1) {$R_0$};
\node[above,newblue] at (-3.5,0) {$\sigma$};
\node[above,newblue] at (-2.5,0) {$\sigma$};
\node[below] at (-3.2,0) {$x_i$};
\node[below] at (-2.8,0) {$x_{i+1}$};
  \draw[fill=black] (-3,0) circle (0.0556);
\end{scope}
\end{tikzpicture}
\;=\;
\begin{tikzpicture}[baseline=(base)]
\begin{scope}[shift={(0,0)}, scale=1.8]
\coordinate (base) at (-3,-0.5);
\draw[very thick,newblue] (-3.65,0) -- (-3.35,0);
\draw[very thick]         (-3.35,0) -- (-2.65,0);
\draw[very thick,newblue] (-2.65,0) -- (-2.35,0);
\draw[very thick,newblue] (-3.35,0) arc[start angle=180,end angle=0,radius=0.35];
\draw[very thick,->-] (-3,-1) -- (-3,-0.5);
\draw[very thick,->-] (-3,-0.5) -- (-3,0);
\draw[fill=black] (-3,-0.5) circle (0.0556);
\draw[newred, dashed, very thick]  (-3.1,0.25) -- (-2.5,0.25) -- (-2.5, -0.2) -- (-3.1, -0.2)-- (-3.1, 0.25);
\node[below] at (-3,-1) {$R_0$};
\node[above,newblue] at (-3.5,0) {$\sigma$};
\node[above,newblue] at (-2.5,0) {$\sigma$};
\node[above,newblue] at (-3,0.35) {$\sigma$};
\node[below] at (-3.2,0) {$x_i$};
\node[below] at (-2.8,0) {$x_{i+1}$};
\node[left] at (-3,-0.3) {$x_i x_{i+1}$};
\end{scope}
\end{tikzpicture}\;=\;
\begin{tikzpicture}[baseline=(base)]
\begin{scope}[shift={(0,0)}, scale=1.8]
\coordinate (base) at (-3,-0.5);
\draw[very thick,newblue] (-3.65,0) -- (-2.35,0);
\draw[fill=black] (-3,-0.5) circle (0.0556);
\draw[very thick,->-] (-3,-1) -- (-3,-0.5);
\draw[very thick,->-] (-3,-0.5) -- (-3,0);
\node[below] at (-3,-1) {$R_0$};
\node[above,newblue] at (-3.5,0) {$\sigma$};
\node[above,newblue] at (-2.5,0) {$\sigma$};
\node[left] at (-3,-0.3) {$x_i x_{i+1}$};
\end{scope}
\end{tikzpicture}
    \ee
which is obtained by an $F$-move in the red rectangle region. This $F$-move is trivial since $F_\sigma^{\sigma st}=1$ for $s, t\in \{1,f\}$. We conclude that 
\begin{equation}
    \D_\sigma |\cdots,x_i, x_{i+1},x_{i+2},\cdots\>_{1,f}=\sum_{\{t_i\}}\prod_i\delta_{t_i, x_ix_{i+1}}|\{t_i\}\>_{\sigma}=|\cdots, x_ix_{i+1},x_{i+1}x_{i+2}, x_{i+2}x_{i+3},\cdots\>_\sigma
\end{equation}
where the subscripts of the states indicate the sectors they belong to. To return to the tensor product space $\H_{1,f}$, we introduce a sequential circuit $U_\sigma$ that acts on the anyon chain from below:  
\begin{figure}[H]       
\be 
\begin{tikzpicture}
\begin{scope}[shift={(0,0)}, scale=1.7]
 \draw [very thick,->-](-3,1) -- (-3,1.6); 
  \draw [very thick,->-](-2,1) -- (-2,1.6); 
   \draw [very thick,->-](-1,1) -- (-1,1.6); 
 \draw [very thick,->-](0,1) -- (0,1.6); 
 \foreach \x in {-3,-2,-1,0} {
  \draw [very thick, newblue, ->-](\x,0) -- (\x,1); 
 }
 \foreach \x in {-4,-3,-2,-1,0}{
 \draw [very thick,newblue, ->-](\x,0) -- ({\x+1},1); 
 }
\node[above] at (-3.6, 0.5) {$\sigma$};
\node[above] at (-2.6, 0.5) {$\sigma$};
\node[above] at (-1.6, 0.5) {$\sigma$};
\node[above] at (-0.6, 0.5) {$\sigma$};
\node[above] at (0.4, 0.5) {$\sigma$};
 \draw [very thick,->-](-3,-0.6) -- (-3,0); 
  \draw [very thick,->-](-2,-0.6) -- (-2,0); 
   \draw [very thick,->-](-1,-0.6) -- (-1,0); 
 \draw [very thick,->-](0,-0.6) -- (0,0); 
 \draw[fill=blue,draw=black, very thick] (-3,1) circle (0.0556);
 \draw[fill=blue,draw=black, very thick] (-2,1) circle (0.0556);
 \draw[fill=blue,draw=black, very thick] (-1,1) circle (0.0556);
 \draw[fill=blue,draw=black, very thick] (0,1) circle (0.0556);
 \draw[fill=red, draw=black, very thick] (-3,0) circle (0.0556);
 \draw[fill=red, draw=black, very thick] (-2,0) circle (0.0556);
 \draw[fill=red, draw=black, very thick] (-1,0) circle (0.0556);
 \draw[fill=red, draw=black, very thick] (0,0) circle (0.0556);
 \node[left] at (-3, 0.4) {$\sigma$};
  \node[left] at (-2, 0.4) {$\sigma$};
   \node[left] at (-1, 0.4) {$\sigma$};
    \node[left] at (-0, 0.4) {$\sigma$};
 \node[below] at (-3, -0.6) {$R_0=1 \oplus f$};
 \node[below] at (-2, -0.6) {$R_0$};
 \node[below] at (-1, -0.6) {$R_0$};
 \node[below] at (-0, -0.6) {$R_0$};
\node[above] at (-3, 1.6) {$R_0$} ;
\node[above] at (-2, 1.6) {$R_0$} ;
\node[above] at (-1, 1.6) {$R_0$} ;
\node[above] at (0, 1.6) {$R_0$} ;
\end{scope}
\end{tikzpicture}
\ee
\caption{A sequential circuit $U_\sigma$ that maps $\cc{H}_\sigma\to \cc{H}_{1,f}$. \label{fig_sequential_Ising}}
    \end{figure}
The red and blue dots are unitary isomorphisms $i_1\oplus i_f:1\oplus f\to \sigma\otimes \sigma,~p_1\oplus p_f:\sigma\otimes \sigma \to 1\oplus f$. Therefore $U_\sigma$ is an isomorphism  $R_0^{\otimes L}\to R_0^{\otimes L}$, where $L$ is the length of the anyon chain. Since the diagonal lines of the sequential circuit are $\sigma$, we expect the circuit to act as $\H_{\sigma}\to \H_{1,f}$. We now derive the action of the circuit $U_\sigma$ on the sector $\H_\sigma$: 

\begin{align}
\begin{split}
\begin{tikzpicture}
\begin{scope}[shift={(0,0)}, scale=1.7]
 \foreach \x in {-4,-3,-2,-1,0}{
 \draw [very thick,newblue,->-](\x,0) -- ({\x+1},1);
   \draw [very thick,newblue,->-](\x,1.6) -- ({\x+1},1.6); 
   \node[above] at ({\x+0.4}, 0.5) {$\sigma$};
 }
  \foreach \x in {-3,-2,-1,0}{
  \draw [very thick,newblue,->-](\x,0) -- (\x,1); 
   \draw [very thick,->-](\x,1) -- (\x,1.6); 
  \draw [very thick,->-](\x,-0.6) -- (\x,0);
   \node[left] at (\x, 0.4) {$\sigma$};
    \draw[fill=red, draw=black, very thick] (\x,0) circle (0.0556);
 }
 \node[below] at (-3, -0.6) {$R_0=1 \oplus f$};
 \node[below] at (-2, -0.6) {$R_0$};
 \node[below] at (-1, -0.6) {$R_0$};
 \node[below] at (-0, -0.6) {$R_0$};
\foreach \x in {-3.5,-2.5,-1.5,-0.5, 0.5}{
  \node[above] at (\x,1.6) {$\sigma$};
}
  \node[left] at (-3,1.3) {$t_i$};
    \node[left] at (-2,1.3) {$t_{i+1}$};
      \node[left] at (-1,1.3) {$t_{i+2}$};
        \node[left] at (0,1.3) {$t_{i+3}$};
\foreach \x in {0,1,2,3}{
  \draw[dashed, very thick, newred]
    ({-3.3+\x},1.8) -- ({-3.3+\x},0.8) --
    ({-2.7+\x},0.8) -- ({-2.7+\x},1.8) --
    ({-3.3+\x},1.8);
}
\end{scope}
\end{tikzpicture}
\end{split}\label{eq_ising_circuit_1}\\
\vspace{2cm}=
\sum_{\{x_i\}}\prod_i H_{t_i x_i}
\begin{split}
\begin{tikzpicture}
\begin{scope}[shift={(0,0)}, scale=2]
\draw [very thick,->-](-4,1.6) -- (-3.35,1.6); 
\draw [very thick,newblue,->-](-3.35,1.6) -- (-2.65,1.6); 
\draw [very thick,->-](-2.65,1.6) -- (-2.35,1.6);
\draw [very thick,newblue,->-](-2.35,1.6) -- (-1.65,1.6);
\draw [very thick,->-](-1.65,1.6) -- (-1.35,1.6);
\draw [very thick,newblue,->-](-1.35,1.6) -- (-0.65,1.6);
\draw [very thick,->-](-0.65,1.6) -- (0,1.6);
\foreach \x in {-3,-2,-1}{
  \draw[very thick,newblue,->-] (\x,1.25) arc[start angle=270,end angle=180,radius=0.35];
  \draw[very thick,newblue,->-] (\x,1.25) arc[start angle=270,end angle=360,radius=0.35];
  \node[below left]  at ({\x-0.25},1.42) {$\sigma$};
  \node[below right] at ({\x+0.25},1.42) {$\sigma$};
}
\node[above] at (-3.5,1.6) {$x_i$};
\node[above] at (-2.5,1.6) {$x_{i+1}$};
\node[above] at (-1.5,1.6) {$x_{i+2}$};
\node[above] at (-0.5,1.6) {$x_{i+3}$};
\foreach \x in {-3,-2,-1}{
  \draw[very thick,->-] (\x,0.8) -- (\x,1.25); 
  \node[above] at (\x,1.6) {$\sigma$};
  \node[below] at (\x,0.8) {$R_0$}; 
     \draw[fill=red, draw=black, very thick] (\x,1.25) circle (0.0556);
}
\end{scope}
\end{tikzpicture}
\end{split}
\label{eq_ising_circuit_2}\\
=\sum_{\{x_i\}}\prod_i H_{t_i x_i} 
\begin{split}
\begin{tikzpicture}
\begin{scope}[shift={(0,0)}, scale=2]
\foreach \x in {-4,-3,-2,-1}{
\draw [very thick,->-](\x,0) -- ({\x+1},0);
}
 \foreach \x in {-3,-2,-1}{
 \draw [very thick,->-](\x,-0.5) -- (\x,0); 
  \node[below] at (\x, -0.5) {$R_0$};
    \draw[fill=black] (\x,0) circle (0.0556);
 }
\node[above] at (-3.5, 0) {$x_i$} ;
\node[above] at (-2.5, 0) {$x_{i+1}$} ;
\node[above] at (-1.5, 0) {$x_{i+2}$} ;
\node[above] at (-0.5, 0) {$x_{i+3}$} ;
\end{scope}
\end{tikzpicture}
\end{split}\label{eq_ising_circuit_3}
\end{align} 

Here in the first step we applied $F$-moves in the regions indicated by the red rectangles and used $(F_\sigma^{\sigma\sigma\sigma})_{st}=H_{st}$ where $H$ is the Hadamard gate. In the second step we used the following relation:
\be
    \begin{split}
\begin{tikzpicture}
\begin{scope}[shift={(0,0)}, scale=2]
\draw [very thick,->-](-3.65,1.6) -- (-3.35,1.6); 
\draw [very thick,newblue,->-](-3.35,1.6) -- (-2.65,1.6); 
\draw [very thick,->-](-2.65,1.6) -- (-2.35,1.6);
\foreach \x in {-3}{
  \draw[very thick,newblue,->-] (\x,1.25) arc[start angle=270,end angle=180,radius=0.35];
  \draw[very thick,newblue,->-] (\x,1.25) arc[start angle=270,end angle=360,radius=0.35];
  \node[below left]  at ({\x-0.25},1.42) {$\sigma$};
  \node[below right] at ({\x+0.25},1.42) {$\sigma$};
}
\node[above] at (-3.5,1.6) {$x_i$};
\node[above] at (-2.5,1.6) {$x_{i+1}$};
\foreach \x in {-3}{
  \draw[very thick,->-] (\x,0.8) -- (\x,1.25); 
  \node[above] at (\x,1.6) {$\sigma$};
  \node[below] at (\x,0.8) {$R_0$}; 
}
  \draw[fill=red, draw=black, very thick] (-3,1.25) circle (0.0556);
\draw[red, dashed, very thick] (-3.5,1.7)--(-2.5,1.7)--(-2.5,1.5)--(-3.5,1.5)--(-3.5,1.7);
\end{scope}
\end{tikzpicture}
\end{split} 
\xrightarrow{F_{x_{i+1}}^{x_i\sigma\sigma }=1}
    \begin{split}
\begin{tikzpicture}
\begin{scope}[shift={(0,0)}, scale=2]
\draw [very thick,->-](-3.65,1.6) -- (-3,1.6); 
\draw [very thick,->-](-3,1.6) -- (-2.35,1.6);
\draw [very thick,->-](-3,1.35) -- (-3,1.6);
  \draw[very thick,newblue,->-] (-3,1.05) arc[start angle=270,end angle=90,radius=0.15];
    \draw[very thick,newblue,->-] (-3,1.05) arc[start angle=-90,end angle=90,radius=0.15];
\node[above] at (-3.5,1.6) {$x_i$};
\node[above] at (-2.5,1.6) {$x_{i+1}$};
\foreach \x in {-3}{
  \draw[very thick,->-] (\x,0.8) -- (\x,1.05); 
  \node[left] at (\x,1.4) {$x_ix_{i+1}$};
  \node[below] at (\x,0.8) {$R_0$}; 
}
  \draw[fill=red, draw=black, very thick] (-3,1.05) circle (0.0556);
\end{scope}
\end{tikzpicture}
\end{split} 
\;=\; 
    \begin{split}
\begin{tikzpicture}
\begin{scope}[shift={(0,0)}, scale=2]
\draw [very thick,->-](-3.65,1.6) -- (-3,1.6); 
\draw [very thick,->-](-3,1.6) -- (-2.35,1.6);
\draw [very thick,->-](-3,0.8) -- (-3,1.6);
\node[above] at (-3.5,1.6) {$x_i$};
\node[above] at (-2.5,1.6) {$x_{i+1}$};
  \node[below] at (-3,0.8) {$R_0$}; 
    \draw[fill=black] (-3,1.6) circle (0.0556);
\end{scope}
\end{tikzpicture}
\end{split} 
\ee
We conclude
\begin{equation}
    U_\sigma|\{t_i\}\>_\sigma=\sum_{\{x_i\}}\prod_i H_{t_ix_i} |\{x_i\}\>_{1,f}\,.
\end{equation}
We then define a modified symmetry operator $\sf{D}_\sigma:=U_\sigma\circ \D_\sigma:\H_{1,f}\to \H_{1,f}$, 
\begin{equation}\label{eq_sf_sigma_action}
    \sf{D}_\sigma|\cdots, x_i, x_{i+1}, x_{i+2},\cdots\>_{1,f}=\prod_k H_k |\cdots, x_ix_{i+1}, x_{i+1}x_{i+2}, x_{i+2}x_{i+3},\cdots\>_{1,f}.
\end{equation}
with the Hadamard gates $H_k$ acting on site $k$. Notice that this precisely agrees with the Kramers-Wannier
duality operator defined in various papers in the literature~\cite{Li:2023ani,Tantivasadakarn:2021vel,Aasen:2016dop}. 

Our derivation of this operator provides a new perspective: The Kramers-Wannier
duality operator can be divided into the product of two operators, one is $\D_\sigma$ which is defined by the fusing with anyon string as in the standard anyon chain model. The other operator is the product of the Hadamard gates, which comes from the action of a sequential circuit from below the anyon chain and is designed to map the image of $\D_\sigma$ back to the tensor-product Hilbert space $\H_{1,f}$. Since actions from above the anyon chain commute with actions from below the anyon chain, we have 
\begin{equation}
    \sf{D}_\sigma \sf{D}_\sigma=U_\sigma^2\times\D_\sigma \D_\sigma=U_\sigma^2(\D_1+\D_f). \label{eq_Ising_lattice_fusion_rule}
\end{equation}

Given the expression \Cref{eq_sf_sigma_action}, it is known that $U_\sigma^2=T$ is translation by one lattice site to the left. Here we provide a diagrammatic calculation. Observe that the sequential circuit $U_\sigma$ may be redrawn as follows:
\be 
\begin{split}
\begin{tikzpicture}
\begin{scope}[shift={(0,0)}, scale=1]
\node[left] at (-4.5,1.5) {$\cdots$};
\node[right] at (4,1.5) {$\cdots$};
\draw[very thick,newblue, ->-] 
  (-4.2,0) .. controls (-4.2,0.7) and (-3.8,1.2) .. (-3.3,1.6)
  .. controls (-2.9,1.9) and (-2.8,2.3) .. (-2.8,3.0);
\draw[very thick,newblue, ->-] (-2.5,0) -- (-2.5,3.0);
\draw[very thick,newblue, ->-] 
  (-2.1,0) .. controls (-2.1,0.7) and (-1.7,1.2) .. (-1.2,1.6)
  .. controls (-0.8,1.9) and (-0.7,2.3) .. (-0.7,3.0);
\draw[very thick,newblue, ->-] (-0.4,0) -- (-0.4,3.0);
\draw[very thick,newblue, ->-] 
  (0.0,0) .. controls (0.0,0.7) and (0.4,1.2) .. (0.9,1.6)
  .. controls (1.3,1.9) and (1.4,2.3) .. (1.4,3.0);
\draw[very thick,newblue, ->-] (1.7,0) -- (1.7,3.0);
\draw[very thick,newblue, ->-] 
  (2.1,0) .. controls (2.1,0.7) and (2.5,1.2) .. (3.0,1.6)
  .. controls (3.4,1.9) and (3.5,2.3) .. (3.5,3.0);
\node[left] at (-3.55,1.55) {$\sigma$};
\node[left] at (-1.45,1.55) {$\sigma$};
\node[left] at (0.65,1.55) {$\sigma$};
\node[left] at (2.75,1.55) {$\sigma$};
\node[above] at (-2.5,3.0) {$\sigma$};
\node[below] at (-2.5,0) {$\sigma$};
\node[above] at (-0.4,3.0) {$\sigma$};
\node[below] at (-0.4,0) {$\sigma$};
\node[above] at (1.7,3.0) {$\sigma$};
\node[below] at (1.7,0) {$\sigma$};
\node[above] at (-2.8,3.0) {$\sigma$};
\node[below] at (-2.2,0) {$\sigma$};
\node[above] at (-0.7,3.0) {$\sigma$};
\node[below] at (-0.1,0) {$\sigma$};
\node[above] at (1.4,3.0) {$\sigma$};
\node[below] at (2,0) {$\sigma$};
\end{scope}
\end{tikzpicture}
\end{split}
\ee
where we have chosen an isomorphism $R_0=1 \oplus f\simeq \sigma \otimes \sigma$. Then $U_{\sigma}^2$ can be drawn as the following sequential circuit, which shows that $U_{\sigma}^2=T$ is translation by one lattice site to the left:
\begin{figure}[H]
\centering
\begin{tikzpicture}
\begin{scope}[shift={(0,0)}]
\node[left] at (-5,3) {$\cdots$};
\node[right] at (4,3) {$\cdots$};
\draw[newblue, very thick,->-] (-4.6,0) -- (-4.6,3.0);
\draw[newblue, very thick,->-] 
  (-4.3,0) .. controls (-4.3,0.7) and (-3.8,1.2) .. (-3.3,1.6)
  .. controls (-2.9,1.9) and (-2.8,2.3) .. (-2.8,3.0);

\draw[newblue, very thick,->-] (-2.5,0) -- (-2.5,3.0);

\draw[newblue, very thick,->-] 
  (-2.2,0) .. controls (-2.2,0.7) and (-1.7,1.2) .. (-1.2,1.6)
  .. controls (-0.8,1.9) and (-0.7,2.3) .. (-0.7,3.0);

\draw[newblue, very thick,->-] (-0.4,0) -- (-0.4,3.0);

\draw[newblue, very thick,->-] 
  (-0.1,0) .. controls (-0.1,0.7) and (0.4,1.2) .. (0.9,1.6)
  .. controls (1.3,1.9) and (1.4,2.3) .. (1.4,3.0);

\draw[newblue, very thick,->-] (1.7,0) -- (1.7,3.0);

\draw[newblue, very thick,->-] 
  (2.0,0) .. controls (2.0,0.7) and (2.5,1.2) .. (3.0,1.6)
  .. controls (3.4,1.9) and (3.5,2.3) .. (3.5,3.0);

\draw[newblue, very thick,->-] (3.5,3) -- (3.5,6);
\node[left] at (-3.55,1.55) {$\sigma$};
\node[left] at (-1.45,1.55) {$\sigma$};
\node[left] at (0.65,1.55) {$\sigma$};
\node[left] at (2.75,1.55) {$\sigma$};
\foreach \x in {-4.6,-4.3,-2.5,-0.4, 1.7, -2.2,-0.1,2.0}{
\node[below] at (\x,0) {$\sigma$};
}
\draw[newblue, very thick,->-] 
  (-4.3-0.3,3.0) .. controls (-4.3-0.3,3.7) and (-3.8-0.3,4.2) .. (-3.3-0.3,4.6)
  .. controls (-2.9-0.3,4.9) and (-2.8-0.3,5.3) .. (-2.8-0.3,6.0);

\draw[newblue, very thick,->-] (-2.5-0.3,3.0) -- (-2.5-0.3,6.0);

\draw[newblue, very thick,->-] 
  (-2.2-0.3,3.0) .. controls (-2.2-0.3,3.7) and (-1.7-0.3,4.2) .. (-1.2-0.3,4.6)
  .. controls (-0.8-0.3,4.9) and (-0.7-0.3,5.3) .. (-0.7-0.3,6.0);

\draw[newblue, very thick,->-] (-0.4-0.3,3.0) -- (-0.4-0.3,6.0);

\draw[newblue, very thick,->-] 
  (-0.1-0.3,3.0) .. controls (-0.1-0.3,3.7) and (0.4-0.3,4.2) .. (0.9-0.3,4.6)
  .. controls (1.3-0.3,4.9) and (1.4-0.3,5.3) .. (1.4-0.3,6.0);

\draw[newblue, very thick,->-] (1.7-0.3,3.0) -- (1.7-0.3,6.0);

\draw[newblue, very thick,->-] 
  (2.0-0.3,3.0) .. controls (2.0-0.3,3.7) and (2.5-0.3,4.2) .. (3.0-0.3,4.6)
  .. controls (3.4-0.3,4.9) and (3.5-0.3,5.3) .. (3.5-0.3,6.0);

\node[left] at (-3.55-0.3,4.55) {$\sigma$};
\node[left] at (-1.45-0.3,4.55) {$\sigma$};
\node[left] at (0.65-0.3,4.55) {$\sigma$};
\node[left] at (2.75-0.3,4.55) {$\sigma$};

\node[above] at (-2.5-0.3,6.0) {$\sigma$};
\node[above] at (-0.4-0.3,6.0) {$\sigma$};
\node[above] at (1.7-0.3,6.0) {$\sigma$};
\node[above] at (-2.8-0.3,6.0) {$\sigma$};
\node[above] at (-0.7-0.3,6.0) {$\sigma$};
\node[above] at (1.4-0.3,6.0) {$\sigma$};
\node[above] at (3.5-0.3,6.0) {$\sigma$};
\node[above] at (3.8-0.3,6.0) {$\sigma$};
\draw[newred, dashed, thick] (-2.65, -0.4) -- (-2,-0.4) -- (-2, 0.1) -- (-2.65,0.1) -- (-2.65, -0.4); 
\draw[newred, dashed, thick] (-1.15, 5.95) -- (-0.5,5.95) -- (-0.5, 6.45) -- (-1.15, 6.45) -- (-1.15, 5.95); 
\end{scope}
\end{tikzpicture}
\caption{The circuit $U_{\sigma}^2$. It is clear that $U_{\sigma}^2=T$ is translation by a lattice site to the left. The red rectangles show a site before and after translation. }
\label{fig_Ug_two_layers}
\end{figure}
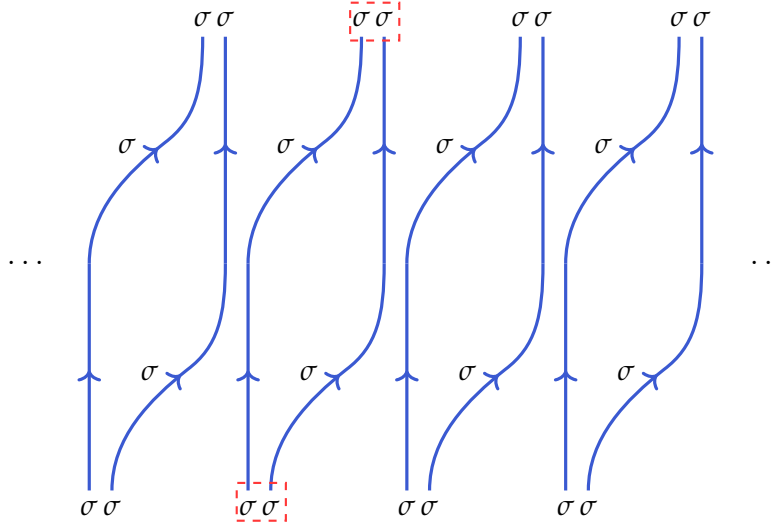
We conclude that the Ising fusion category can indeed be realized  on a tensor-product Hilbert space with on-site dimension 2, up to lattice translation. The mixing with lattice translation now acquires a new interpretation: it comes from a sequential circuit that acts on the anyon chain from below. 

\paragraph{General Strategy.}
This example suggests the following strategy for realizing a general weakly integral fusion categorical symmetry on a tensor-product Hilbert space: 
\begin{itemize}
    \item Use the regular object $R_0$ to build an anyon chain with Hilbert space $\cc{H}=\oplus_{g\in E} \cc{H}_g$, with each $\H_g$ being a tensor-product Hilbert space. 
    \item For each $g\in E$, construct a sequential circuit $U_g$ acting on the anyon chain from below such that it defines an isomorphism $\cc{H}_g\simeq \cc{H}_0$. 
    \item Then for each simple object $x\in \C_g$, define a modified symmetry operator $\sf{D}_x:=U_g\circ \D_x: \cc{H}_0\to \cc{H}_0$. Then we have modified fusion rules for $x\in \C_g,~y\in \C_h$:
    \begin{equation}
        \sf{D}_x\circ \sf{D}_y=U_g\circ U_h\circ U_{gh}^{-1} \sum_{z\in \C_{gh}} N^z_{x,y}\sf{D}_z=U_{(g,h)} \sum_{z\in \C_{gh}} N^z_{x,y}\sf{D}_z
    \label{eq: lattice fusion rules}
    \end{equation}
    where $U_{(g,h)}:=U_g\circ U_h\circ U_{gh}^{-1}: \H_0\to \H_0$ is now a QCA on $\H_0$. By construction, the QCA refinements $U_{(g,h)}$ only depend on the sector labels $g,h$, and not specific anyons in the sectors. According to our~\Cref{thm_main},  we  expect 
    \begin{equation}
        \ind(U_{(g,h)})=\sqrt{\frac{n_gn_h}{n_{gh}}} \frac{m_gm_h}{m_{gh}},
    \label{eq: q expectation}
    \end{equation}
    for some rational numbers $m_g$ that also only depend on the sectors. 
\end{itemize}

We now return to the construction for a general weakly integral fusion category following this strategy. We will then compute the indices of the QCAs $U_{(g,h)}$ in~\Cref{sec_index_computation} and confirm that the above expectation is indeed true.

\subsubsection{General Weakly Integral Categories}
To complete the construction, it suffices to construct the sector-change sequential circuits $U_g: \H_g\to \H_0$. Below we will provide two constructions. The first one applies to a subclass of weakly integral fusion categories and has on-site Hilbert space dimension $\dim(\C_0)$. The second one applies to arbitrary weakly integral fusion categories but has a larger on-site dimension of $\dim(\C_0)^2$.

We begin by defining certain objects in a general weakly integral fusion category. Let $\C=\oplus_{g\in E} \C_g$ be a weakly integral fusion category, and for every simple $x$ we have $d_x=a_x\sqrt{n_g}$ as before.  Define for every $g\in E$ an object 
\begin{equation}\label{Rgdef}
    R_g:=\bigoplus_{x\in \irr(\C_g)}a_x x
\end{equation}
which is the $g$-sector analog of the regular object $R_0$. It
 has quantum dimension $d_{R_g}=\sqrt{n_g}\sum_{x\in \irr(\C_g)}a_x^2$. For $x\in \C_g$  we have 
\begin{align}
    x\otimes R_h&=\bigoplus_{y\in \irr(\C_h)}a_y x\otimes y=\bigoplus_{y\in \irr(\C_h), z\in \irr(\C_{gh})}a_y N^z_{xy}z=\bigoplus_{y\in \irr(\C_h), z\in \irr(\C_{gh})}a_y N^y_{x^*z}z\nonumber\\
    &= \bigoplus_{z\in \irr(\C_{gh})} \frac{1}{\sqrt{n_h}} d_xd_z z=\sqrt{\frac{n_g n_{gh}}{n_h}}a_x R_{gh}\,.
\end{align}
One can also check that $R_h\otimes x=x\otimes R_h$. In particular we have 
\begin{align}\label{xReqs}
    &x\otimes R_g=R_g\otimes x=a_x R_0,\qquad x\otimes R_0=R_0\otimes x=n_g a_x R_g,\\
    &
    R_g\otimes R_g=\bigoplus_{x\in \irr(\C_g)}a_x x\otimes R_g=\bigoplus_{x\in \irr(\C_g)} a_x^2 R_0=\frac{\dim(\C_0)}{n_g}R_0\,.
\end{align}
\begin{construction} \label{construction_special_case}
    Take $\rho=R_0$ and build an anyon chain as before. Assume that for each $g\in E$, there exists an object $x_g\in \C_g$ with $d_{x_g}=\sqrt{n_g}$ (i.e., $a_{x_g}=1$). Then we define a sequential circuit $U_g$ as follows. 
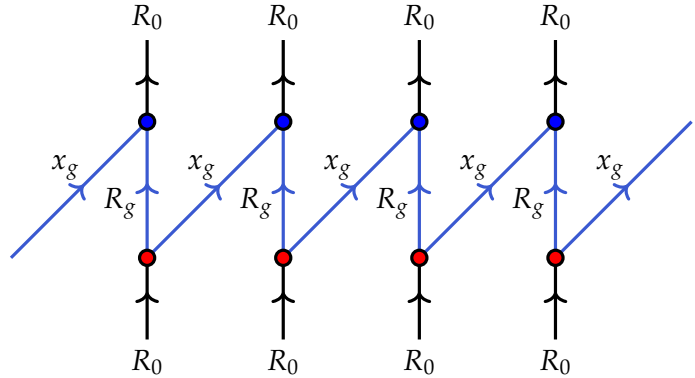
\begin{figure}[H]       
$$
\begin{tikzpicture}
\begin{scope}[shift={(0,0)}, scale=1.8]
 \draw [very thick,->-](-3,1) -- (-3,1.6); 
  \draw [very thick,->-](-2,1) -- (-2,1.6); 
   \draw [very thick,->-](-1,1) -- (-1,1.6); 
 \draw [very thick,->-](0,1) -- (0,1.6); 
 \foreach \x in {-3,-2,-1,0}{
  \draw [very thick,newblue,->-](\x,0) -- (\x,1); 
 }
 \foreach \x in {-4,-3,-2,-1,0} {
  \draw [very thick,newblue,->-](\x,0) -- ({\x+1},1); 
 }
\node[above] at (-3.6, 0.5) {$x_g$};
\node[above] at (-2.6, 0.5) {$x_g$};
\node[above] at (-1.6, 0.5) {$x_g$};
\node[above] at (-0.6, 0.5) {$x_g$};
\node[above] at (0.4, 0.5) {$x_g$};
 \draw [very thick,->-](-3,-0.6) -- (-3,0); 
  \draw [very thick,->-](-2,-0.6) -- (-2,0); 
   \draw [very thick,->-](-1,-0.6) -- (-1,0); 
 \draw [very thick,->-](0,-0.6) -- (0,0); 
 \draw[fill=blue,draw=black, very thick] (-3,1) circle (0.0556);
 \draw[fill=blue,draw=black, very thick] (-2,1) circle (0.0556);
 \draw[fill=blue,draw=black, very thick] (-1,1) circle (0.0556);
 \draw[fill=blue,draw=black, very thick] (0,1) circle (0.0556);
 \draw[fill=red, draw=black, very thick] (-3,0) circle (0.0556);
 \draw[fill=red, draw=black, very thick] (-2,0) circle (0.0556);
 \draw[fill=red, draw=black, very thick] (-1,0) circle (0.0556);
 \draw[fill=red, draw=black, very thick] (0,0) circle (0.0556);
 \node[left] at (-3, 0.4) {$R_g$};
  \node[left] at (-2, 0.4) {$R_g$};
   \node[left] at (-1, 0.4) {$R_g$};
    \node[left] at (-0, 0.4) {$R_g$};
 \node[below] at (-3, -0.6) {$R_0$};
 \node[below] at (-2, -0.6) {$R_0$};
 \node[below] at (-1, -0.6) {$R_0$};
 \node[below] at (-0, -0.6) {$R_0$};
\node[above] at (-3, 1.6) {$R_0$} ;
\node[above] at (-2, 1.6) {$R_0$} ;
\node[above] at (-1, 1.6) {$R_0$} ;
\node[above] at (0, 1.6) {$R_0$} ;
\end{scope}
\end{tikzpicture}
$$   
\caption{A unitary sequential circuit $U_g$ that maps $\cc{H}_g\to \cc{H}_0$. \label{fig:sequentialc_1}}
    \end{figure}
    
Notice that since $a_{x_g}=1$, we have $R_g\otimes x_g\simeq R_0\simeq x_g\otimes R_g$ by~\Cref{xReqs}. We then take the red/blue dots to be unitary isomorphisms. Then clearly $U_g$ is a unitary. Since the diagonal lines are $x_g\in \C_g$, this circuit maps $\cc{H}_g\to \cc{H}_0$. The condition of this construction is satisfied by Metaplectic categories $SO(M)_2$ for odd square-free $M$ and Tambara-Yamagami categories $\TY(A)$ for $|A|$ square free. 
\end{construction}
\begin{remark}
   For a general weakly integral fusion category, we can replace $R_0$ by $NR_0$ where $N$ is the minimal positive integer such that for any $g\in E$, there exists $x_g\in \C_g$ with $a_{x_g}|N$. Then we can replace the vertical $R_g$ by $(N/a_{x_g}) R_g$. This works because we have isomorphisms $(N/a_{x_g}) R_g\otimes x_g\simeq NR_0\simeq x_g\otimes (N/a_{x_g}) R_g$. Alternatively, we have the following more uniform construction.
\end{remark}
\begin{construction} \label{construction_general_case}
    Now we do not make any assumption about the weakly integral fusion category. We take $\rho=\dim(\C_0)R_0$. This will lead to an anyon chain with Hilbert space $\cc{H}=\oplus_{g\in E}\cc{H}_g$ and $\cc{H}_g$ is a tensor product space with on-site dimension $\dim(\C_0)^2$. We then define the sequential circuit as in \Cref{fig:sequentialc_2}. 
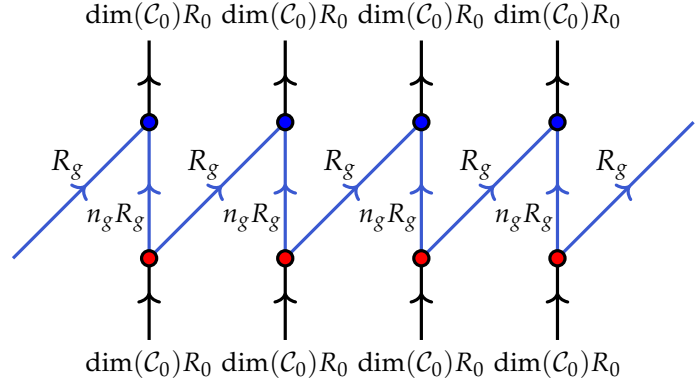
\begin{figure}[ht]       
$$
\begin{tikzpicture}
\begin{scope}[shift={(0,0)}, scale=1.8]
 \draw [very thick,->-](-3,1) -- (-3,1.6); 
  \draw [very thick,->-](-2,1) -- (-2,1.6); 
   \draw [very thick,->-](-1,1) -- (-1,1.6); 
 \draw [very thick,->-](0,1) -- (0,1.6); 

 \foreach \x in {-4,-3,-2,-1,0}{
 \draw [very thick,newblue,->-](\x,0) -- ({\x+1},1); 
 }
\node[above] at (-3.6, 0.5) {$R_g$};
\node[above] at (-2.6, 0.5) {$R_g$};
\node[above] at (-1.6, 0.5) {$R_g$};
\node[above] at (-0.6, 0.5) {$R_g$};
\node[above] at (0.4, 0.5) {$R_g$};
\foreach \x in {-3,-2,-1,0}{
  \draw [very thick,->-](\x,-0.6) -- (\x,0); 
  \draw [very thick,newblue, ->-](\x,0) -- (\x,1); 
    \draw[fill=blue,draw=black, very thick] (\x,1) circle (0.0556);
  \draw[fill=red, draw=black, very thick] (\x,0) circle (0.0556);
}
 \node[left] at (-2.95, 0.3) {\small $n_gR_g$};
  \node[left] at (-1.95, 0.3) {\small $n_gR_g$};
   \node[left] at (-0.95, 0.3) {\small $n_gR_g$};
    \node[left] at (0.05, 0.3) {\small $n_gR_g$};
 \node[below] at (-3, -0.6) {\small $\dim(\C_0)R_0$};
 \node[below] at (-2, -0.6) {\small $\dim(\C_0)R_0$};
 \node[below] at (-1, -0.6) {\small $\dim(\C_0)R_0$};
 \node[below] at (-0, -0.6) {\small $\dim(\C_0)R_0$};
\node[above] at (-3, 1.6) {\small $\dim(\C_0)R_0$} ;
\node[above] at (-2, 1.6) {\small $\dim(\C_0)R_0$} ;
\node[above] at (-1, 1.6) {\small $\dim(\C_0)R_0$} ;
\node[above] at (0, 1.6) {\small $\dim(\C_0)R_0$} ;
\end{scope}
\end{tikzpicture}
$$   
\caption{A unitary sequential circuit $U_g$ that maps $\cc{H}_g\to \cc{H}_0$. \label{fig:sequentialc_2}}
    \end{figure}
   The red and blue dots are isomorphisms
    \begin{equation}\label{eq_gamma_delta}
        \gamma_g: \dim(\C_0)R_0\to n_g R_g\otimes R_g, \qquad
        \delta_g: R_g\otimes n_gR_g\to \dim(\C_0) R_0.
    \end{equation}
 Then clearly $U_g$ is an isomorphism and maps $\cc{H}_g\to \cc{H}_0$.
  We can take $\gamma_g$ and $\delta_g$ to be unitaries so that $U_g$ is unitary.
\end{construction}
\begin{remark}
    Notice that the sequential circuits in both constructions take the form of "anyon translation". Namely, in \Cref{construction_special_case}, the anyon $x_g$ is translated to the left by one site (recall our circuit acts from below the anyon chain); in \Cref{construction_general_case}, the anyon $R_g$ is translated to the left by one site. This intuitive picture can be turned into a precise notion called generalized translation in \cite{Jones:2023imy}. We will discuss in more details the relation between our sequential circuits and generalized translation of \cite{Jones:2023imy} in \Cref{sec_relation_to_gen_translation}.
\end{remark}
\begin{remark}\label{rmk_flexible}
    The above two constructions are not the only possibilities. For specific examples one can be more flexible and construct other sequential circuits that do the job. The only requirement is that the anyon chain input is $\rho=N R_0$, i.e., some multiple of $R_0$ and there exist pairs $X_g, Y_g \in \C_g$ for every $g\in E$ such that $NR_0\simeq X_g\otimes Y_g\simeq Y_g\otimes X_g$, then one can take the sequential circuit $U_g$ to be ``translating $Y_g$ by one site". 
\end{remark}
\begin{remark}
    In fact the modified operators $\sf{D}_x=U_g\circ \D_x,~x\in \irr(\C_g)$ fix each component $\H_h$ in the direct sum $\H=\oplus_{h\in E} \H_h$. Hence more generally one could consider the restriction of the operators $\{\sf{D}_x: x\in \irr(\C)\}$ to the subspace $\H_h$, which would give a QCA-refined realization  on the tensor-product Hilbert space $\H_h$.
\end{remark}
\subsubsection{Some Technical Aspects of the Constructions}
Here we discuss some technical details of the general constructions above that are not visible in the simple Ising category example. We highlight some subtleties one needs to pay attention to if one wishes to concretely implement the constructions for more complicated weakly integral categories such as the Metaplectic categories.  We will follow \Cref{construction_special_case} for illustration.  

\paragraph{Choice of isomorphisms $\delta, \gamma$.} In the definition of the sequential circuit $U_g$, a pair of isomorphisms 
\begin{equation}
    \gamma_g: R_0\to R_g\otimes x_g,\quad \delta_g:  x_g\otimes R_g\to R_0. 
\end{equation}
must be chosen. Here we explain the meaning of making such a choice. We can expand these isomorphisms in a basis consisting of compositions of inclusion/projection of simples as below:
\begin{equation}
 \begin{split}
\begin{tikzpicture}
\begin{scope}[shift={(0,0)}, scale=2]
\node[left] at (0,0.5) {$\gamma_g$};
\draw[very thick, ->-] (0,0) -- (0,0.5);
\draw[very thick, ->-, newblue] (0,0.5) -- (0.5,1);
\draw[very thick,newblue, ->-] (0,0.5) -- (0,1);
 \draw[fill=red, draw=black, very thick] (0,0.5) circle (0.05); 
 \node[below] at (0,0){$R_0$};
 \node[above] at (0.5,1){$x_g$};
 \node[above] at (0,1) {$R_g$};
\end{scope}
\end{tikzpicture}
\end{split}=\sum_{t,s, \mu,\nu}(\gamma_g)^{st}_{\mu\nu} 
\begin{split}
\begin{tikzpicture}
\begin{scope}[scale=2.2]
\draw[very thick, ->-] (0,0) -- (0,0.25);
\draw[very thick, ->-] (0,0.25) -- (0,0.5);
\draw[very thick, ->-, newblue] (0,0.5) -- (0.5,1);
\draw[very thick,newblue, ->-] (0,0.5) -- (0,0.75);
\draw[very thick,newblue, ->-] (0,0.75) -- (0,1);
 \node[below] at (0,0){$R_0$};
 \node[above] at (0.5,1){$x_g$};
 \node[above] at (0,1) {$R_g$};
\draw[fill=white,draw=black, very thick] (0-0.04,0.25-0.04) rectangle ++ (0.08,0.08);
\node[left] at (0,0.23) {$p_s^\mu$};
\node[left] at (0, 0.4) {$s$};
\node[left] at (0, 0.6){$t$}; 
\draw[fill=white,draw=black, very thick] (0-0.04,0.75-0.04) rectangle ++ (0.08,0.08);
\node[left] at (0, 0.75) {$i_t^\nu$};
\end{scope}
\end{tikzpicture}
\end{split},
\quad 
    \begin{split}
\begin{tikzpicture}
\begin{scope}[shift={(0,0)}, scale=2]
\node[right] at (0,0.5) {$\delta_g$};
\draw[very thick,newblue, ->-] (0,0) -- (0,0.5);
\draw[very thick, ->-, newblue] (-0.5,0) -- (0,0.5);
\draw[very thick, ->-] (0,0.5) -- (0,1);
 \draw[fill=blue, draw=black, very thick] (0,0.5) circle (0.05); 
 \node[below] at (0,0){$R_g$};
 \node[below] at (-0.5,0){$x_g$};
 \node[above] at (0,1) {$R_0$};
\end{scope}
\end{tikzpicture}
\end{split}=\sum_{t,s, \mu,\nu}(\delta_g)^{st}_{\mu\nu} 
\begin{split}
\begin{tikzpicture}
\begin{scope}[scale=2.2]
\draw[very thick,newblue, ->-] (0,0) -- (0,0.25);
\draw[very thick,newblue, ->-] (0,0.25) -- (0,0.5);
\draw[very thick, ->-, newblue] (-0.5,0) -- (0,0.5);
\draw[very thick, ->-] (0,0.5) -- (0,0.75);
\draw[very thick, ->-] (0,0.75) -- (0,1);
 \node[below] at (0,0){$R_g$};
 \node[below] at (-0.5,0){$x_g$};
 \node[above] at (0,1) {$R_0$};
\draw[fill=white,draw=black, very thick] (0-0.04,0.25-0.04) rectangle ++ (0.08,0.08);
\node[right] at (0,0.25) {$p_t^\mu$};
\node[right] at (0, 0.4) {$t$};
\node[right] at (0, 0.6){$s$};
\draw[fill=white,draw=black, very thick] (0-0.04,0.75-0.04) rectangle ++ (0.08,0.08);
\node[right] at (0, 0.75) {$i_s^\nu$};
\end{scope}
\end{tikzpicture}
\end{split}
\end{equation}
Here $t\in R_g$ label simple summads of $R_g$, $s\in R_0$ label simple summads of $R_0$, and $i_t^\nu, p_s^{\mu}, i_s^\nu, p_t^\mu$ are  inclusions/projections. We have assumed the fusion category is multiplicity-free. The summations run over configurations such that the trivalent vertex Hilbert space is nonzero, i.e., $\Hom(s,t\otimes x_g)\neq 0$ for the left diagram and $\Hom(x_g\otimes t, s)\neq 0$ for the right diagram. 
The coefficients $(\gamma_g)^{st}_{\mu\nu}, (\delta_g)_{\mu\nu}^{st}$ are the choices here. They must be chosen so that $\gamma_g, \delta_g$ are unitaries. In general there is no natural choice for them, one simply needs to make one. Once the coefficients have been chosen, one can then  calculate the action of the sequential circuit using standard diagrammatics of unitary fusion category that only involve simple objects. 

\paragraph{Choice of on-site basis}
To express the symmetry operator $\sf{D}_x=U_g \circ \D_x$ explicitly as an operator acting on the tensor-product Hilbert space $\H_0$ requires choosing a basis for the on-site Hilbert space. The basis of the anyon chain Hilbert space that is natural for diagrammatic computation does not show the tensor product structure explicitly, hence one needs to perform a basis change. The anyon chain $(\C, \C, R_0)$ has a Hilbert space structure $\H=\oplus_g\H_g$. To illustrate the idea, let us consider the case $\H=\H_0\oplus \H_1$ is $\Z_2$-graded. The subspace $\H_0$ has a basis represented as below where the red rectangle indicates the basis for a single site.
\begin{equation}
\begin{split}
\begin{tikzpicture}
\begin{scope}[shift={(0,0)}, scale=3]
\draw [very thick,->-](-3.5,0) -- (-3,0); 
\draw [very thick,->-](-3,0) -- (-2,0);
\draw [very thick,->-](-2,0) -- (-1.5,0);
 \foreach \x in {-3,-2}{
  \node[below] at (\x, -0.3) {$R_0$};
  \draw [very thick,->-](\x,-0.3) -- (\x,0); 
   \draw[fill=black] (\x,0) circle (0.0333);
 }
\node[above] at (-3.5, 0) {$s_k$} ;
\node[above] at (-2.5, 0) {$s_{k+1}$} ;
\node[above] at (-1.5, 0) {$s_{k+2}$} ;
\node[above] at (-3, 0) {$|\mu_k^{s_k},\mu_k^{s_{k+1}}\>$};
\node[above] at (-2, 0) {$|\mu_{k+1}^{s_{k+1}},\mu_{k+1}^{s_{k+2}}\>$};
\draw[very thick, newred, dashed] (-3.05,0.25) -- (-2,0.25) -- (-2,0.02) -- (-3.05, 0.02) -- (-3.05, 0.25);
\end{scope}
\end{tikzpicture}
\end{split}\label{eq_basis_H_0}
\end{equation}
From \Cref{eq_dim_weakly_int} we know the Hilbert space at the vertex between $s_k$ and $s_{k+1}$ has dimension $d_{s_k}d_{s_{k+1}}$. Hence we can choose a basis with the following index structure:
\begin{equation}
    |\mu_k^{s_k},\mu_k^{s_{k+1}}\>: \mu_k^{s_k}=1,\cdots,d_{s_k};\mu_k^{s_{k+1}}=1,\cdots, d_{s_{k+1}}.
\end{equation}
The basis for a single site is then 
\begin{equation}
    |\mu_k^{s_{k+1}},s_{k+1},\mu_{k+1}^{s_{k+1}}\>:s_{k+1}\in \irr(\C_0); \mu_k^{s_{k+1}}=1,\cdots, d_{s_{k+1}};\mu_k^{s_{k+1}}=1,\cdots, d_{s_{k+1}},
\end{equation}
as indicated by the red rectangle in \Cref{eq_basis_H_0}. This is a basis that makes the tensor product structure of $\H_0$ transparent. However the natural basis for performing diagrammatic calculation is the following:
\begin{equation}
\begin{split}
\begin{tikzpicture}
\begin{scope}[shift={(0,0)}, scale=3]
\draw [very thick,->-](-3.5,0) -- (-3,0); 
\draw [very thick,->-](-3,0) -- (-2,0);
\draw [very thick,->-](-2,0) -- (-1.5,0);
 \foreach \x in {-3,-2}{
  \node[below] at (\x, -0.5) {$R_0$};
  \draw [very thick,->-](\x,-0.5) -- (\x,-0.25); 
  \draw [very thick,->-](\x,-0.25) -- (\x,0); 
   \draw[fill=white,draw=black, very thick] (\x-0.03,-0.25-0.03) rectangle ++ (0.06,0.06);
 }
\node[above] at (-3.5, 0) {$s_k$} ;
\node[above] at (-2.5, 0) {$s_{k+1}$} ;
\node[above] at (-1.5, 0) {$s_{k+2}$} ;
\node[right] at (-3,-0.25) {$\epsilon_k$};
\node[left] at (-3,-0.125) {$\xi_k$};
\node[right] at (-2,-0.25) {$\epsilon_{k+1}$};
\node[left] at (-2,-0.125) {$\xi_{k+1}$};
\end{scope}
\end{tikzpicture}
\end{split}\label{eq_basis_of_H_0}
\end{equation}
Here $\xi_k\in \irr(\C_0)$ are simple objects such that $s_{k+1}\in s_k\otimes \xi_k$ and 
\begin{equation}
    \epsilon_k=1,\cdots, d_{\xi_k}: R_0\to \xi_k
\end{equation}
label orthonormal projections. We again assumed multiplicity-free. We then have a natural basis $\{|\xi_k,\epsilon_k\>\}$ for $\Hom(s_k\otimes R_0, s_{k+1})$. There should be $d_{s_k}d_{s_{k+1}}$ many admissible values of $(\xi_k,\epsilon_k)$. We then need to specify how the previous basis $\{|\mu_k^{s_k},\mu_k^{s_{k+1}}\>\}$  is related to this basis. For instance, we  can simply choose the previous basis $\{|\mu_k^{s_k},\mu_k^{s_{k+1}}\>\}$ to be related to this projection basis by relabeling. That is, we take some bijection of sets 
\begin{equation}
    \gamma_{s_k, s_{k+1}}: [d_{s_k}]\times [d_{s_{k+1}}]\to \{(\xi_k,\epsilon_k) :s_{k+1}\in s_k\otimes \xi_k,\epsilon_k=1,\cdots,d_{\xi_k}\}
\end{equation}
and identify 
\begin{equation}
    |\gamma_{s_k, s_{k+1}}(\mu_k^{s_k},\mu_k^{s_{k+1}})\>=|\mu_k^{s_k},\mu_k^{s_{k+1}}\>.
\end{equation}
Similarly, the subspace $\H_1$ has a basis represented as below:
\begin{equation}
\begin{split}
\begin{tikzpicture}
\begin{scope}[shift={(0,0)}, scale=3]
\draw [very thick,->-](-3.5,0) -- (-3,0); 
\draw [very thick,->-](-3,0) -- (-2,0);
\draw [very thick,->-](-2,0) -- (-1.5,0);
 \foreach \x in {-3,-2}{
  \node[below] at (\x, -0.3) {$R_0$};
  \draw [very thick,->-](\x,-0.3) -- (\x,0); 
   \draw[fill=black] (\x,0) circle (0.0333);
 }
\node[below] at (-3.5, 0) {${t_k}$} ;
\node[below] at (-2.5, 0) {$t_{k+1}$} ;
\node[below] at (-1.5, 0) {$t_{k+2}$} ;
\node[above] at (-3, 0) {$|\nu_k^{t_k},\sigma_k,\nu_k^{t_{k+1}}\>$};
\node[above] at (-2, 0) {$|\nu_{k+1}^{t_{k+1}},\sigma_{k+1},\nu_{k+1}^{t_{k+2}}\>$};
\draw[very thick, newred, dashed] (-2.95,0.25) -- (-1.85,0.25) -- (-1.85,-0.2) -- (-2.95, -0.2) -- (-2.95, 0.25);
\end{scope}
\end{tikzpicture}
\end{split}\label{eq_basis_H_1}
\end{equation}
The Hilbert space at the vertex between $t_k$ and $t_{k+1}$ has dimension $a_{t_k}a_{t_{k+1}}n_1$. Hence we can choose a basis for $\Hom(t_k\otimes R_0,t_{k+1})$ with the  following index structure :
\begin{equation}
    |\nu_k^{t_k},\sigma_{k},\nu_k^{t_{k+1}}\>: \nu_k^{t}=1,\cdots,a_t; \sigma_{k}=1,\cdots,n_1.
\end{equation}
The basis for a single site is then indicated by the red rectangle region in \Cref{eq_basis_H_1}. Again the basis suitable for diagrammatic computation is given by projections:
\begin{equation}
\begin{split}
\begin{tikzpicture}
\begin{scope}[shift={(0,0)}, scale=3]
\draw [very thick,->-](-3.5,0) -- (-3,0); 
\draw [very thick,->-](-3,0) -- (-2,0);
\draw [very thick,->-](-2,0) -- (-1.5,0);
 \foreach \x in {-3,-2}{
  \node[below] at (\x, -0.5) {$R_0$};
  \draw [very thick,->-](\x,-0.5) -- (\x,-0.3); 
  \draw [very thick,->-](\x,-0.3) -- (\x,0); 
   \draw[fill=white,draw=black, very thick] (\x-0.03,-0.3-0.03) rectangle ++ (0.06,0.06);
 }
   \node[right] at (-3,-0.3) {$\eta_k$};
   \node[left] at (-3,-0.15) {$\xi_k$};
     \node[right] at (-2,-0.3) {$\eta_{k+1}$};
   \node[left] at (-2,-0.15) {$\xi_{k+1}$};
\node[above] at (-3.5, 0) {$t_k$} ;
\node[above] at (-2.5, 0) {$t_{k+1}$} ;
\node[above] at (-1.5, 0) {$t_{k+2}$} ;
\end{scope}
\end{tikzpicture}
\end{split}\label{eq_basis_of_H_1}
\end{equation}
Here $\xi_k\in \irr(\C_0)$ takes values in anyons such that $t_{k+1}\in t_k\otimes \xi_k$, and $\eta_k\in 1,\cdots , d_{\xi_k}$ label orthonormal projections. This gives us a natural basis $|\xi_k,\eta_k\>$ for $\Hom(t_k\otimes R_0, t_{k+1})$. We can then choose some bijection of sets
\begin{equation}
    \gamma_{t_k, t_{k+1}}: [a_t]\times [n_1]\times [a_{t+1}]\to \{(\xi_k,\eta_k): t_{k+1}\in t_k\otimes \xi_k, \eta_k\in 1,\cdots , d_{\xi_k}\}
\end{equation}
and identify 
\begin{equation}
    |  \gamma_{t_k, t_{k+1}}(\nu_k^{t_k},\sigma_{k},\nu_k^{t_{k+1}})\>=|\nu_k^{t_k},\sigma_{k},\nu_k^{t_{k+1}}\>.
\end{equation}
In other words the two basis are related by relabeling.  In practice, one should perform calculation of $\sf{D}_x$ first in the basis given by projections onto simples (\Cref{eq_basis_of_H_0} and \Cref{eq_basis_of_H_1}) via standard diagrammatics, then do a basis change to the tensor-product basis to reveal the tensor-product structure of the symmetry operator.

\section{Indices of Symmetry Operators}
\label{sec_index_computation}

We now compute the indices of the symmetry operators in the lattice models obtained by \Cref{construction_general_case}.
We will show that the index of $\mathsf{D}_x$ for $x \in \mathcal{C}_g$ is given by 
\be 
\ind(\mathsf{D}_x) = \dim(R_g)^{-1} = m_x \sqrt{n_g}, \quad
m_x \coloneq \frac{1}{\dim(\mathcal{C}_0)}\,.
\ee
We note that the index of $\mathsf{D}_x$ agrees with the dimension of the object translated by the sequential circuit in \Cref{fig:sequentialc_2}.\footnote{Similarly, one can also show that the index of $\mathsf{D}_x$ in \Cref{construction_special_case} is given by $\dim(x_g)^{-1} = \sqrt{n_g}^{-1}$, where $x_g$ is the object translated to the left in \Cref{fig:sequentialc_1}. In this section, we will focus only on \Cref{construction_general_case}.}
Furthermore, we will also show that the index of QCA $U_{(g, h)}$ appearing in the lattice-level fusion rules~\Cref{eq: lattice fusion rules} is given by
\begin{equation}
\ind(U_{(g, h)}) = \frac{\dim(R_{gh})}{\dim(R_g) \dim(R_h)} = \frac{m_x m_y}{m_z} \sqrt{\frac{n_g n_h}{n_{gh}}}.
\end{equation}
This result agrees with the prediction of \Cref{thm_main}.

In what follows, the symmetry operators on the graded Hilbert space $\mathcal{H} = \bigoplus_{h \in E} \mathcal{H}_h$ are denoted by $\widetilde{\mathsf{D}}_x \coloneq U_g \mathcal{D}_x$, while the symmetry operators on a single sector $\mathcal{H}_0$ are denoted by $\mathsf{D}_x \coloneq \widetilde{\mathsf{D}}_x|_{\mathcal{H}_0}$.
Here, $\mathcal{D}_x: \mathcal{H} \to \mathcal{H}$ is the standard symmetry operator defined by~\Cref{eq_anyon_chain_action}.
We note that $\widetilde{\mathsf{D}}_x$ preserves the grading, while $\mathcal{D}_x$ does not.

\subsection{Defect Hilbert spaces}
To compute the index of $\mathsf{D}_x$, we first define the left and right defect Hilbert spaces for $\mathsf{D}_x$.
To this end, we first recall the definition of the defect Hilbert spaces in the anyon chains \cite{Buican:2017rxc, Aasen:2020jwb, Bhardwaj:2024kvy}.
Since we focus on \Cref{construction_general_case}, we consider the anyon chain with input data $(\mathcal{C}, \mathcal{C}, \dim(\mathcal{C}_0) R_0)$.
In this anyon chain, the left defect Hilbert space for the symmetry operator $\mathcal{D}_x$ without refining by $U_g$ is defined by the vector space of morphisms represented by the following fusion diagrams:
\begin{equation}
\begin{tikzpicture}[scale = 1, very thick, baseline=(current bounding box.center)]
\pgfmathsetmacro{\x}{2.25}
\pgfmathsetmacro{\y}{1}
\draw[->-] (0, 0) -- (\x, 0);
\draw[->-] (\x, 0) -- (2*\x, 0);
\draw[->-] (2*\x, 0) -- (2.5*\x, 0);
\draw[->-] (2.5*\x, 0) -- (3*\x, 0);
\draw[->-] (3*\x, 0) -- (4*\x, 0);
\draw[->-] (4*\x, 0) -- (5*\x, 0);
\draw[->-] (\x, -\y) -- (\x, 0);
\draw[->-] (2*\x, -\y) -- (2*\x, 0);
\draw[newblue, ->-] (2.5*\x, 0) -- (2.5*\x, \y);
\draw[->-] (3*\x, -\y) -- (3*\x, 0);
\draw[->-] (4*\x, -\y) -- (4*\x, 0);
\draw[fill=black] (\x, 0) circle (0.1);
\draw[fill=black] (2*\x, 0) circle (0.1);
\draw[fill=black] (2.5*\x, 0) circle (0.1);
\draw[fill=black] (3*\x, 0) circle (0.1);
\draw[fill=black] (4*\x, 0) circle (0.1);
\node[below] at (\x, -\y) {$\mathrm{dim}(\mathcal{C}_0)R_0$};
\node[below] at (2*\x, -\y) {$\mathrm{dim}(\mathcal{C}_0)R_0$};
\node[below] at (3*\x, -\y) {$\mathrm{dim}(\mathcal{C}_0)R_0$};
\node[below] at (4*\x, -\y) {$\mathrm{dim}(\mathcal{C}_0)R_0$};
\node[right] at (2.5*\x, 0.75*\y) {$x$};
\end{tikzpicture}
\;.
\label{eq: defect Hilb cal}
\end{equation}
Here, the horizontal edges are labeled by objects of $\mathcal{C}$.
The right defect Hilbert space is defined in the same way by flipping the orientation of the defect labeled by $x$.
We denote the number of vertical legs below the chain by $L$.

The defect Hilbert space for the symmetry operator $\widetilde{\mathsf{D}}_x = U_g \mathcal{D}_x$ is obtained by further inserting a $U_g$-defect that ends on the chain from below.
The insertion of a $U_g$-defect is done by following the definition of the defect Hilbert space for a sequential circuit \cite{Inamura:2026hif}.
More specifically, we define the left defect Hilbert space for $\widetilde{\mathsf{D}}_x$ by the vector space of morphisms represented by the following fusion diagrams:
\begin{equation}
\begin{tikzpicture}[scale = 1, very thick, baseline=(current bounding box.center)]
\pgfmathsetmacro{\x}{2.25}
\pgfmathsetmacro{\y}{1.2}
\pgfmathsetmacro{\d}{0.1}
\draw[->-] (0, 0) -- (\x, 0);
\draw[->-] (\x, 0) -- (2*\x, 0);
\draw[->-] (2*\x, 0) -- (2.5*\x, 0);
\draw[->-] (2.5*\x, 0) -- (3*\x, 0);
\draw[->-] (3*\x, 0) -- (3.5*\x, 0);
\draw[->-] (3.5*\x, 0) -- (4.5*\x, 0);
\draw[->-] (4.5*\x, 0) -- (5.5*\x, 0);
\draw[->-] (\x, -1.5*\y) -- (\x, 0);
\draw[->-] (2*\x, -1.5*\y) -- (2*\x, 0);
\draw[newblue, ->-] (2.5*\x, 0) -- (2.5*\x, \y);
\draw[->-] (3.5*\x, -1.5*\y) -- (3.5*\x, -\y);
\draw[->-] (3.5*\x, -\y) -- (3.5*\x, -0.4*\y);
\draw[->-] (3.5*\x, -0.4*\y) -- (3.5*\x, 0);
\draw[->-] (4.5*\x, -1.5*\y) -- (4.5*\x, 0);
\draw[newblue] (3*\x, -1.5*\y) -- (3*\x, -\y)-- (4*\x, -\y) -- (4*\x, -0.5*\y) -- (3*\x, -0.5*\y) -- (3*\x, 0);
\draw[newblue, ->-] (3*\x, -1.2*\y) -- (3*\x, -1.5*\y);
\draw[newblue, ->-] (3.5*\x, -\y) -- (3*\x, -\y);
\draw[newblue, ->-] (4*\x, -\y) -- (3.5*\x, -\y);
\draw[newblue, ->-] (3.5*\x, -0.5*\y) -- (4*\x, -0.5*\y);
\draw[newblue, ->-] (3*\x, -0.5*\y) -- (3.5*\x, -0.5*\y);
\draw[newblue, ->-] (3*\x, -0.2*\y) -- (3*\x, -0.5*\y);
\draw[fill=black] (\x, 0) circle (0.1);
\draw[fill=black] (2*\x, 0) circle (0.1);
\draw[fill=black] (2.5*\x, 0) circle (0.1);
\draw[fill=black] (3*\x, 0) circle (0.1);
\draw[fill=black] (3.5*\x, 0) circle (0.1);
\draw[fill=black] (4.5*\x, 0) circle (0.1);
\draw[fill=white] (3.5*\x-\d, -\y-\d) rectangle (3.5*\x+\d, -\y+\d);
\draw[fill=black] (3.5*\x-\d, -0.5*\y-\d) rectangle (3.5*\x+\d, -0.5*\y+\d);
\node[below] at (\x, -1.5*\y) {$\mathrm{dim}(\mathcal{C}_0)R_0$};
\node[below] at (2*\x, -1.5*\y) {$\mathrm{dim}(\mathcal{C}_0)R_0$};
\node[below] at (3.5*\x, -1.5*\y) {$\mathrm{dim}(\mathcal{C}_0)R_0$};
\node[below] at (4.5*\x, -1.5*\y) {$\mathrm{dim}(\mathcal{C}_0)R_0$};
\node[right] at (2.5*\x, \y) {$x$};
\node[left] at (3*\x, -\y) {$R_g$};
\end{tikzpicture}
\; .
\label{eq: defect Hilb l}
\end{equation}
Here, the black and white squares at the quadrivalent junctions represent the local pieces of the sequential circuits $U_g$ and $U_g^{\dagger}$, i.e.,
\begin{equation}
\begin{tikzpicture}[scale = 0.7, very thick, baseline=(current bounding box.center)]
\pgfmathsetmacro{\x}{2}
\pgfmathsetmacro{\y}{2}
\pgfmathsetmacro{\d}{0.15}
\draw[newblue, ->-] (-0.5*\x, 0) -- (0, 0);
\draw[newblue] (0, 0) -- (0.5*\x, 0);
\draw[newblue, ->-] (0.2*\x, 0) -- (0.5*\x, 0);
\draw (0, -0.5*\y) -- (0, 0.5*\y);
\draw[->-] (0, -0.5*\y) -- (0, 0);
\draw[->-] (0, 0.2*\y) -- (0, 0.5*\y);
\draw[fill=black] (0-\d, -\d) rectangle (0+\d, \d);
\node[below] at (0, -0.5*\y) {$\mathrm{dim}(\mathcal{C}_0)R_0$};
\node[above] at (0, 0.5*\y) {$\mathrm{dim}(\mathcal{C}_0)R_0$};
\node[left] at (-0.5*\x, 0) {$R_g$};
\node[right] at (0.5*\x, 0) {$R_g$};
\end{tikzpicture}
=
\begin{tikzpicture}[scale = 0.7, very thick, baseline=(current bounding box.center)]
\pgfmathsetmacro{\x}{3.5}
\pgfmathsetmacro{\y}{3}
\pgfmathsetmacro{\d}{0.15}
\draw[newblue, ->-] (-0.5*\x, 0) -- (0, 0.2*\y);
\draw[newblue, ->-] (0, -0.2*\y) -- (0.5*\x, 0);
\draw (0, -0.5*\y) -- (0, -0.2*\y);
\draw[myorange] (0, -0.2*\y) -- (0, 0.2*\y);
\draw (0, 0.2*\y) -- (0, 0.5*\y);
\draw[->-] (0, -0.5*\y) -- (0, -0.2*\y);
\draw[myorange, ->-] (0, -0.1*\y) -- (0, 0.2*\y);
\draw[->-] (0, 0.3*\y) -- (0, 0.5*\y);
\draw[fill=red] (0, -0.2*\y) circle (0.15);
\draw[fill=blue] (0, 0.2*\y) circle (0.15);
\node[below] at (0, -0.5*\y) {$\mathrm{dim}(\mathcal{C}_0)R_0$};
\node[above] at (0, 0.5*\y) {$\mathrm{dim}(\mathcal{C}_0)R_0$};
\node[left] at (-0.5*\x, 0) {$R_g$};
\node[right] at (0.5*\x, 0) {$R_g$};
\node[right] at (0, 0.15) {$n_g R_g$};
\node[left] at (0-0.1, -0.2*\y-0.1) {$\gamma_g$};
\node[left] at (0-0.1, 0.2*\y+0.3) {$\delta_g$};
\end{tikzpicture}
, \qquad
\begin{tikzpicture}[scale = 0.7, very thick, baseline=(current bounding box.center)]
\pgfmathsetmacro{\x}{2}
\pgfmathsetmacro{\y}{2}
\pgfmathsetmacro{\d}{0.15}
\draw[newblue] (0, 0) -- (-0.5*\x, 0);
\draw[newblue, ->-] (-0.2*\x, 0) -- (-0.5*\x, 0);
\draw[newblue, ->-] (0.5*\x, 0) -- (0, 0);
\draw (0, -0.5*\y) -- (0, 0.5*\y);
\draw[->-] (0, -0.5*\y) -- (0, 0);
\draw[->-] (0, 0.2*\y) -- (0, 0.5*\y);
\draw[fill=white] (0-\d, -\d) rectangle (0+\d, \d);
\node[below] at (0, -0.5*\y) {$\mathrm{dim}(\mathcal{C}_0)R_0$};
\node[above] at (0, 0.5*\y) {$\mathrm{dim}(\mathcal{C}_0)R_0$};
\node[left] at (-0.5*\x, 0) {$R_g$};
\node[right] at (0.5*\x, 0) {$R_g$};
\end{tikzpicture}
=
\begin{tikzpicture}[scale = 0.7, very thick, baseline=(current bounding box.center)]
\pgfmathsetmacro{\x}{3.5}
\pgfmathsetmacro{\y}{3}
\pgfmathsetmacro{\d}{0.15}
\draw[newblue, ->-] (0, -0.2*\y) -- (-0.5*\x, 0);
\draw[newblue, ->-] (0.5*\x, 0) -- (0, 0.2*\y);
\draw (0, -0.5*\y) -- (0, -0.2*\y);
\draw[myorange] (0, -0.2*\y) -- (0, 0.2*\y);
\draw (0, 0.2*\y) -- (0, 0.5*\y);
\draw[->-] (0, -0.5*\y) -- (0, -0.2*\y);
\draw[myorange, ->-] (0, -0.1*\y) -- (0, 0.2*\y);
\draw[->-] (0, 0.3*\y) -- (0, 0.5*\y);
\draw[fill=blue] (0, -0.2*\y) circle (0.15);
\draw[fill=red] (0, 0.2*\y) circle (0.15);
\node[below] at (0, -0.5*\y) {$\mathrm{dim}(\mathcal{C}_0)R_0$};
\node[above] at (0, 0.5*\y) {$\mathrm{dim}(\mathcal{C}_0)R_0$};
\node[left] at (-0.5*\x, 0) {$R_g$};
\node[right] at (0.5*\x, 0) {$R_g$};
\node[left] at (0, 0.15) {$n_g R_g$};
\node[right] at (0+0.05, -0.2*\y) {$\delta_g^{\dagger}$};
\node[right] at (0+0.05, 0.2*\y+0.3) {$\gamma_g^{\dagger}$};
\end{tikzpicture},
\label{eq: 4-valent junction mor}
\end{equation}
where $\gamma_g: \dim(\mathcal{C}_0)R_0 \to n_g R_g \otimes R_g$ and $\delta_g: R_g \otimes n_g R_g \to \dim(\mathcal{C}_0) R_0$ are isomorphisms in~\Cref{eq_gamma_delta}.
The orientation of the defect $R_g$ in \Cref{eq: defect Hilb l} is chosen so that the product of local unitaries that move the defect to the right acts in the same way as the corresponding operator $U_g$; cf.~\Cref{fn: defect orientation}.
Here, the local unitary $u$ that moves the defect $R_g$ to the right by one site is given by
\begin{equation}
u = \frac{1}{\dim(R_g)^2}
\begin{tikzpicture}[scale = 1, very thick, baseline=(current bounding box.center)]
\pgfmathsetmacro{\x}{2}
\pgfmathsetmacro{\y}{0.6}
\pgfmathsetmacro{\d}{0.1}
\draw (0, 0) -- (0, 3*\y);
\draw (\x, 0) -- (\x, 3*\y);
\draw[->-] (0, 0) -- (0, 2*\y);
\draw[->-] (0, 2.5*\y) -- (0, 3*\y);
\draw[->-] (\x, 0) -- (\x, \y);
\draw[->-] (\x, 1.25*\y) -- (\x, 2*\y);
\draw[->-] (\x, 2.5*\y) -- (\x, 3*\y);
\draw[newblue] (-0.5*\x, 3*\y) -- (-0.5*\x, 2*\y) -- (1.5*\x, 2*\y) -- (1.5*\x, \y) -- (0.5*\x, \y) -- (0.5*\x, 0);
\draw[newblue, ->-] (-0.5*\x, 2*\y) -- (0, 2*\y);
\draw[newblue, ->-] (0, 2*\y) -- (\x, 2*\y);
\draw[newblue, ->-] (\x, 2*\y) -- (1.5*\x, 2*\y);
\draw[newblue, ->-] (1.5*\x, \y) -- (\x, \y);
\draw[newblue, ->-] (\x, \y) -- (0.5*\x, \y);
\draw[fill=black] (0-\d, 2*\y-\d) rectangle (0+\d, 2*\y+\d);
\draw[fill=black] (\x-\d, 2*\y-\d) rectangle (\x+\d, 2*\y+\d);
\draw[fill=white] (\x-\d, \y-\d) rectangle (\x+\d, \y+\d);
\node[below] at (0, 0) {$\rho$};
\node[below] at (0.5*\x, 0) {$R_g$};
\node[below] at (\x, 0) {$\rho$};
\node[above] at (-0.5*\x, 3*\y) {$R_g$};
\node[above] at (0, 3*\y) {$\rho$};
\node[above] at (\x, 3*\y) {$\rho=\dim(\mathcal{C}_0)R_0$};
\end{tikzpicture}
\;.
\label{eq: local unitary l}
\end{equation}
A direct computation shows that $u^{\dagger} u$ acts as the identity on the state in \Cref{eq: defect Hilb l}, meaning that $u$ is unitary.\footnote{The factor of $\dim(R_g)^2$ in \Cref{eq: local unitary l} cancels out $\dim(R_g)$'s coming from the loops of $R_g$ that appear in the computation. We note that the diagram on the right-hand side is not unitary unless normalized correctly. This is because non-unitary evaluation and coevaluation morphisms are implicitly used in the diagram.}

Similarly, the right defect Hilbert space for $\widetilde{\mathsf{D}}_x$ is defined by the vector space of morphisms represented by the fusion diagrams
\begin{equation}
\begin{tikzpicture}[scale = 1, very thick, baseline=(current bounding box.center)]
\pgfmathsetmacro{\x}{2.25}
\pgfmathsetmacro{\y}{1.2}
\pgfmathsetmacro{\d}{0.1}
\draw[->-] (0, 0) -- (\x, 0);
\draw[->-] (\x, 0) -- (2*\x, 0);
\draw[->-] (2*\x, 0) -- (2.5*\x, 0);
\draw[->-] (2.5*\x, 0) -- (3*\x, 0);
\draw[->-] (3*\x, 0) -- (3.5*\x, 0);
\draw[->-] (3.5*\x, 0) -- (4.5*\x, 0);
\draw[->-] (4.5*\x, 0) -- (5.5*\x, 0);
\draw[->-] (\x, -1.5*\y) -- (\x, 0);
\draw[->-] (2*\x, -1.5*\y) -- (2*\x, 0);
\draw[newblue, ->-] (2.5*\x, \y) -- (2.5*\x, 0);
\draw[->-] (3.5*\x, -1.5*\y) -- (3.5*\x, -\y);
\draw[->-] (3.5*\x, -\y) -- (3.5*\x, -0.4*\y);
\draw[->-] (3.5*\x, -0.4*\y) -- (3.5*\x, 0);
\draw[->-] (4.5*\x, -1.5*\y) -- (4.5*\x, 0);
\draw[newblue] (3*\x, -1.5*\y) -- (3*\x, -\y)-- (4*\x, -\y) -- (4*\x, -0.5*\y) -- (3*\x, -0.5*\y) -- (3*\x, 0);
\draw[newblue, ->-] (3*\x, -1.5*\y) -- (3*\x, -\y);
\draw[newblue, ->-] (3*\x, -\y) -- (3.5*\x, -\y);
\draw[newblue, ->-] (3.5*\x, -\y) -- (4*\x, -\y);
\draw[newblue, ->-] (4*\x, -0.5*\y) -- (3.5*\x, -0.5*\y);
\draw[newblue, ->-] (3.5*\x, -0.5*\y) -- (3*\x, -0.5*\y);
\draw[newblue, ->-] (3*\x, -0.5*\y) -- (3*\x, 0);
\draw[fill=black] (\x, 0) circle (0.1);
\draw[fill=black] (2*\x, 0) circle (0.1);
\draw[fill=black] (2.5*\x, 0) circle (0.1);
\draw[fill=black] (3*\x, 0) circle (0.1);
\draw[fill=black] (3.5*\x, 0) circle (0.1);
\draw[fill=black] (4.5*\x, 0) circle (0.1);
\draw[fill=black] (3.5*\x-\d, -\y-\d) rectangle (3.5*\x+\d, -\y+\d);
\draw[fill=white] (3.5*\x-\d, -0.5*\y-\d) rectangle (3.5*\x+\d, -0.5*\y+\d);
\node[below] at (\x, -1.5*\y) {$\mathrm{dim}(\mathcal{C}_0)R_0$};
\node[below] at (2*\x, -1.5*\y) {$\mathrm{dim}(\mathcal{C}_0)R_0$};
\node[below] at (3.5*\x, -1.5*\y) {$\mathrm{dim}(\mathcal{C}_0)R_0$};
\node[below] at (4.5*\x, -1.5*\y) {$\mathrm{dim}(\mathcal{C}_0)R_0$};
\node[right] at (2.5*\x, \y) {$x$};
\node[left] at (3*\x, -\y) {$R_g$};
\end{tikzpicture}
\;,
\label{eq: defect Hilb r}
\end{equation}
where the black and white squares are again given by \Cref{eq: 4-valent junction mor}.
For this defect Hilbert space, the local unitary $v$ that moves the defect $R_g$ to the right is given by
\begin{equation}
v =
\begin{tikzpicture}[scale = 1, very thick, baseline=(current bounding box.center)]
\pgfmathsetmacro{\x}{2}
\pgfmathsetmacro{\y}{0.6}
\pgfmathsetmacro{\d}{0.1}
\draw (0, 0) -- (0, 3*\y);
\draw (\x, 0) -- (\x, 3*\y);
\draw[->-] (0, 0) -- (0, 2*\y);
\draw[->-] (0, 2.5*\y) -- (0, 3*\y);
\draw[->-] (\x, 0) -- (\x, \y);
\draw[->-] (\x, 1.25*\y) -- (\x, 2*\y);
\draw[->-] (\x, 2.5*\y) -- (\x, 3*\y);
\draw[newblue] (-0.5*\x, 3*\y) -- (-0.5*\x, 2*\y) -- (1.5*\x, 2*\y) -- (1.5*\x, \y) -- (0.5*\x, \y) -- (0.5*\x, 0);
\draw[newblue, ->-] (0, 2*\y) -- (-0.5*\x, 2*\y);
\draw[newblue, ->-] (\x, 2*\y) -- (0, 2*\y);
\draw[newblue, ->-] (1.5*\x, 2*\y) -- (\x, 2*\y);
\draw[newblue, ->-] (\x, \y) -- (1.5*\x, \y);
\draw[newblue, ->-] (0.5*\x, \y) -- (\x, \y);
\draw[fill=white] (0-\d, 2*\y-\d) rectangle (0+\d, 2*\y+\d);
\draw[fill=white] (\x-\d, 2*\y-\d) rectangle (\x+\d, 2*\y+\d);
\draw[fill=black] (\x-\d, \y-\d) rectangle (\x+\d, \y+\d);
\node[below] at (0, 0) {$\rho$};
\node[below] at (0.5*\x, 0) {$R_g$};
\node[below] at (\x, 0) {$\rho$};
\node[above] at (-0.5*\x, 3*\y) {$R_g$};
\node[above] at (0, 3*\y) {$\rho$};
\node[above] at (\x, 3*\y) {$\rho=\dim(\mathcal{C}_0)R_0$};
\end{tikzpicture}
\;.
\end{equation}
Unlike \Cref{eq: local unitary l}, there is no factor of $\dim(R_g)$ on the right-hand side.
We note that the product of the above local unitaries acts in the same way as $U_g^{\dagger}$, up to defect creation and annihilation operators.

One can simplify the fusion diagrams in \Cref{eq: defect Hilb l} and \Cref{eq: defect Hilb r} by using
\begin{equation}
\begin{tikzpicture}[scale = 0.9, very thick, baseline=(current bounding box.center)]
\pgfmathsetmacro{\x}{2}
\pgfmathsetmacro{\y}{2}
\pgfmathsetmacro{\d}{0.1}
\draw[->-] (0, -1.5*\y) -- (0, -\y);
\draw[->-] (0, -\y) -- (0, -0.5*\y);
\draw[->-] (0, -0.5*\y) -- (0, 0);
\draw[newblue] (-0.5*\x, -1.5*\y) -- (-0.5*\x, -\y)-- (0.5*\x, -\y) -- (0.5*\x, -0.5*\y) -- (-0.5*\x, -0.5*\y) -- (-0.5*\x, 0);
\draw[newblue, ->-] (-0.5*\x, -1.5*\y) -- (-0.5*\x, -\y);
\draw[newblue, ->-] (-0.5*\x, -\y) -- (0, -\y);
\draw[newblue, ->-] (0, -\y) -- (0.5*\x, -\y);
\draw[newblue, ->-] (0.5*\x, -0.5*\y) -- (0, -0.5*\y);
\draw[newblue, ->-] (0, -0.5*\y) -- (-0.5*\x, -0.5*\y);
\draw[newblue, ->-] (-0.5*\x, -0.5*\y) -- (-0.5*\x, 0);
\draw[fill=black] (0-\d, -\y-\d) rectangle (0+\d, -\y+\d);
\draw[fill=white] (0-\d, -0.5*\y-\d) rectangle (0+\d, -0.5*\y+\d);
\node[below] at (0, -1.5*\y) {$\mathrm{dim}(\mathcal{C}_0)R_0$};
\node[above] at (0, 0) {$\mathrm{dim}(\mathcal{C}_0)R_0$};
\node[left] at (-0.5*\x, -\y) {$R_g$};
\node[left] at (-0.5*\x, -0.5*\y) {$R_g$};
\end{tikzpicture}
=
\begin{tikzpicture}[scale = 0.9, very thick, baseline=(current bounding box.center)]
\pgfmathsetmacro{\x}{2}
\pgfmathsetmacro{\y}{2}
\pgfmathsetmacro{\d}{0.1}
\draw[->-] (0, -1.5*\y) -- (0, 0);
\draw[newblue, ->-] (-0.5*\x, -1.5*\y) -- (-0.5*\x, 0);
\node[below] at (0, -1.5*\y) {$\mathrm{dim}(\mathcal{C}_0)R_0$};
\node[above] at (0, 0) {$\mathrm{dim}(\mathcal{C}_0)R_0$};
\node[left] at (-0.5*\x, -\y) {$R_g$};
\end{tikzpicture}
, \qquad
\begin{tikzpicture}[scale = 0.9, very thick, baseline=(current bounding box.center)]
\pgfmathsetmacro{\x}{2}
\pgfmathsetmacro{\y}{2}
\pgfmathsetmacro{\d}{0.1}
\draw[->-] (0, -1.5*\y) -- (0, -\y);
\draw[->-] (0, -\y) -- (0, -0.5*\y);
\draw[->-] (0, -0.5*\y) -- (0, 0);
\draw[newblue] (-0.5*\x, -1.5*\y) -- (-0.5*\x, -\y)-- (0.5*\x, -\y) -- (0.5*\x, -0.5*\y) -- (-0.5*\x, -0.5*\y) -- (-0.5*\x, 0);
\draw[newblue, ->-] (-0.5*\x, -\y) -- (-0.5*\x, -1.5*\y);
\draw[newblue, ->-] (0, -\y) -- (-0.5*\x, -\y);
\draw[newblue, ->-] (0.5*\x, -\y) -- (0, -\y);
\draw[newblue, ->-] (0, -0.5*\y) -- (0.5*\x, -0.5*\y);
\draw[newblue, ->-] (-0.5*\x, -0.5*\y) -- (0, -0.5*\y);
\draw[newblue, ->-] (-0.5*\x, 0) -- (-0.5*\x, -0.5*\y);
\draw[fill=white] (0-\d, -\y-\d) rectangle (0+\d, -\y+\d);
\draw[fill=black] (0-\d, -0.5*\y-\d) rectangle (0+\d, -0.5*\y+\d);
\node[below] at (0, -1.5*\y) {$\mathrm{dim}(\mathcal{C}_0)R_0$};
\node[above] at (0, 0) {$\mathrm{dim}(\mathcal{C}_0)R_0$};
\node[left] at (-0.5*\x, -\y) {$R_g$};
\node[left] at (-0.5*\x, -0.5*\y) {$R_g$};
\end{tikzpicture}
= \dim(R_g)
\begin{tikzpicture}[scale = 0.9, very thick, baseline=(current bounding box.center)]
\pgfmathsetmacro{\x}{2}
\pgfmathsetmacro{\y}{2}
\pgfmathsetmacro{\d}{0.1}
\draw[->-] (0, -1.5*\y) -- (0, -\y);
\draw[myorange, ->-] (0, -\y) -- (0, -0.5*\y);
\draw[->-] (0, -0.5*\y) -- (0, 0);
\draw[newblue] (-0.5*\x, -1.5*\y) -- (-0.5*\x, -\y) -- (0, -\y);
\draw[newblue] (0, -0.5*\y) -- (-0.5*\x, -0.5*\y) -- (-0.5*\x, 0);
\draw[newblue, ->-] (-0.5*\x, -\y) -- (-0.5*\x, -1.5*\y);
\draw[newblue, ->-] (0, -\y) -- (-0.5*\x, -\y);
\draw[newblue, ->-] (-0.5*\x, -0.5*\y) -- (0, -0.5*\y);
\draw[newblue, ->-] (-0.5*\x, 0) -- (-0.5*\x, -0.5*\y);
\draw[fill=black] (0, -\y) circle (0.1);
\draw[fill=black] (0, -0.5*\y) circle (0.1);
\node[below] at (0, -1.5*\y) {$\mathrm{dim}(\mathcal{C}_0)R_0$};
\node[above] at (0, 0) {$\mathrm{dim}(\mathcal{C}_0)R_0$};
\node[left] at (-0.5*\x, -\y) {$R_g$};
\node[left] at (-0.5*\x, -0.5*\y) {$R_g$};
\node[left] at (0, -1.5*\y/2) {$n_gR_g$};
\node[right] at (0+0.05, -\y) {$\delta_g$};
\node[right] at (0+0.05, -0.5*\y) {$\delta_g^{\dagger}$};
\end{tikzpicture}
\; ,
\label{eq: zigzag}
\end{equation}
which immediately follows from the unitarity of $\gamma_g$ and $\delta_g$.
By plugging the first equality into \Cref{eq: defect Hilb r}, we find that each state in $\mathcal{H}_{\widetilde{\mathsf{D}}_x}^r$ can be written as
\begin{equation}
\begin{tikzpicture}[scale = 1, very thick, baseline=(current bounding box.center)]
\pgfmathsetmacro{\x}{2.25}
\pgfmathsetmacro{\y}{1}
\draw[->-] (0, 0) -- (\x, 0);
\draw[->-] (\x, 0) -- (2*\x, 0);
\draw[->-] (2*\x, 0) -- (2.5*\x, 0);
\draw[->-] (2.5*\x, 0) -- (3*\x, 0);
\draw[->-] (3*\x, 0) -- (4*\x, 0);
\draw[->-] (4*\x, 0) -- (5*\x, 0);
\draw[->-] (\x, -\y) -- (\x, 0);
\draw[->-] (2*\x, -\y) -- (2*\x, 0);
\draw[newblue, ->-] (2.5*\x,\y) -- (2.5*\x, 0);
\draw[newblue, ->-] (2.5*\x, -\y) -- (2.5*\x, 0); 
\draw[->-] (3*\x, -\y) -- (3*\x, 0);
\draw[->-] (4*\x, -\y) -- (4*\x, 0);
\draw[fill=black] (\x, 0) circle (0.1);
\draw[fill=black] (2*\x, 0) circle (0.1);
\draw[fill=black] (2.5*\x, 0) circle (0.1);
\draw[fill=black] (3*\x, 0) circle (0.1);
\draw[fill=black] (4*\x, 0) circle (0.1);
\node[below] at (\x, -\y) {$\mathrm{dim}(\mathcal{C}_0)R_0$};
\node[below] at (2*\x, -\y) {$\mathrm{dim}(\mathcal{C}_0)R_0$};
\node[below] at (3*\x, -\y) {$\mathrm{dim}(\mathcal{C}_0)R_0$};
\node[below] at (4*\x, -\y) {$\mathrm{dim}(\mathcal{C}_0)R_0$};
\node[right] at (2.5*\x, 0.75*\y) {$x$};
\node[right] at (2.5*\x, -0.75*\y) {$R_g$};
\end{tikzpicture}
\;.
\label{eq: defect Hilb r2}
\end{equation}
This diagram shows that the right defect Hilbert space for $\widetilde{\mathsf{D}}_x$ is obtained by inserting an $x$-defect from above and an $R_g$-defect from below.
Similarly, due to the second equality of \Cref{eq: zigzag}, the left defect Hilbert space $\mathcal{H}_{\widetilde{\mathsf{D}}_x}^l$ is isomorphic to the vector space of morphisms represented by
\begin{equation}
\begin{tikzpicture}[scale = 1, very thick, baseline=(current bounding box.center)]
\pgfmathsetmacro{\x}{2.25}
\pgfmathsetmacro{\y}{1}
\draw[->-] (0, 0) -- (\x, 0);
\draw[->-] (\x, 0) -- (2*\x, 0);
\draw[->-] (2*\x, 0) -- (2.5*\x, 0);
\draw[->-] (2.5*\x, 0) -- (4*\x, 0);
\draw[->-] (4*\x, 0) -- (5*\x, 0);
\draw[->-] (\x, -\y) -- (\x, 0);
\draw[->-] (2*\x, -\y) -- (2*\x, 0);
\draw[newblue, ->-] (2.5*\x, 0) -- (2.5*\x, \y);
\draw[myorange, ->-] (2.5*\x, -\y) -- (2.5*\x, 0); 
\draw[->-] (4*\x, -\y) -- (4*\x, 0);
\draw[fill=black] (\x, 0) circle (0.1);
\draw[fill=black] (2*\x, 0) circle (0.1);
\draw[fill=black] (2.5*\x, 0) circle (0.1);
\draw[fill=black] (4*\x, 0) circle (0.1);
\node[below] at (\x, -\y) {$\mathrm{dim}(\mathcal{C}_0)R_0$};
\node[below] at (2*\x, -\y) {$\mathrm{dim}(\mathcal{C}_0)R_0$};
\node[below] at (4*\x, -\y) {$\mathrm{dim}(\mathcal{C}_0)R_0$};
\node[right] at (2.5*\x+0.1, 0.75*\y) {$x$};
\node[right] at (2.5*\x+0.1, -0.75*\y) {$n_gR_g$};
\end{tikzpicture}
\; .
\label{eq: defect Hilb l2}
\end{equation}
Namely, the left defect Hilbert space for $\widetilde{\mathsf{D}}_x$ is obtained by inserting an $x$-defect from above, removing a single leg labeled by $\dim(\mathcal{C}_0) R_0$, and adding an $n_gR_g$-defect from below.
Since the removed leg is labeled by $\dim(\mathcal{C}_0) R_0 \cong n_gR_g \otimes R_g$, removing a single leg and adding an $n_g R_g$-defect is equivalent to removing a single $R_g$-defect.
We note that the number of vertical legs for $\mathcal{H}_{\widetilde{\mathsf{D}}_x}^l$ is one less than that for $\mathcal{H}_{\widetilde{\mathsf{D}}_x}^r$.

As in the case of no defect, all horizontal edges in the fusion diagrams \eqref{eq: defect Hilb l2} and \eqref{eq: defect Hilb r2} have the same grading.
Thus, the defect Hilbert spaces $\mathcal{H}_{\widetilde{\mathsf{D}}_x}^l$ and $\mathcal{H}_{\widetilde{\mathsf{D}}_x}^r$ are graded by $E$, i.e.,
\begin{equation}
\mathcal{H}_{\widetilde{\mathsf{D}}_x}^l = \bigoplus_{h \in E} (\mathcal{H}_{\widetilde{\mathsf{D}}_x}^l)_h, \qquad
\mathcal{H}_{\widetilde{\mathsf{D}}_x}^r = \bigoplus_{h \in E} (\mathcal{H}_{\widetilde{\mathsf{D}}_x}^r)_h,
\label{eq: defect Hilb grading}
\end{equation}
where $h \in E$ is the grading of the objects on the horizontal edges.
Based on this decomposition, we define the left and right defect Hilbert spaces for the symmetry operator $\mathsf{D}_x: \mathcal{H}_0 \to \mathcal{H}_0$ on the trivially graded sector by
\begin{equation}
\mathcal{H}_{\mathsf{D}_x}^l \coloneq (\mathcal{H}_{\widetilde{\mathsf{D}}_x}^l)_0, \qquad
\mathcal{H}_{\mathsf{D}_x}^r \coloneq (\mathcal{H}_{\widetilde{\mathsf{D}}_x}^r)_0.
\label{eq: defect Hilb def}
\end{equation}

\vspace*{\baselineskip}
\noindent{\bf Tensor product decomposition.}
The defect Hilbert space $\mathcal{H}_{\mathsf{D}_x}^l$ defined above admits a tensor product decomposition.
More generally, for every $h \in E$, the $h$-graded component of $\mathcal{H}_{\widetilde{\mathsf{D}}_x}^l$ admits a tensor product decomposition.
To see this, we recall that each state in the $h$-graded component of $\mathcal{H}_{\widetilde{\mathsf{D}}_x}^l$ is represented by a fusion diagram in \Cref{eq: defect Hilb l2} whose horizontal edges are labeled by objects of $\mathcal{C}_h$.
For fixed labels of the horizontal edges, the Hilbert space associated with each vertex away from the defect has dimension
\begin{equation}
\dim \Hom_{\mathcal{C}}(x_i \otimes \dim(\mathcal{C}_0) R_0, x_{i+1}) = \dim(\mathcal{C}_0) \dim(x_i) \dim(x_{i+1}) = \dim(\mathcal{C}_0) n_h a_{x_i} a_{x_{i+1}},
\end{equation}
where $x_i, x_{i+1} \in \irr(\mathcal{C}_h)$ are the labels of the left and right edges connected to the vertex.
This Hilbert space can be factorized into the tensor product of three Hilbert spaces assigned to the vertex and the two adjacent edges as follows:
\begin{equation}
\begin{tikzpicture}[scale = 1, very thick, baseline=(current bounding box.center)]
\pgfmathsetmacro{\x}{2}
\pgfmathsetmacro{\y}{0.8}
\pgfmathsetmacro{\d}{0.1}
\draw[->-] (0, 0) -- (\x, 0);
\draw[->-] (\x, 0) -- (2*\x, 0);
\draw[->-] (\x, -1.5*\y) -- (\x, 0);
\draw[fill=black] (\x, 0) circle (0.1);
\node[below] at (\x, -1.5*\y) {$\mathrm{dim}(\mathcal{C}_0)R_0$};
\node[below] at (0.5*\x, 0-0.1) {$x_i$};
\node[below] at (1.5*\x, 0-0.1) {$x_{i+1}$};
\node[above] at (\x, 0+0.1) {$\mathbb{C}^{\dim(\mathcal{C}_0)n_h a_{x_i} a_{x_{i+1}}}$};
\end{tikzpicture}
\; \Rightarrow \;
\begin{tikzpicture}[scale = 1, very thick, baseline=(current bounding box.center)]
\pgfmathsetmacro{\x}{2}
\pgfmathsetmacro{\y}{0.8}
\pgfmathsetmacro{\d}{0.1}
\draw (0, 0) -- (\x, 0);
\draw (\x, 0) -- (2*\x, 0);
\draw (\x, -1.5*\y) -- (\x, 0);
\draw[fill=blue, draw=black] (\x, 0) circle (0.1);
\draw[fill=orange, draw=black] (0.5*\x, 0) circle (0.1);
\draw[fill=orange, draw=black] (1.5*\x, 0) circle (0.1);
\node[below] at (\x, -1.5*\y) {$\mathrm{dim}(\mathcal{C}_0)R_0$};
\node[below] at (0.2*\x, 0-0.1) {$x_i$};
\node[below] at (1.8*\x, 0-0.1) {$x_{i+1}$};
\node[above] at (\x, 0+0.1) {$\mathbb{C}^{\dim(\mathcal{C}_0)}$};
\node[above] at (0.3*\x, 0+0.1) {$\mathbb{C}^{n_h a_{x_i}}$};
\node[above] at (1.7*\x, 0+0.1) {$\mathbb{C}^{a_{x_{i+1}}}$};
\end{tikzpicture}
\;.
\label{eq: factorization 1}
\end{equation}
On the other hand, the Hilbert space associated with the vertex at the defect locus has dimension
\begin{equation}
\dim \Hom_{\mathcal{C}}(x_i \otimes n_g R_g, x \otimes x_{i}^{\prime}) = \dim(x) \dim(x_i) \dim(x_i^{\prime})\sqrt{n_g} = n_g a_x n_h a_{x_i} a_{x_i^{\prime}},
\end{equation}
where $x_i, x_i^{\prime} \in \irr(\mathcal{C}_h)$ are the labels of the left and right edges connected to the vertex.
This Hilbert space can also be factorized into the tensor product of three Hilbert spaces as
\begin{equation}
\begin{tikzpicture}[scale = 1, very thick, baseline=(current bounding box.center)]
\pgfmathsetmacro{\x}{2}
\pgfmathsetmacro{\y}{0.8}
\pgfmathsetmacro{\d}{0.1}
\draw[->-] (0, 0) -- (\x, 0);
\draw[->-] (\x, 0) -- (2*\x, 0);
\draw[newblue, ->-] (\x, 0) -- (\x, 1.5*\y);
\draw[myorange, ->-] (\x, -1.5*\y) -- (\x, 0);
\draw[fill=black] (\x, 0) circle (0.1);
\node[right] at (\x, 1.25*\y) {$x$};
\node[below] at (\x, -1.5*\y) {$n_g R_g$};
\node[below] at (0.5*\x, 0-0.1) {$x_i$};
\node[below] at (1.5*\x, 0-0.1) {$x_i^{\prime}$};
\node[above right] at (\x, 0) {$\mathbb{C}^{n_g a_x n_h a_{x_i} a_{x_i^{\prime}}}$};
\end{tikzpicture}
\; \Rightarrow \;
\begin{tikzpicture}[scale = 1, very thick, baseline=(current bounding box.center)]
\pgfmathsetmacro{\x}{2}
\pgfmathsetmacro{\y}{0.8}
\pgfmathsetmacro{\d}{0.1}
\draw (0, 0) -- (\x, 0);
\draw (\x, 0) -- (2*\x, 0);
\draw[newblue] (\x, 0) -- (\x, 1.5*\y);
\draw[myorange] (\x, -1.5*\y) -- (\x, 0);
\draw[fill=blue, draw=black] (\x, 0) circle (0.1);
\draw[fill=orange, draw=black] (0.5*\x, 0) circle (0.1);
\draw[fill=orange, draw=black] (1.5*\x, 0) circle (0.1);
\node[right] at (\x, 1.25*\y) {$x$};
\node[below] at (\x, -1.5*\y) {$R_g$};
\node[below] at (0.2*\x, 0-0.1) {$x_i$};
\node[below] at (1.8*\x, 0-0.1) {$x_i^{\prime}$};
\node[above right] at (\x, 0+0.1) {$\mathbb{C}^{n_g a_x}$};
\node[above] at (0.5*\x, 0+0.1) {$\mathbb{C}^{n_h a_{x_i}}$};
\node[above] at (1.7*\x, 0+0.1) {$\mathbb{C}^{a_{x_i^{\prime}}}$};
\end{tikzpicture}
\;.
\label{eq: factorization 2}
\end{equation}
\Cref{eq: factorization 1} and \Cref{eq: factorization 2} imply the following decomposition of the Hilbert space for fixed labels: each edge has a Hilbert space of dimension $n_h a_{x_i}^2 = \dim(x_i)^2$, each vertex except for the one on the defect has a Hilbert space of dimension $\dim(\mathcal{C}_0)$, and the vertex at the defect locus has a Hilbert space of dimension $n_g a_x$.
By taking the direct sum over all labels in $\irr(\mathcal{C}_h)$, the dimension of the Hilbert space on each edge becomes $\sum_{x_i \in \irr(\mathcal{C}_h)} \dim(x_i)^2 = \dim(\mathcal{C}_h) = \dim(\mathcal{C}_0)$.
Now, we define a site by the pair of a vertex and the adjacent edge on its left.
If a site contains a vertex at the defect locus, the on-site dimension is $n_g a_x \dim(\mathcal{C}_0)$.
Otherwise, the on-site dimension is $\dim(\mathcal{C}_0)^2$.
Thus, the $h$-graded component of $\mathcal{H}_{\widetilde{\mathsf{D}}_x}^l$ can be written as
\begin{equation}
(\mathcal{H}_{\widetilde{\mathsf{D}}_x}^l)_h \cong \mathbb{C}^{n_g a_x \dim(\mathcal{C}_0)} \otimes \bigotimes_{i = 1}^{L-1} \mathbb{C}^{\dim(\mathcal{C}_0)^2}.
\label{eq: graded defect Hilb l}
\end{equation}
Here, we recall that the number of vertical legs in \Cref{eq: defect Hilb l2} is $L-1$.
A similar computation shows that the $h$-graded component of $\mathcal{H}_{\widetilde{\mathsf{D}}_x}^r$ is given by
\begin{equation}
(\mathcal{H}_{\widetilde{\mathsf{D}}_x}^r)_h \cong \mathbb{C}^{a_x \dim(\mathcal{C}_0)} \otimes \bigotimes_{i = 1}^{L} \mathbb{C}^{\dim(\mathcal{C}_0)^2}.
\label{eq: graded defect Hilb r}
\end{equation}
We note that the Hilbert spaces \eqref{eq: graded defect Hilb l} and \eqref{eq: graded defect Hilb r} do not depend on the grading $h \in E$ up to isomorphism.
In particular, the trivially graded components are given by
\begin{equation}
\mathcal{H}_{\mathsf{D}_x}^l \cong \mathbb{C}^{n_g a_x \dim(\mathcal{C}_0)} \otimes \bigotimes_{i = 1}^{L-1} \mathbb{C}^{\dim(\mathcal{C}_0)^2}, \qquad
\mathcal{H}_{\mathsf{D}_x}^r \cong \mathbb{C}^{a_x \dim(\mathcal{C}_0)} \otimes \bigotimes_{i = 1}^{L} \mathbb{C}^{\dim(\mathcal{C}_0)^2}.
\label{eq: defect Hilb tens prod}
\end{equation}

\subsection{Indices}
\noindent{\bf Index of symmetry operator $\mathsf{D}_x$.}
Based on the above definition of the defect Hilbert spaces, we can now compute the indices of the symmetry operators.
The tensor product decomposition in \Cref{eq: defect Hilb tens prod} implies that the left and right dimensions of $\mathsf{D}_x$ for $x \in \mathcal{C}_g$ are given by
\begin{equation}
\ldim(\mathsf{D}_x) = \frac{n_g a_x}{\dim(\mathcal{C}_0)} = \frac{\dim(x)}{\dim(R_g)}, \qquad 
\rdim(\mathsf{D}_x) = a_x \dim(\mathcal{C}_0) = \dim(x) \dim(R_g).
\end{equation}
Therefore, the index of $\mathsf{D}_x$ can be computed as
\begin{equation}
\ind(\mathsf{D}_x) = \sqrt{\frac{\ldim(\mathsf{D}_x)}{\rdim(\mathsf{D}_x)}} = \dim(R_g)^{-1} = \frac{\sqrt{n_g}}{\dim(\mathcal{C}_0)}.
\label{eq: ind Dx}
\end{equation}
In particular, the irrational part of $\ind(\mathsf{D}_x)$ is $\sqrt{n_g}$, as predicted in \Cref{thm_main}.

The index of $\mathsf{D}_x$ computed above depends only on the grading of $x \in \mathcal{C}_g$.
This implies that the index is homogeneous.
Indeed, for the lattice-level fusion rules~\eqref{eq: lattice fusion rules}, all fusion channels $U_{(g, h)} \mathsf{D}_z$ with $z \in \irr(\mathcal{C}_{gh})$ have the same index because the indices of both $U_{(g, h)}$ and $\mathsf{D}_z$ depend only on the grading.
Here, we recall $U_{(g, h)} \coloneq U_g U_h U_{gh}^{-1}$.
In what follows, we will explicitly compute $\ind(U_{(g, h)})$ and show that the index satisfies
\begin{equation}
\ind(\mathsf{D}_x) \ind(\mathsf{D}_y) = \ind(U_{(g, h)}) \ind(\mathsf{D}_z)
\label{eq: homogeneity}
\end{equation}
for all fusion channels.

\vspace*{\baselineskip}
\noindent{\bf Index of QCA $U_{(g, h)}$.}
To compute the index of $U_{(g, h)}$, we first define the left and right defect Hilbert spaces for $U_{(g, h)}$.
Following the definition of the $U_g$-defect discussed above, we define the left defect Hilbert space $\mathcal{H}^l_{U_{(g, h)}}$ by the vector space of morphisms represented by the following fusion diagrams:
\begin{equation}
\begin{tikzpicture}[scale = 0.98, very thick, baseline=(current bounding box.center)]
\pgfmathsetmacro{\x}{2}
\pgfmathsetmacro{\y}{6}
\pgfmathsetmacro{\d}{0.125}
\draw[->-] (0.5*\x, 0) -- (\x, 0);
\draw[->-] (\x, 0) -- (2*\x, 0);
\draw[->-] (2*\x, 0) -- (2.5*\x, 0);
\draw[->-] (2.5*\x, 0) -- (3*\x, 0);
\draw[->-] (3*\x, 0) -- (3.5*\x, 0);
\draw[->-] (3.5*\x, 0) -- (4*\x, 0);
\draw[->-] (4*\x, 0) -- (5*\x, 0);
\draw[->-] (5*\x, 0) -- (6*\x, 0);
\draw[->-] (6*\x, 0) -- (7*\x, 0);
\draw[->-] (7*\x, 0) -- (7.5*\x, 0);
\draw[->-] (\x, -0.7*\y) -- (\x, 0);
\draw[->-] (2*\x, -0.7*\y) -- (2*\x, 0);
\draw[] (4*\x, -0.7*\y) -- (4*\x, 0);
\draw[->-] (5*\x, -0.7*\y) -- (5*\x, 0);
\draw[->-] (6*\x, -0.7*\y) -- (6*\x, 0);
\draw[->-] (7*\x, -0.7*\y) -- (7*\x, 0);
\draw[newblue] (2.5*\x, -0.7*\y) -- (2.5*\x, -0.4*\y) -- (4.5*\x, -0.4*\y) -- (4.5*\x, -0.3*\y) -- (2.5*\x, -0.3*\y) -- (2.5*\x, 0);
\draw[newblue, ->-] (2.5*\x, -0.4*\y) -- (2.5*\x, -0.7*\y);
\draw[newblue, ->-] (4*\x, -0.4*\y) -- (2.5*\x, -0.4*\y);
\draw[newblue, ->-] (4.5*\x, -0.4*\y) -- (4*\x, -0.4*\y);
\draw[newblue, ->-] (4*\x, -0.3*\y) -- (4.5*\x, -0.3*\y);
\draw[newblue, ->-] (2.5*\x, -0.3*\y) -- (4*\x, -0.3*\y);
\draw[newblue, ->-] (2.5*\x, 0) -- (2.5*\x, -0.3*\y);
\draw[newblue] (3*\x, -0.7*\y) -- (3*\x, -0.5*\y) -- (5.5*\x, -0.5*\y) -- (5.5*\x, -0.2*\y) -- (3*\x, -0.2*\y) -- (3*\x, 0);
\draw[newblue, ->-] (3*\x, 0) -- (3*\x, -0.2*\y);
\draw[newblue, ->-] (3*\x, -0.2*\y) -- (4*\x, -0.2*\y);
\draw[newblue, ->-] (4*\x, -0.2*\y) -- (5*\x, -0.2*\y);
\draw[newblue, ->-] (5*\x, -0.2*\y) -- (5.5*\x, -0.2*\y);
\draw[newblue, ->-] (5.5*\x, -0.5*\y) -- (5*\x, -0.5*\y);
\draw[newblue, ->-] (5*\x, -0.5*\y) -- (4*\x, -0.5*\y);
\draw[newblue, ->-] (4*\x, -0.5*\y) -- (3*\x, -0.5*\y);
\draw[newblue, ->-] (3*\x, -0.55*\y) -- (3*\x, -0.7*\y);
\draw[newblue] (3.5*\x, -0.7*\y) -- (3.5*\x, -0.6*\y) -- (6.5*\x, -0.6*\y) -- (6.5*\x, -0.1*\y) -- (3.5*\x, -0.1*\y) -- (3.5*\x, 0);
\draw[newblue, ->-] (3.5*\x, -0.1*\y) -- (3.5*\x, 0);
\draw[newblue, ->-] (4*\x, -0.1*\y) -- (3.5*\x, -0.1*\y);
\draw[newblue, ->-] (5*\x, -0.1*\y) -- (4*\x, -0.1*\y);
\draw[newblue, ->-] (6*\x, -0.1*\y) -- (5*\x, -0.1*\y);
\draw[newblue, ->-] (6.5*\x, -0.1*\y) -- (6*\x, -0.1*\y);
\draw[newblue, ->-] (6*\x, -0.6*\y) -- (6.5*\x, -0.6*\y);
\draw[newblue, ->-] (5*\x, -0.6*\y) -- (6*\x, -0.6*\y);
\draw[newblue, ->-] (4*\x, -0.6*\y) -- (5*\x, -0.6*\y);
\draw[newblue, ->-] (3.5*\x, -0.6*\y) -- (4*\x, -0.6*\y);
\draw[newblue, ->-] (3.5*\x, -0.7*\y) -- (3.5*\x, -0.6*\y);
\draw[fill=black] (\x, 0) circle (0.1);
\draw[fill=black] (2*\x, 0) circle (0.1);
\draw[fill=black] (2.5*\x, 0) circle (0.1);
\draw[fill=black] (3*\x, 0) circle (0.1);
\draw[fill=black] (3.5*\x, 0) circle (0.1);
\draw[fill=black] (4*\x, 0) circle (0.1);
\draw[fill=black] (5*\x, 0) circle (0.1);
\draw[fill=black] (6*\x, 0) circle (0.1);
\draw[fill=black] (7*\x, 0) circle (0.1);
\draw[fill=white] (4*\x-\d, -0.1*\y-\d) rectangle (4*\x+\d, -0.1*\y+\d);
\draw[fill=black] (4*\x-\d, -0.2*\y-\d) rectangle (4*\x+\d, -0.2*\y+\d);
\draw[fill=black] (4*\x-\d, -0.3*\y-\d) rectangle (4*\x+\d, -0.3*\y+\d);
\draw[fill=white] (4*\x-\d, -0.4*\y-\d) rectangle (4*\x+\d, -0.4*\y+\d);
\draw[fill=white] (4*\x-\d, -0.5*\y-\d) rectangle (4*\x+\d, -0.5*\y+\d);
\draw[fill=black] (4*\x-\d, -0.6*\y-\d) rectangle (4*\x+\d, -0.6*\y+\d);
\draw[fill=white] (5*\x-\d, -0.1*\y-\d) rectangle (5*\x+\d, -0.1*\y+\d);
\draw[fill=black] (5*\x-\d, -0.2*\y-\d) rectangle (5*\x+\d, -0.2*\y+\d);
\draw[fill=white] (5*\x-\d, -0.5*\y-\d) rectangle (5*\x+\d, -0.5*\y+\d);
\draw[fill=black] (5*\x-\d, -0.6*\y-\d) rectangle (5*\x+\d, -0.6*\y+\d);
\draw[fill=white] (6*\x-\d, -0.1*\y-\d) rectangle (6*\x+\d, -0.1*\y+\d);
\draw[fill=black] (6*\x-\d, -0.6*\y-\d) rectangle (6*\x+\d, -0.6*\y+\d);
\node[below] at (\x, -0.7*\y) {$\mathrm{dim}(\mathcal{C}_0)R_0$};
\node[below] at (2*\x, -0.7*\y) {$\mathrm{dim}(\mathcal{C}_0)R_0$};
\node[below] at (4*\x, -0.7*\y) {$\mathrm{dim}(\mathcal{C}_0)R_0$};
\node[below] at (5*\x, -0.7*\y) {$\mathrm{dim}(\mathcal{C}_0)R_0$};
\node[below] at (6*\x, -0.7*\y) {$\mathrm{dim}(\mathcal{C}_0)R_0$};
\node[below] at (7*\x, -0.7*\y) {$\mathrm{dim}(\mathcal{C}_0)R_0$};
\node[above left] at (2.5*\x, -0.55*\y) {$R_g$};
\node[above left] at (3*\x, -0.6*\y) {$R_h$};
\node[above left] at (3.5*\x, -0.65*\y) {$R_{gh}$};
\end{tikzpicture}
\;.
\label{eq: U defect Hilb l}
\end{equation}
Here, the horizontal edges at the top are labeled by objects of $\mathcal{C}_0$,\footnote{More precisely, the two horizontal edges between the defects are labeled by objects of $\mathcal{C}_g$ and $\mathcal{C}_{gh}$. All other horizontal edges are labeled by objects of $\mathcal{C}_0$.} and the black and white squares are the local pieces of the sequential circuits defined as in \Cref{eq: 4-valent junction mor}.
The order of defects $R_g$, $R_h$, and $R_{gh}$ is chosen so that the product of local unitaries that move the defects to the right acts in the same way as $U_{(g, h)}$.
Similarly, the right defect Hilbert space $\mathcal{H}^r_{U_{(g, h)}}$ is defined by the vector space of morphisms represented by the fusion diagrams
\begin{equation}
\begin{tikzpicture}[scale = 0.98, very thick, baseline=(current bounding box.center)]
\pgfmathsetmacro{\x}{2}
\pgfmathsetmacro{\y}{6}
\pgfmathsetmacro{\d}{0.125}
\draw[->-] (0.5*\x, 0) -- (\x, 0);
\draw[->-] (\x, 0) -- (2*\x, 0);
\draw[->-] (2*\x, 0) -- (2.5*\x, 0);
\draw[->-] (2.5*\x, 0) -- (3*\x, 0);
\draw[->-] (3*\x, 0) -- (3.5*\x, 0);
\draw[->-] (3.5*\x, 0) -- (4*\x, 0);
\draw[->-] (4*\x, 0) -- (5*\x, 0);
\draw[->-] (5*\x, 0) -- (6*\x, 0);
\draw[->-] (6*\x, 0) -- (7*\x, 0);
\draw[->-] (7*\x, 0) -- (7.5*\x, 0);
\draw[->-] (\x, -0.7*\y) -- (\x, 0);
\draw[->-] (2*\x, -0.7*\y) -- (2*\x, 0);
\draw[] (4*\x, -0.7*\y) -- (4*\x, 0);
\draw[->-] (5*\x, -0.7*\y) -- (5*\x, 0);
\draw[->-] (6*\x, -0.7*\y) -- (6*\x, 0);
\draw[->-] (7*\x, -0.7*\y) -- (7*\x, 0);
\draw[newblue] (2.5*\x, -0.7*\y) -- (2.5*\x, -0.4*\y) -- (4.5*\x, -0.4*\y) -- (4.5*\x, -0.3*\y) -- (2.5*\x, -0.3*\y) -- (2.5*\x, 0);
\draw[newblue, ->-] (2.5*\x, -0.4*\y) -- (2.5*\x, -0.7*\y);
\draw[newblue, ->-] (4*\x, -0.4*\y) -- (2.5*\x, -0.4*\y);
\draw[newblue, ->-] (4.5*\x, -0.4*\y) -- (4*\x, -0.4*\y);
\draw[newblue, ->-] (4*\x, -0.3*\y) -- (4.5*\x, -0.3*\y);
\draw[newblue, ->-] (2.5*\x, -0.3*\y) -- (4*\x, -0.3*\y);
\draw[newblue, ->-] (2.5*\x, 0) -- (2.5*\x, -0.3*\y);
\draw[newblue] (3*\x, -0.7*\y) -- (3*\x, -0.5*\y) -- (5.5*\x, -0.5*\y) -- (5.5*\x, -0.2*\y) -- (3*\x, -0.2*\y) -- (3*\x, 0);
\draw[newblue, ->-] (3*\x, -0.7*\y) -- (3*\x, -0.5*\y);
\draw[newblue, ->-] (3*\x, -0.5*\y) -- (4*\x, -0.5*\y);
\draw[newblue, ->-] (4*\x, -0.5*\y) -- (5*\x, -0.5*\y);
\draw[newblue, ->-] (5*\x, -0.5*\y) -- (5.5*\x, -0.5*\y);
\draw[newblue, ->-] (5.5*\x, -0.2*\y) -- (5*\x, -0.2*\y);
\draw[newblue, ->-] (5*\x, -0.2*\y) -- (4*\x, -0.2*\y);
\draw[newblue, ->-] (4*\x, -0.2*\y) -- (3*\x, -0.2*\y);
\draw[newblue, ->-] (3*\x, -0.2*\y) -- (3*\x, 0);
\draw[newblue] (3.5*\x, -0.7*\y) -- (3.5*\x, -0.6*\y) -- (6.5*\x, -0.6*\y) -- (6.5*\x, -0.1*\y) -- (3.5*\x, -0.1*\y) -- (3.5*\x, 0);
\draw[newblue, ->-] (3.5*\x, -0.1*\y) -- (3.5*\x, 0);
\draw[newblue, ->-] (4*\x, -0.1*\y) -- (3.5*\x, -0.1*\y);
\draw[newblue, ->-] (5*\x, -0.1*\y) -- (4*\x, -0.1*\y);
\draw[newblue, ->-] (6*\x, -0.1*\y) -- (5*\x, -0.1*\y);
\draw[newblue, ->-] (6.5*\x, -0.1*\y) -- (6*\x, -0.1*\y);
\draw[newblue, ->-] (6*\x, -0.6*\y) -- (6.5*\x, -0.6*\y);
\draw[newblue, ->-] (5*\x, -0.6*\y) -- (6*\x, -0.6*\y);
\draw[newblue, ->-] (4*\x, -0.6*\y) -- (5*\x, -0.6*\y);
\draw[newblue, ->-] (3.5*\x, -0.6*\y) -- (4*\x, -0.6*\y);
\draw[newblue, ->-] (3.5*\x, -0.7*\y) -- (3.5*\x, -0.6*\y);
\draw[fill=black] (\x, 0) circle (0.1);
\draw[fill=black] (2*\x, 0) circle (0.1);
\draw[fill=black] (2.5*\x, 0) circle (0.1);
\draw[fill=black] (3*\x, 0) circle (0.1);
\draw[fill=black] (3.5*\x, 0) circle (0.1);
\draw[fill=black] (4*\x, 0) circle (0.1);
\draw[fill=black] (5*\x, 0) circle (0.1);
\draw[fill=black] (6*\x, 0) circle (0.1);
\draw[fill=black] (7*\x, 0) circle (0.1);
\draw[fill=white] (4*\x-\d, -0.1*\y-\d) rectangle (4*\x+\d, -0.1*\y+\d);
\draw[fill=white] (4*\x-\d, -0.2*\y-\d) rectangle (4*\x+\d, -0.2*\y+\d);
\draw[fill=black] (4*\x-\d, -0.3*\y-\d) rectangle (4*\x+\d, -0.3*\y+\d);
\draw[fill=white] (4*\x-\d, -0.4*\y-\d) rectangle (4*\x+\d, -0.4*\y+\d);
\draw[fill=black] (4*\x-\d, -0.5*\y-\d) rectangle (4*\x+\d, -0.5*\y+\d);
\draw[fill=black] (4*\x-\d, -0.6*\y-\d) rectangle (4*\x+\d, -0.6*\y+\d);
\draw[fill=white] (5*\x-\d, -0.1*\y-\d) rectangle (5*\x+\d, -0.1*\y+\d);
\draw[fill=white] (5*\x-\d, -0.2*\y-\d) rectangle (5*\x+\d, -0.2*\y+\d);
\draw[fill=black] (5*\x-\d, -0.5*\y-\d) rectangle (5*\x+\d, -0.5*\y+\d);
\draw[fill=black] (5*\x-\d, -0.6*\y-\d) rectangle (5*\x+\d, -0.6*\y+\d);
\draw[fill=white] (6*\x-\d, -0.1*\y-\d) rectangle (6*\x+\d, -0.1*\y+\d);
\draw[fill=black] (6*\x-\d, -0.6*\y-\d) rectangle (6*\x+\d, -0.6*\y+\d);
\node[below] at (\x, -0.7*\y) {$\mathrm{dim}(\mathcal{C}_0)R_0$};
\node[below] at (2*\x, -0.7*\y) {$\mathrm{dim}(\mathcal{C}_0)R_0$};
\node[below] at (4*\x, -0.7*\y) {$\mathrm{dim}(\mathcal{C}_0)R_0$};
\node[below] at (5*\x, -0.7*\y) {$\mathrm{dim}(\mathcal{C}_0)R_0$};
\node[below] at (6*\x, -0.7*\y) {$\mathrm{dim}(\mathcal{C}_0)R_0$};
\node[below] at (7*\x, -0.7*\y) {$\mathrm{dim}(\mathcal{C}_0)R_0$};
\node[above left] at (2.5*\x, -0.55*\y) {$R_{gh}$};
\node[above left] at (3*\x, -0.6*\y) {$R_h$};
\node[above left] at (3.5*\x, -0.65*\y) {$R_g$};
\end{tikzpicture}
\;.
\label{eq: U defect Hilb r}
\end{equation}
By a similar computation to the case of a single defect, we find that $\mathcal{H}^l_{U_{(g, h)}}$ is isomorphic to the vector space of morphisms represented by the fusion diagrams
\begin{equation}
\begin{tikzpicture}[scale = 0.98, very thick, baseline=(current bounding box.center)]
\pgfmathsetmacro{\x}{2}
\pgfmathsetmacro{\y}{1.25}
\pgfmathsetmacro{\d}{0.125}
\draw[->-] (1.5*\x, 0) -- (2*\x, 0);
\draw[->-] (2*\x, 0) -- (3*\x, 0);
\draw[->-] (3*\x, 0) -- (4*\x, 0);
\draw[->-] (4*\x, 0) -- (5*\x, 0);
\draw[->-] (5*\x, 0) -- (5.5*\x, 0);
\draw[->-] (5.5*\x, 0) -- (6*\x, 0);
\draw[->-] (6*\x, 0) -- (7*\x, 0);
\draw[->-] (7*\x, 0) -- (7.5*\x, 0);
\draw[->-] (2*\x, -\y) -- (2*\x, 0);
\draw[->-] (3*\x, -\y) -- (3*\x, 0);
\draw[->-] (6*\x, -\y) -- (6*\x, 0);
\draw[->-] (7*\x, -\y) -- (7*\x, 0);
\draw[myorange, ->-] (4*\x, -\y) -- (4*\x, 0);
\draw[myorange, ->-] (5*\x, -\y) -- (5*\x, 0);
\draw[newblue, ->-] (5.5*\x, -\y) -- (5.5*\x, 0);
\draw[fill=black] (2*\x, 0) circle (0.1);
\draw[fill=black] (3*\x, 0) circle (0.1);
\draw[fill=black] (4*\x, 0) circle (0.1);
\draw[fill=black] (5*\x, 0) circle (0.1);
\draw[fill=black] (5.5*\x, 0) circle (0.1);
\draw[fill=black] (6*\x, 0) circle (0.1);
\draw[fill=black] (7*\x, 0) circle (0.1);
\node[below] at (2*\x, -\y) {$\mathrm{dim}(\mathcal{C}_0)R_0$};
\node[below] at (3*\x, -\y) {$\mathrm{dim}(\mathcal{C}_0)R_0$};
\node[below] at (6*\x, -\y) {$\mathrm{dim}(\mathcal{C}_0)R_0$};
\node[below] at (7*\x, -\y) {$\mathrm{dim}(\mathcal{C}_0)R_0$};
\node[below left] at (4*\x, -\y/2) {$n_g R_g$};
\node[below left] at (5*\x, -\y/2) {$n_h R_h$};
\node[below left] at (5.5*\x, -\y/2) {$R_{gh}$};
\end{tikzpicture}
\;.
\end{equation}
Here, the number of the vertical legs except for the defects is $L-2$.
Similarly, $\mathcal{H}^r_{U_{(g, h)}}$ is isomorphic to the vector space of morphisms represented by the fusion diagrams
\begin{equation}
\begin{tikzpicture}[scale = 0.98, very thick, baseline=(current bounding box.center)]
\pgfmathsetmacro{\x}{2}
\pgfmathsetmacro{\y}{1.25}
\pgfmathsetmacro{\d}{0.125}
\draw[->-] (0.5*\x, 0) -- (\x, 0);
\draw[->-] (\x, 0) -- (2*\x, 0);
\draw[->-] (2*\x, 0) -- (3*\x, 0);
\draw[->-] (3*\x, 0) -- (3.5*\x, 0);
\draw[->-] (3.5*\x, 0) -- (4*\x, 0);
\draw[->-] (4*\x, 0) -- (4.5*\x, 0);
\draw[->-] (4.5*\x, 0) -- (5.5*\x, 0);
\draw[->-] (5.5*\x, 0) -- (6.5*\x, 0);
\draw[->-] (6.5*\x, 0) -- (7*\x, 0);
\draw[->-] (\x, -\y) -- (\x, 0);
\draw[->-] (2*\x, -\y) -- (2*\x, 0);
\draw[->-] (4.5*\x, -\y) -- (4.5*\x, 0);
\draw[->-] (5.5*\x, -\y) -- (5.5*\x, 0);
\draw[->-] (6.5*\x, -\y) -- (6.5*\x, 0);
\draw[myorange, ->-] (3*\x, -\y) -- (3*\x, 0);
\draw[newblue, ->-] (3.5*\x, -\y) -- (3.5*\x, 0);
\draw[newblue, ->-] (4*\x, -\y) -- (4*\x, 0);
\draw[fill=black] (\x, 0) circle (0.1);
\draw[fill=black] (2*\x, 0) circle (0.1);
\draw[fill=black] (3*\x, 0) circle (0.1);
\draw[fill=black] (3.5*\x, 0) circle (0.1);
\draw[fill=black] (4*\x, 0) circle (0.1);
\draw[fill=black] (4.5*\x, 0) circle (0.1);
\draw[fill=black] (5.5*\x, 0) circle (0.1);
\draw[fill=black] (6.5*\x, 0) circle (0.1);
\node[below] at (\x, -\y) {$\mathrm{dim}(\mathcal{C}_0)R_0$};
\node[below] at (2*\x, -\y) {$\mathrm{dim}(\mathcal{C}_0)R_0$};
\node[below] at (4.5*\x, -\y) {$\mathrm{dim}(\mathcal{C}_0)R_0$};
\node[below] at (5.5*\x, -\y) {$\mathrm{dim}(\mathcal{C}_0)R_0$};
\node[below] at (6.5*\x, -\y) {$\mathrm{dim}(\mathcal{C}_0)R_0$};
\node[below left] at (3*\x, -\y/2) {$n_{gh} R_{gh}$};
\node[below left] at (3.5*\x, -\y/2) {$R_h$};
\node[below left] at (4*\x, -\y/2) {$R_g$};
\end{tikzpicture}
\;,
\end{equation}
where the number of the vertical legs except for the defects is $L-1$.

From the above definition of the defect Hilbert spaces, we can compute the left and right dimensions of $U_{(g, h)}$ as
\begin{align}
\ldim(U_{(g, h)}) &= \frac{\dim(n_g R_g)\dim(n_h R_h)\dim(R_{gh})}{\dim(\mathcal{C}_0)^2\dim(R_0)^2} = \frac{\dim(R_{gh})}{\dim(R_g) \dim(R_h)}, \\
\rdim(U_{(g, h)}) &= \frac{\dim(n_{gh} R_{gh})\dim(R_g)\dim(R_h)}{\dim(\mathcal{C}_0) \dim(R_0)} = \frac{\dim(R_g) \dim(R_h)}{\dim(R_{gh})}.
\end{align}
Thus, the index of $U_{(g, h)}$ is given by
\begin{equation}
\ind(U_{(g, h)}) = \frac{\dim(R_{gh})}{\dim(R_g) \dim(R_h)} = \frac{m_g m_h}{m_{gh}} \sqrt{\frac{n_g n_h}{n_{gh}}},
\label{eq: q}
\end{equation}
where $m_g \coloneq \dim(\mathcal{C}_0)^{-1}$ for all $g \in E$.
We note that this index agrees with our expectation in \Cref{eq: q expectation}.
Based on \Cref{eq: ind Dx} and \Cref{eq: q}, we can compute the product of indices as
\begin{equation}
\ind(U_{(g, h)}) \ind(\mathsf{D}_z) = \sqrt{n_g n_h} m_g m_h = \ind(\mathsf{D}_x) \ind(\mathsf{D}_y),
\end{equation}
which shows \Cref{eq: homogeneity}.

\subsection{Relation to the QCA index of Jones and Lim}\label{sec_relation_to_gen_translation}
In \cite{Jones:2023imy}, Jones and Lim defined an index of QCAs on the anyon chain using operator algebras.
As we will see below, our index of non-invertible symmetry operator $\mathsf{D}_x$ for $x \in \mathcal{C}_g$ is related to the QCA index of Jones and Lim by
\begin{equation}
\ind(\mathsf{D}_x) = \ind_{\mathrm{JL}} (U_g),
\label{eq: JL ind}
\end{equation}
where $\ind_{\mathrm{JL}}$ denotes the index defined in \cite{Jones:2023imy} and $U_g$ is the sequential circuit in \Cref{fig:sequentialc_2}.
Furthermore, since $\ind_{\mathrm{JL}}$ is multiplicative \cite{Jones:2023imy}, the composite operator $U_{(g, h)} = U_g U_h U_{gh}^{-1}$ has the index
\begin{equation}
\ind_{\mathrm{JL}}(U_{(g, h)}) = \frac{\dim(R_{gh})}{\dim(R_g) \dim(R_h)} = \ind(U_{(g, h)}),
\end{equation}
which shows that $\ind_{\mathrm{JL}}$ agrees with our index when evaluated on a QCA $U_{(g, h)}$.

In the remaining of this subsection, we show \Cref{eq: JL ind}.
To this end, we first show that $U_g$ is an example of a generalized translation defined in \cite{Jones:2023imy}.
We can then apply the index formula in \cite[Proposition 4.1]{Jones:2023imy} to compute the index of $U_g$.

\vspace*{\baselineskip}
\noindent{\bf Generalized translation.}
Let us first recall the definition of generalized translations.
We consider an anyon chain based on a fusion category $\mathcal{C}$ and an input object $\rho \in \mathcal{C}_0$.\footnote{We do not specify the choice of a $\mathcal{C}$-module category because it is not relevant to the following discussions.}
Here, $\mathcal{C}_0$ is a fusion subcategory of $\mathcal{C}$.
For our purposes, it suffices to take $\mathcal{C}$ to be a weakly integral fusion category and $\mathcal{C}_0$ to be its integral subcategory.
We suppose that there exists a pair of objects $y, z \in \mathcal{C}$ such that
\begin{equation}
\rho \cong y \otimes z \cong z \otimes y.
\end{equation}
We also choose an isomorphism
\begin{equation}
\sigma: y \otimes z \to z \otimes y.
\end{equation}
Using the above data, we can define a generalized translation $\alpha$ as a QCA on the algebra of local operators.
To define $\alpha$, we consider its action on local operators supported on a finite interval $I = [i_0, i_1]$ of length $m$.
Here, $i_0$ and $i_1$ are the left and right ends of $I$.
The generalized translation $\alpha$ maps local operators on $I$ to those supported on $I^{+1} = [i_0 -1, i_1 +1]$.
The action of $\alpha$ on a local operator $w \in \End(\rho^{\otimes m})$ supported on $I$ is defined by \cite{Jones:2023imy}
\begin{equation}
\alpha(w) = \id_{\rho} \otimes \id_y \otimes (\sigma^{\otimes m} \circ w \circ (\sigma^{-1})^{\otimes m}) \otimes \id_z.
\end{equation}
Here, as in the original paper \cite{Jones:2023imy}, we implicitly chose an isomorphism between $\rho$ and $y \otimes z$.
If we make this choice explicit, the above expression becomes
\begin{equation}
\alpha(w) = (\id_{\rho} \otimes (s^{-1})^{\otimes m+1}) \circ (\id_{\rho} \otimes \id_y \otimes (\sigma^{\otimes m} \circ s(w) \circ (\sigma^{-1})^{\otimes m}) \otimes \id_z) \circ (\id_{\rho} \otimes s^{\otimes m+1}),
\end{equation}
where $s: \rho \to y \otimes z$ is an isomorphism and $s(w) \coloneq s^{\otimes m} \circ w \circ (s^{-1})^{\otimes m}$.
One can also write the above $\alpha(w)$ as
\begin{equation}
\alpha(w) = (\id_{\rho} \otimes (s^{-1})^{\otimes m+1}) \circ (\id_{\rho} \otimes \id_y \otimes (t^{\otimes m} \circ w \circ (t^{-1})^{\otimes m}) \otimes \id_z) \circ (\id_{\rho} \otimes s^{\otimes m+1}),
\end{equation}
where $t: \rho \to z \otimes y$ is an isomorphism defined by $t \coloneq \sigma \circ s$.

To see the relation between generalized translation $\alpha$ and our sequential circuit $U_g$, it is helpful to represent $\alpha(w)$ diagrammatically as follows:
\begin{equation}
\alpha(w) =
\begin{tikzpicture}[scale = 1, very thick, baseline=(current bounding box.center)]
\pgfmathsetmacro{\x}{2}
\pgfmathsetmacro{\y}{0.6}
\pgfmathsetmacro{\d}{0.15}
\pgfmathsetmacro{\r}{0.1}
\draw[->-] (0, 0) -- (0, 7*\y);
\draw (\x, 0) -- (\x, \y);
\draw[newblue] (\x, \y) -- (\x, 2*\y);
\draw (\x, 2*\y) -- (\x, 5*\y);
\draw[newblue] (\x, 5*\y) -- (\x, 6*\y);
\draw (\x, 6*\y) -- (\x, 7*\y);
\draw (2*\x, 0) -- (2*\x, \y);
\draw[newblue] (2*\x, \y) -- (2*\x, 2*\y);
\draw (2*\x, 2*\y) -- (2*\x, 5*\y);
\draw[newblue] (2*\x, 5*\y) -- (2*\x, 6*\y);
\draw (2*\x, 6*\y) -- (2*\x, 7*\y);
\draw (3*\x, 0) -- (3*\x, \y);
\draw[newblue] (3*\x, \y) -- (3*\x, 2*\y);
\draw (3*\x, 2*\y) -- (3*\x, 5*\y);
\draw[newblue] (3*\x, 5*\y) -- (3*\x, 6*\y);
\draw (3*\x, 6*\y) -- (3*\x, 7*\y);
\draw (4*\x, 0) -- (4*\x, \y);
\draw[newblue, ->-] (4*\x, \y) -- (4*\x, 6*\y);
\draw (4*\x, 6*\y) -- (4*\x, 7*\y);
\draw[myorange] (\x, \y) -- (0.5*\x, 1.5*\y) -- (0.5*\x, 5.5*\y) -- (\x, 6*\y);
\draw[myorange, ->-] (\x, \y) -- (0.5*\x, 1.5*\y); 
\draw[myorange, ->-] (0.5*\x, 5.5*\y) -- (\x, 6*\y); 
\draw[myorange, ->-] (2*\x, \y) -- (\x, 2*\y);
\draw[myorange, ->-] (3*\x, \y) -- (2*\x, 2*\y);
\draw[myorange, ->-] (4*\x, \y) -- (3*\x, 2*\y);
\draw[myorange, ->-] (\x, 5*\y) -- (2*\x, 6*\y);
\draw[myorange, ->-] (2*\x, 5*\y) -- (3*\x, 6*\y);
\draw[myorange, ->-] (3*\x, 5*\y) -- (4*\x, 6*\y);
\draw[fill = red] (\x, \y) circle (\r);
\draw[fill = red] (2*\x, \y) circle (\r);
\draw[fill = red] (3*\x, \y) circle (\r);
\draw[fill = red] (4*\x, \y) circle (\r);
\draw[fill = blue] (\x, 2*\y) circle (\r);
\draw[fill = blue] (2*\x, 2*\y) circle (\r);
\draw[fill = blue] (3*\x, 2*\y) circle (\r);
\draw[fill = blue] (\x, 5*\y) circle (\r);
\draw[fill = blue] (2*\x, 5*\y) circle (\r);
\draw[fill = blue] (3*\x, 5*\y) circle (\r);
\draw[fill = red] (\x, 6*\y) circle (\r);
\draw[fill = red] (2*\x, 6*\y) circle (\r);
\draw[fill = red] (3*\x, 6*\y) circle (\r);
\draw[fill = red] (4*\x, 6*\y) circle (\r);
\draw[fill=white] (\x-\d, 3*\y) rectangle (3*\x+\d, 4*\y);
\node[below right] at (\x+\r/2, \y) {$s$};
\node[below right] at (2*\x+\r/2, \y) {$s$};
\node[below right] at (3*\x+\r/2, \y) {$s$};
\node[below right] at (4*\x+\r/2, \y) {$s$};
\node[above left] at (\x-\r/2, 2*\y-0.1) {$t^{-1}$};
\node[above left] at (2*\x-\r/2, 2*\y-0.1) {$t^{-1}$};
\node[above left] at (3*\x-\r/2, 2*\y-0.1) {$t^{-1}$};
\node[below left] at (\x-\r/2, 5*\y) {$t$};
\node[below left] at (2*\x-\r/2, 5*\y) {$t$};
\node[below left] at (3*\x-\r/2, 5*\y) {$t$};
\node[above right] at (\x+\r/2, 6*\y) {$s^{-1}$};
\node[above right] at (2*\x+\r/2, 6*\y) {$s^{-1}$};
\node[above right] at (3*\x+\r/2, 6*\y) {$s^{-1}$};
\node[above right] at (4*\x+\r/2, 6*\y) {$s^{-1}$};
\node[below] at (0, 0) {$\rho$};
\node[below] at (\x, 0) {$\rho$};
\node[below] at (2*\x, 0) {$\rho$};
\node[below] at (3*\x, 0) {$\rho$};
\node[below] at (4*\x, 0) {$\rho$};
\node[above] at (0, 7*\y) {$\rho$};
\node[above] at (\x, 7*\y) {$\rho$};
\node[above] at (2*\x, 7*\y) {$\rho$};
\node[above] at (3*\x, 7*\y) {$\rho$};
\node[above] at (4*\x, 7*\y) {$\rho$};
\node[below] at (0.5*\x, 1.5*\y-0.1) {$y$};
\node[below] at (1.5*\x, 1.5*\y-0.1) {$y$};
\node[below] at (2.5*\x, 1.5*\y-0.1) {$y$};
\node[below] at (3.5*\x, 1.5*\y-0.1) {$y$};
\node[above] at (0.5*\x, 5.5*\y+0.1) {$y$};
\node[above] at (1.5*\x, 5.5*\y+0.1) {$y$};
\node[above] at (2.5*\x, 5.5*\y+0.1) {$y$};
\node[above] at (3.5*\x, 5.5*\y+0.1) {$y$};
\node[right] at (\x, 1.5*\y) {$z$};
\node[right] at (2*\x, 1.5*\y) {$z$};
\node[right] at (3*\x, 1.5*\y) {$z$};
\node[left] at (\x, 5.5*\y) {$z$};
\node[left] at (2*\x, 5.5*\y) {$z$};
\node[left] at (3*\x, 5.5*\y) {$z$};
\node[left] at (4*\x, 3.5*\y) {$z$};
\node at (2*\x, 3.5*\y) {$w$};
\end{tikzpicture}
\;.
\label{eq: alpha w}
\end{equation}
Here, we chose $m = 3$ for concreteness.
The above diagram implies
\begin{equation}
\alpha(w) = U \circ [\id \otimes w \otimes \id] \circ U^{-1},
\label{eq: unitary alpha}
\end{equation}
where $U$ is the following unitary operator acting on the entire chain:
\begin{equation}
U \coloneq
\begin{tikzpicture}[scale = 1, very thick, baseline=(current bounding box.center)]
\pgfmathsetmacro{\x}{2}
\pgfmathsetmacro{\y}{1}
\pgfmathsetmacro{\d}{0.15}
\pgfmathsetmacro{\r}{0.1}
\draw[->-] (0, 0) -- (0, \y);
\draw[newblue, ->-] (0, \y) -- (0, 2*\y);
\draw[->-] (0, 2*\y) -- (0, 3*\y);
\draw[->-] (\x, 0) -- (\x, \y);
\draw[newblue, ->-] (\x, \y) -- (\x, 2*\y);
\draw[->-] (\x, 2*\y) -- (\x, 3*\y);
\draw[->-] (2*\x, 0) -- (2*\x, \y);
\draw[newblue, ->-] (2*\x, \y) -- (2*\x, 2*\y);
\draw[->-] (2*\x, 2*\y) -- (2*\x, 3*\y);
\draw[->-] (3*\x, 0) -- (3*\x, \y);
\draw[newblue, ->-] (3*\x, \y) -- (3*\x, 2*\y);
\draw[->-] (3*\x, 2*\y) -- (3*\x, 3*\y);
\draw[->-] (4*\x, 0) -- (4*\x, \y);
\draw[newblue, ->-] (4*\x, \y) -- (4*\x, 2*\y);
\draw[->-] (4*\x, 2*\y) -- (4*\x, 3*\y);
\draw[myorange, ->-] (0, \y) -- (-\x, 2*\y); 
\draw[myorange, ->-] (\x, \y) -- (0, 2*\y); 
\draw[myorange, ->-] (2*\x, \y) -- (\x, 2*\y);
\draw[myorange, ->-] (3*\x, \y) -- (2*\x, 2*\y);
\draw[myorange, ->-] (4*\x, \y) -- (3*\x, 2*\y);
\draw[myorange, ->-] (5*\x, \y) -- (4*\x, 2*\y); 
\draw[fill = red] (0, \y) circle (\r);
\draw[fill = red] (\x, \y) circle (\r);
\draw[fill = red] (2*\x, \y) circle (\r);
\draw[fill = red] (3*\x, \y) circle (\r);
\draw[fill = red] (4*\x, \y) circle (\r);
\draw[fill = blue] (0, 2*\y) circle (\r);
\draw[fill = blue] (\x, 2*\y) circle (\r);
\draw[fill = blue] (2*\x, 2*\y) circle (\r);
\draw[fill = blue] (3*\x, 2*\y) circle (\r);
\draw[fill = blue] (4*\x, 2*\y) circle (\r);
\node[below right] at (0+\r/2, \y) {$s$};
\node[below right] at (\x+\r/2, \y) {$s$};
\node[below right] at (2*\x+\r/2, \y) {$s$};
\node[below right] at (3*\x+\r/2, \y) {$s$};
\node[below right] at (4*\x+\r/2, \y) {$s$};
\node[above left] at (0-\r/2, 2*\y-0.1) {$t^{-1}$};
\node[above left] at (\x-\r/2, 2*\y-0.1) {$t^{-1}$};
\node[above left] at (2*\x-\r/2, 2*\y-0.1) {$t^{-1}$};
\node[above left] at (3*\x-\r/2, 2*\y-0.1) {$t^{-1}$};
\node[above left] at (4*\x-\r/2, 2*\y-0.1) {$t^{-1}$};
\node[below] at (0, 0) {$\rho$};
\node[below] at (\x, 0) {$\rho$};
\node[below] at (2*\x, 0) {$\rho$};
\node[below] at (3*\x, 0) {$\rho$};
\node[below] at (4*\x, 0) {$\rho$};
\node[above] at (0, 3*\y) {$\rho$};
\node[above] at (\x, 3*\y) {$\rho$};
\node[above] at (2*\x, 3*\y) {$\rho$};
\node[above] at (3*\x, 3*\y) {$\rho$};
\node[above] at (4*\x, 3*\y) {$\rho$};
\node[below] at (-0.5*\x, 1.5*\y-0.1) {$y$};
\node[below] at (0.5*\x, 1.5*\y-0.1) {$y$};
\node[below] at (1.5*\x, 1.5*\y-0.1) {$y$};
\node[below] at (2.5*\x, 1.5*\y-0.1) {$y$};
\node[below] at (3.5*\x, 1.5*\y-0.1) {$y$};
\node[below] at (4.5*\x, 1.5*\y-0.1) {$y$};
\node[right] at (0, 1.5*\y) {$z$};
\node[right] at (\x, 1.5*\y) {$z$};
\node[right] at (2*\x, 1.5*\y) {$z$};
\node[right] at (3*\x, 1.5*\y) {$z$};
\node[right] at (4*\x, 1.5*\y) {$z$};
\end{tikzpicture}
\;.
\end{equation}
We note that $U^{-1}$ acts first and $U$ acts last in \Cref{eq: unitary alpha} because $\alpha(w)$ in \Cref{eq: alpha w} acts on the chain from below.
The unitary $U$ defined above is related to our sequential circuit $U_g$ as
\begin{equation}
U = T_+ U_g,
\label{eq: generalized translation unitary}
\end{equation}
where $T_+$ is a lattice translation by one site to the right, and we identified
\begin{equation}
\rho = \dim(\mathcal{C}_0) R_0, \quad
y = n_g R_g, \quad z = R_g, \quad
s  = \gamma_g, \quad t = \delta_g^{\dagger}.
\end{equation}
Equation~\eqref{eq: generalized translation unitary} shows that the sequential circuit $U_g$ is an example of a generalized translation.

\vspace*{\baselineskip}
\noindent{\bf Index of $U_g$.}
Now, we can compute $\ind_{\mathrm{JL}}(U_g)$ by applying the index formula in \cite[Proposition 4.1]{Jones:2023imy}.
Since $\ind_{\mathrm{JL}}$ is multiplicative \cite{Jones:2023imy}, \Cref{eq: generalized translation unitary} implies
\begin{equation}
\ind_{\mathrm{JL}}(U_g) = \frac{\ind_{\mathrm{JL}}(U)}{\ind_{\mathrm{JL}}(T_+)}.
\label{eq: ind JL Ug}
\end{equation}
Furtehrmore, due to \cite[Proposition 4.1]{Jones:2023imy}, the indices of $U$ and $T_+$ are given by
\begin{equation}
\ind_{\mathrm{JL}}(U) = \dim(n_g R_g), \qquad
\ind_{\mathrm{JL}}(T_+) = \dim(\rho) = \dim(n_g R_g) \dim(R_g).
\end{equation}
By plugging these into \Cref{eq: ind JL Ug}, we find
\begin{equation}
\ind_{\mathrm{JL}}(U_g) = \dim(R_g)^{-1}.
\end{equation}
In particular, due to \Cref{eq: ind Dx}, we have $\ind_{\mathrm{JL}}(U_g) = \ind(\mathsf{D}_x)$ for all $x \in \mathcal{C}_g$, which shows \Cref{eq: JL ind}.

\section{Tensor Product Realization of  General Tambara-Yamagami Categories}
\label{sec_eg_TY}

We now apply our general construction in \Cref{sec_lattice_model_weakly} to the class of general Tambara-Yamagami fusion categories. We calculate the explicit expressions of symmetry operators and the QCA refinement following \Cref{construction_special_case}. Tambara-Yamagami categories of the form $\TY(\Z_N)$ have been realized on a tensor-product Hilbert space~\cite{Cao:2024qjj}. Here we are not restricted to such cases and consider general Tambara-Yamagami categories. We start by reviewing the basic structures of general Tambara-Yamagami fusion categories.
\subsection{Review of Tambara-Yamagami Categories}

A Tambara--Yamagami category 
\cite{Tambara:1998vmj} is a $\Z_2$-graded fusion category $\C=\C_0\oplus \C_1$ with $\C_0=\vc_A,~\C_1=\vc$, where $A$ is a finite abelian group. We denote the simple objects in $\vc_A$ as $a\in A$ and the simple object in the nontrivial component as $m$. Then the fusion rules is  
\be 
a\otimes b=ab,~ a\otimes m=m\otimes a=m,~m\otimes m=\oplus_{a\in A} a.
\ee

A Tambara--Yamagami category of the above form is parameterized by (1). A symmetric non-degenerate bicharacter $\chi: A\times A\to U(1)$. (2). A Frobenius-Schur indicator $\epsilon=\pm$. We denote the corresponding fusion category as $\TY(A,\chi, \epsilon)$. The quantum dimensions of simples are
\begin{equation}
  \dim(a) = 1 \quad \forall\, a \in A,
  \qquad
  \dim(m) = \sqrt{|A|}.
\end{equation}

The global quantum dimension of $\TY(A,\chi,\epsilon)$ is thus
\begin{equation}
\dim\bigl(\TY(A,\chi,\epsilon)\bigr)
= \sum_{a \in A}\dim(a)^2 + \dim(m)^2
= 2|A|,
\end{equation}
which shows the fusion category $\TY(A,\chi,\epsilon)$ is weakly integral.

\paragraph{$F$-Symbols.}
The associativity constraints are largely trivial. The complete list of non-trivial $F$-symbols was obtained in \cite{Tambara:1998vmj}. The ones relevant for us are:
\begin{equation}\label{eq:F_amm}
  F^{mab}_m = 1,
  \qquad
  F^{amm}_b = 1,
\end{equation}
and 
\begin{equation}\label{eq_F_mmmm}
    (F^{mmm}_m)_{ab}=\frac{\epsilon}{\sqrt{|A|}} \chi(a,b).
\end{equation}

\begin{example}
    Take $A=\Z_N=\{0,1,\cdots, N-1\}$, then symmetric bicharacters are classified by $\Z_N$, for $k\in \Z_N$ the corresponding bicharacter $\chi_k$ takes the form 
    \begin{equation}
        \chi_k(a,b)=\exp\left(\frac{2\pi i k}{N} ab\right).
    \end{equation}
    This bicharacter is non-degenerate precisely when $\mathrm{gcd}(k,N)=1$. The $F$-symbol 
    \begin{equation}
        (F^{mmm}_m)_{ab}=\frac{\epsilon}{\sqrt{N}} \exp\left(\frac{2\pi i k}{N} ab\right)
    \end{equation}
    is nothing but the normalized discrete Fourier transformation, up to the sign $\epsilon$.
\end{example}

\subsection{Tensor-product Hilbert Space Realization}
We now provide an explicit QCA-refined tensor-product realization of a general Tambara-Yamagami category. If $|A|$ is a perfect square then  $\TY(A,\chi,\epsilon)$  is integral and admits a strict tensor-product realization. Hence we only consider the case where $|A|$ is not a perfect square.  Then the integral subcategory $\Vec_A$ has total quantum dimension $|A|$, and we will realize the symmetry $\TY(A, \chi, \epsilon)$ on a tensor-product Hilbert space with on-site dimension $|A|$.

For a generic abelian group $A$ we can not apply our \Cref{construction_special_case}: $a_m$ is the integer part of $d_m=\sqrt{|A|}$, which is not 1 unless $|A|$ is square-free. We will instead consider the ``flexible construction" in \Cref{rmk_flexible}. The construction will be a natural generalization of the Ising category example in \Cref{sec_ising_example}. For example, the sequential circuit $U_1$ will be similar to \Cref{fig_sequential_Ising} with $\sigma$ replaced by $m$. We start with the anyon chain Hilbert space  $\cH= \cH_0 \oplus \cH_1$ built from choosing 
\be
 \rho=R_0 =\bigoplus_{a \in A} a.
\ee
The subspace $\H_0$ has a basis similar to \Cref{fig_anyon_chain_Ising_triv} but with $x_i\in A$, and $\H_1$ has a basis similar to \Cref{fig_anyon_chain_Ising_nontriv}  with $ t_i \in A$. The basis vectors are denoted as 
\begin{equation}
    |\{x_i\}\>_0\in \H_0, \quad |\{t_i\}\>_1\in \H_1.
\end{equation}

The integral subcategory $\vc_A$ is represented on $\H_0$ directly by anyon fusion action as in \Cref{eq_anyon_chain_action}. We have 
\begin{align}
    \sf{D}_a|\{x_i\}\>_0=|\{ax_i\}\>_0.
\end{align}

The symmetry action for the noninvertible anyon $m$ is $\sf{D}_m=U_1\circ \cc{D}_m$. The  operator $\cc{D}_m: \H_0\to \H_1$ may be calculated in a similar fashion as in \Cref{fig:SymActionIsing}, where the relevant $F$-symbol is  $F^{mab}_m=1$, and maps $\cH_0 \to \cH_1$. We arrive at 
\begin{equation}
    \cc{D}_m|\cdots,x_i, x_{i+1}, x_{i+2},\cdots\>_0=|\cdots, x_i^{-1}x_{i+1},x_{i+1}^{-1}x_{i+2},x_{i+2}^{-1}x_{i+3},\cdots\>_1.
\end{equation}

The sequential circuit $U_1$ is defined similar to \Cref{fig_sequential_Ising} with $\sigma$ replaced by $m$. The action of the sequential circuit $U_1$ can be calculated similar to that in  \Cref{eq_ising_circuit_1}. The relevant $F$-symbols for this calculation are $F^{mmm}_{m}$ (\Cref{eq_F_mmmm}) and $F^{amm}_b=1$. The action is thus
\begin{equation}
    U_1|\{t_i\}\>_1=\left(\frac{\epsilon}{\sqrt{|A|}}\right)^L\sum_{\{x_i\}} \prod_i \chi(t_i,x_i)|\{x_i\}\>_0.
\end{equation}

The combined action is  therefore
\begin{equation}
   \boxed{ \sf{D}_m|\{x_i\}\>_0=U_1\circ \cc{D}_m|\{x_i\}\>_0=\left(\frac{\epsilon}{\sqrt{|A|}}\right)^L\sum_{\{x_i'\}} \prod_{i=1}^L \chi(x_i^{-1}x_{i+1},x_i')|\{x_i'\}\>_0.
   }
\end{equation}

We now compute $\sf{D}_m^2$ directly.
\begin{align}
    \sf{D}_m^2|\{x_i\}\>_0&=\left(\frac{1}{|A|}\right)^L \sum_{\{x_i',x_i''\}}\prod_i \chi(x_i^{-1}x_{i+1},x_i')\chi(x_i'^{-1}x_{i+1}',x_i'')|\{x_i''\}\>_0\\
    &=\left(\frac{1}{|A|}\right)^L \sum_{\{x_i',x_i''\}} \prod_i \chi(x_i^{-1}x_{i+1},x_i') \chi(x_i',x_i'')^{-1}\chi(x_i',x_{i-1}'')|\{x_i''\}\>_0.\\
    &=\left(\frac{1}{|A|}\right)^L \sum_{\{x_i',x_i''\}}\prod_{i}\chi(x_i^{-1}x_{i+1}x_i''^{-1}x_{i-1}'',x_i') |\{x_i''\}\>_0.
    \\
    &=\sum_{\{x_i''\}} \prod_i \delta_{x_i^{-1}x_{i+1}x_i''^{-1}x_{i-1}'',1} |\{x_i''\}\>_0, 
\end{align}
where in the last step we used $\sum_{a\in A}\chi(a,b)=|A|\delta_{b,1}$. Now observe that the delta symbols impose 
\begin{equation}
    x_i^{-1}x_{i-1}''=x_{i+1}^{-1}x_i''.
\end{equation}
Hence $y:=x_i^{-1}x_{i-1}''$ is independent of $i$. Therefore we have 
\begin{align}
     \sf{D}_m^2|\{x_i\}\>_0&=\sum_{y \in A} \sum_{\{x_i''\}} \prod_i \delta_{x_{i-1}'',x_iy}|\{x_i''\}\>_0\\
     &=\sum_{y\in A} |yx_2, yx_3,\cdots, yx_L,yx_1\>_0=T\sum_{y\in A} \sf{D}_y |\{x_i\}\>_0,
\end{align}
where $T$ is translation by one site to the left. We conclude 
\begin{equation}
 \sf{D}_m^2=T\sum_{y\in A}\sf{D}_y.\label{eq_TY_refined_fusion}
\end{equation}

Here $\ind(T)=|A|^{-1}$. We can obtain a canonical QCA-refined realization by redefining $\widetilde{\sf{D}}_m=\sf{D}_m T^{-1}$. We note that the QCA-refined fusion rules \Cref{eq_TY_refined_fusion}  can also be derived diagrammatically as in \Cref{fig_Ug_two_layers}.

\section*{Acknowledgments}
We thank Thomas Bartsch for discussions. 
RW and SSN are supported by the UKRI Frontier Research Grant, underwriting the ERC Advanced Grant ``Generalized Symmetries in Quantum Field Theory and Quantum Gravity”. 
KI and SSN are supported in part by the EPSRC Open Fellowship EP/X01276X/1 (Schafer-Nameki).
KI is also supported by the Leverhulme-Peierls Fellowship funded by the Leverhulme Trust.

\appendix 

\section{Conventions}
\label{app:Conventions}

All fusion categories are defined over $\bC$ and are assumed to be unitary  \cite{egno2015tensor}. For a fusion category $\C$, we write $x\in \C$ to mean $x$ is an object of $\C$ and $x\in \irr(\C)$ to mean $x$ is a simple object of $\C$. For $x\in \irr(\C)$ and $n\in \Z_{\ge0}$ we will write $nx$ to mean $n$-fold direct sum of $x$.  The quantum dimension of an object $x\in \C$ is denoted as $d_x$. For two objects $x, y \in \C$, we will sometimes abuse notation and write $x=y$ to mean that $x$ is isomorphic to $y$ with an isomorphism chosen. For a general object $X=\oplus_{a\in \irr(\C)}n_a a$ there are  projection and inclusion morphisms
\be
p_a^\mu:X\to a,~\ i_a^\mu: a\to X,~\ \mu=1,\cdots, n_a,
\ee
corresponding to the $\mu$-th copy of $a$ in $X$. The unitary structure on the fusion category is defined as $i^\mu_a=(p_a^\mu)^\dagger$, and the inclusions/projections are taken to be orthonormal: $\<i_a^\mu, i_a^\nu\>=\<p_a^\mu,p_a^\nu\>=\delta_{\mu,\nu}$. 
The concrete examples of fusion categories that we will consider in this paper are all multiplicity-free. We use projection/inclusion basis for fusion/splitting spaces: For simples $a,b,c \in \irr(\C)$, trivalent vertices below are understood as the unique projection/inclusion $p_c: a\otimes b\to c$ and $i_c: c\to a\otimes b$.
\begin{equation}
    \begin{split}
        \begin{tikzpicture}
            \begin{scope}[scale=0.6]
                \draw[very thick, ->-] (0,0) -- (1,1);
                \draw[very thick, ->-] (2,0) -- (1,1);
                \draw[very thick, ->-] (1,1) -- (1,2.414);
                \node[below left] at (0,0) {$a$};
                \node[below right] at (2,0) {$b$};
                \node[above] at (1,2.414) {$c$};
            \end{scope}
        \end{tikzpicture}
    \end{split},\qquad
    \begin{split}
        \begin{tikzpicture}
            \begin{scope}[scale=0.6]
                \draw[very thick, ->-] (1,0) -- (1,1.414);
                \draw[very thick, ->-] (1,1.414) -- (0,2.414);
                \draw[very thick, ->-] (1,1.414) -- (2,2.414);
                \node[above left] at (0,2.414) {$a$};
                \node[above right] at (2,2.414) {$b$};
                \node[below] at (1,0) {$c$};
            \end{scope}
        \end{tikzpicture}
    \end{split}
\end{equation}
Other types of morphisms will be defined explicitly and are represented as dots or boxes. The $F$-symbols follow the following convention:

\begin{equation}
    \begin{split}
        \begin{tikzpicture}
            \begin{scope}[scale=0.6]
                \draw[very thick, ->-] (0,0) -- (2,2);
                \draw[very thick, ->-] (2,2) -- (3,3);
                \draw[very thick, ->-] (2,0) -- (3,1);
                \draw[very thick, ->-] (3,1) -- (2,2);
                \draw[very thick, ->-] (4,0) -- (3,1);
                \node[below] at (0,0) {$a$};
                \node[below] at (2,0) {$b$};
                \node[below] at (4,0) {$c$};
                \node[above] at (3,1) {$e$};
                \node[above] at (3,3) {$d$};
            \end{scope}
        \end{tikzpicture}
    \end{split}
    =\sum_f (F^{abc}_d)_{ef} \begin{split}
        \begin{tikzpicture}
            \begin{scope}[scale=0.6]
                \draw[very thick, ->-] (0,0) -- (1,1);
                \draw[very thick, ->-] (1,1) -- (2,2);
                \draw[very thick, ->-] (2,2) -- (3,3);
                \draw[very thick, ->-] (2,0) -- (1,1);
                \draw[very thick, ->-] (4,0) -- (2,2);
                \node[below] at (0,0) {$a$};
                \node[below] at (2,0) {$b$};
                \node[below] at (4,0) {$c$};
                \node[above] at (1,1) {$f$};
                \node[above] at (3,3) {$d$};
            \end{scope}
        \end{tikzpicture}
    \end{split}
\end{equation}


\bibliographystyle{ytphys}
\small 
\baselineskip=.94\baselineskip
\let\bbb\bibitem\def\bibitem{\itemsep4pt\bbb}
\bibliography{QCA}

\end{document}